\documentclass[onefignum,onetabnum,,nohypdvips]{siamonline171218}

\usepackage{booktabs}
\usepackage[normalem]{ulem}
\usepackage{lipsum}
\usepackage{amsfonts}
\usepackage{graphicx}
\usepackage{epstopdf}
\usepackage{algorithmic}
\ifpdf
  \DeclareGraphicsExtensions{.eps,.pdf,.png,.jpg}
\else
  \DeclareGraphicsExtensions{.eps}
\fi

\usepackage{enumitem}
\setlist[enumerate]{leftmargin=.5in}
\setlist[itemize]{leftmargin=.5in}

\newsiamremark{remark}{Remark}
\newsiamremark{hypothesis}{Hypothesis}
\crefname{hypothesis}{Hypothesis}{Hypotheses}
\newsiamthm{claim}{Claim}

\headers{Optimal payoff under BW divergence}{Pesenti, Vanduffel, Yang, Yao}

\title{Optimal payoff under Bregman-Wasserstein divergence constraints\thanks{This version: {May 16, 2026};  Earlier versions: October 28, 2025; November 21, 2024.
}}

\author{Silvana M. Pesenti\thanks{Department of Statistical Sciences, University of Toronto, Canada
  (\email{silvana.pesenti@utoronto.ca}).}
\and Steven Vanduffel\thanks{Department of Economics and Political
Science, Vrije Universiteit Brussel, Belgium (\email{Steven.Vanduffel@vub.be}).}
\and Yang Yang\thanks{Center for Financial Engineering, Soochow University, China (\email{yangy\_1217@126.com}, \email{j.yao@suda.edu.cn}).}
\and Jing Yao\footnotemark[4]
}

\usepackage{amsopn}

\makeatletter
\newcommand*{\addFileDependency}[1]{% argument=file name and extension
  \typeout{(#1)}% latexmk will find this if $recorder=0 (however, in that case, it will ignore #1 if it is a .aux or .pdf file etc and it exists! if it doesn't exist, it will appear in the list of dependents regardless)
  \@addtofilelist{#1}% if you want it to appear in \listfiles, not really necessary and latexmk doesn't use this
  \IfFileExists{#1}{}{\typeout{No file #1.}}% latexmk will find this message if #1 doesn't exist (yet)
}
\makeatother

\usepackage{todonotes}
\usepackage{dsfont}
\usepackage{natbib,amsfonts,amsmath}
\usepackage{mathrsfs}
\usepackage{subcaption}
\usepackage{graphicx}

\newtheorem{assumption}{Assumption}[section]

\DeclareMathOperator*{\argmin}{arg\,min}

\newcommand{\cost}{{\mathfrak{c}}}

\newcommand{\ep}{{\varepsilon}}
\renewcommand{\d}{{\mathrm{d}}}
\renewcommand{\epsilon}{{\varepsilon}}

\DeclareMathOperator*{\coloneqq}{:=}

\newcommand{\Finv}{{\Breve{F}}}
\newcommand{\Ginv}{{\Breve{G}}}

\newcommand{\Fbenchinv}{{\Breve{F}_b}}

\newcommand{\DB}{{\mathscr{B}}}

\newcommand{\E}{{\mathbb{E}}}
\renewcommand{\P}{{\mathbb{P}}}
\newcommand{\Q}{{\mathbb{Q}}}

\newcommand{\mQ}{{\mathcal{Q}}}

\ifpdf
\hypersetup{
  pdftitle={Optimal payoff under Bregman-Wasserstein divergence constraints},
  pdfauthor={Pesenti, Vanduffel, Yang, and Yao}
}
\fi

\begin{document}

\maketitle

% REQUIRED
\begin{abstract}
  We study optimal payoff choice for an expected utility maximizer under the constraint that their payoff is not allowed to deviate ``too much'' from a given benchmark. We solve this problem when the deviation is assessed via a Bregman-Wasserstein (BW) divergence, generated by a convex function $\phi$. Unlike the Wasserstein distance (i.e., when $\phi(x)=x^2$) the inherent asymmetry of the BW divergence makes it possible to penalize positive deviations different than negative ones. As a main contribution, we provide the optimal payoff in this setting. Numerical examples illustrate that the choice of $\phi$ allows us to better align the payoff choice with the objectives of investors. 
\end{abstract}

% REQUIRED

% \textbf{MSC} 91B28, 91B30, 91B02.
% EUT
% REQUIRED
% \begin{AMS}
% \end{AMS}

\section{Introduction}
A standard assumption in the literature on optimal payoff choice is that investors' preferences can be described by the expected utility theory (EUT) of \cite{vonNeumann_Morgenstern_47}. This approach has however faced significant criticisms and researchers have explored optimal payoff choices under alternative decision theories. Notable are the dual theory of \cite{Yaari_87}, the rank-dependent expected utility (RDEUT) approach of \cite{Quiggin_93}, and several behavioural approaches, including the SP/A theory of \cite{Lopes_87} and \cite{Shefrin_Statman_00}, and the cumulative prospect theory (CPT) of \cite{Tversky_Kahneman_92}.
Despite extensive research within these alternative frameworks\footnote{Optimal payoffs have been derived for an RDEUT investor in \cite{carlier2011optimal}, \cite{xia2016arrow}, \cite{he2016hope} \cite{xu2016note}, \cite{jaimungal2022SIAM}; for a CPT investor in \cite{jin2008behavioral}, \cite{zhang2011behavioral}, and \cite{ruschendorf2017}; and for a Yaari investor in \cite{He_Zhou_11},  \cite{he2021optimal} and \cite{ boudt2022optimal}.}, none prove superior to EUT for the problem of optimal portfolio choice. For instance, under Yaari's dual theory, optimal payoffs often materialize as binary (digital) or trinary options, which are infrequently practical since investors rarely opt for such payoffs. In addition,  \cite{bernard2015rationalizing} demonstrate that the optimal payoff of an investor with law-invariant increasing preferences can always be rationalized via a concave utility function, and their optimal payoff obtained by maximizing the corresponding expected utility.

In real-world investment contexts, investors commonly encounter constraints on fund allocation. For instance, shareholders of an investment company often seek a balance between risk and return, whereas a supervisory authority of a pension fund prioritizes the fund's ability to meet participant obligations. Additionally, asset managers frequently establish a benchmark when engaging with clients, striving to outperform it without significant deviation. Moreover, within Markowitz's seminal mean-variance framework, additional constraints are typically imposed on the composition of buy-and-hold portfolios to ensure ``reasonable'' allocations.
As a consequence,  various studies have delved into optimal payoff selection under diverse constraints. However, the majority of these works predominantly tackle distributional constraints of the payoff itself, rather than exploring the behaviour of payoffs in relation to benchmarks. For instance, \cite{basak2001value} derive the optimal payoff for an EUT maximizer under the  constraint of maintaining the  Value-at-Risk (VaR) within acceptable limits. Other studies in this realm include \cite{wei2021risk}, \cite{cuoco2008optimal}, and \cite {chen2024equivalence}.

In this paper, we examine the optimal choices for an EUT investor who is concerned about the deviation of their payoff from a specified benchmark. Related problems have been investigated in \cite{pesenti2020portfolio} and \cite{jaimungal2022SIAM}. The former considers Yaari investors and uses (among other constraints) the 2-Wasserstein distance to measure divergence while the latter focuses on RDEUT and the $p$-Wasserstein distance. The $p$-Wasserstein distance arises as the solution to the Monge-Kantorovich optimal transport (OT) problem with the \textit{symmetric} cost function $c(z_1,z_2) =|z_1-z_2|^p$. Thus positive deviations (gains) from a benchmark are penalised to the same extent as negative ones (losses)\footnote{Other research that has addressed optimal payoff choices under symmetric distance constraints include works by \cite{Bernard_Moraux_Rueschendorf_Vanduffel_13}, \cite{ruschendorf2017}, and \cite{he2021optimal}.}.  In optimal payoff choice, however, 
asymmetry may be desired. Specifically, since the fundamental work of \cite{Tversky_Kahneman_92}, it is well accepted that individuals tend to feel the pain of losses more acutely than the pleasure from equivalent gains. We contribute to the literature by studying the optimal payoff choice when  the investor assigns asymmetric penalties for gains versus losses with respect to an entire reference distribution -- a benchmark -- rather than a specific reference point.  

To the best of our knowledge, such asymmetry to a benchmark distribution has not yet been explored. A notable exception is \cite{pesenti2026WP}, who, in a later development, study a related but distinct problem.
To quantify dissimilarities between gains and losses we use the so-called Bregman-Wasserstein (BW) divergence which was introduced in the OT literature by  \cite{carlier2007monge} and extensively studied in \cite{kainth2025bregman}. A BW divergence is the minimiser of the Monge-Kantorovich OT problem where the cost function is a Bregman divergence. The BW divergence, which is generated by a convex function $\phi$, is thus an asymmetric generalisation of the (squared) Wasserstein distance. In particular, when $\phi(x) = x^2$, it reduces to the squared Wasserstein distance.

We consider a static setting, where admissible payoffs are non-negative measurable functions of a risky asset's terminal value, i.e., we deal with path-independent payoffs (see also \cite{ruschendorf2017}). Clearly an optimal payoff in the static setting might be surpassed if dynamic trading is permitted, but this outperformance is only apparent if  transactions do not bear a cost. In fact, if each transaction has a minimum cost, continuous trading leads to instantaneous bankruptcy and thus is not feasible. Allowing for discrete intermediate trading could mitigate this issue, but it presents major mathematical challenges\footnote{Few studies address portfolio maximization under discrete trading with minimum transaction costs. Notable exceptions include  \cite{BEL22} and \cite{BAY22} who explore the maximization of expected utility in a Black-Scholes market.}. The class of path-independent payoffs is extensive in that any distribution function of terminal wealth can be attained and - under technical conditions - the path-independent payoff can be replicated by trading at series of suitable calls and puts\footnote{The assumption that all calls and puts are available is arguably as reasonable as the assumption that continuous trading is feasible. To  this regard, note that \cite{CAR97} consider this assumption ``analogous to the continuous trading assumption permeating the continuous time literature.''} (see \cite{BRE78}).

The paper is structured as follows. \Cref{sec:2} provides the investors problem description and \Cref{sec:3} provides our main result, that is the optimal investment strategy. \Cref{sec:4} discusses examples and \Cref{sec:5} concludes.

\section{Problem Formulation}
\label{sec:2}

\subsection{Financial Market}
There are two assets available in the market, a risk-free asset with current value $B_0>0$ and value $B_T = B_0 e^{rT}$ at terminal time $T>0$, where $r>0$ is the risk-free interest rate, and a risky asset with current value $S_0>0$ and value at time $T$ that is described by the random variable $S_T: \Omega \to [0,\infty)$. We operate within the { complete atomless} probability space $(\Omega, \mathcal{F}, \mathbb{P})$, in which $\mathcal{F}=\sigma(S_T)$ denotes the $\sigma$-algebra generated by $S_T$ and $\mathbb{P}$ is a probability  measure, which may be an objective statistical measure derived from historical data or a subjective measure chosen by the investor based on personal beliefs and preferences.
	
We consider an investor who aims to acquire at $t=0$ a non-negative payoff $X_T$ that is a measurable function of the risky asset's terminal value {and lies in the space of integrable, non-negative random variables}. Furthermore, $X_T$ must be affordable  given the investor's initial budget $x_0>0$. The admissible payoff $X_T$ thus belongs to the set $\mathcal{X}(x_0)$ defined as 
	\begin{align*}
		 \mathcal{X}(x_0) \coloneqq \big\{X_T=g(S_T) ~|~ g\text{ non-negative and measurable, } {\mathbb{E}_{\mathbb{P}}[X_T] < \infty} \, \text{ and } \,c(X_T)  \leq x_0\big\}\,,
	\end{align*}
	where $c(X_T)$ denotes the cost of $X_T$. We assume that a pricing measure $\mathbb{Q}$ exists and that the cost $c(X_T)$ of $X_T$ is given by $c(X_T):= \mathbb{E}_{\mathbb{Q}}[e^{-rT}X_T]$, where $\E_\Q[\cdot]$ denotes the expected value under $\Q$.  
 Equivalently, the cost $c(X_T)$ can be written as
	\begin{equation*}
	c(X_T)=\mathbb{E}_\mathbb{P}[\varphi _{T}X_{T}],
	\end{equation*}
	where $\varphi_T:=e^{-rT}\frac{\text{d}\mathbb{Q}}{\text{d}\mathbb{P}}$ is the state price density. For simplicity, we write $\E[\cdot] := \E_\P[\cdot]$, when considering the expected value under $\P$.
	
{\begin{remark} If call options $c_K := \max\{S_T-K,0\}$ are traded for all exercise prices $K\geq0$, then the pricing measure $\mathbb{Q}$ is uniquely determined by these call prices, see e.g., \cite{BRE78}, \cite{ROS76}, and \cite{NAC88}. 
\end{remark}}

In the sequel, all cumulative distribution functions (cdfs) are taken with respect to $\mathbb{P}$. For a given cdf $F$, we denote its (left-continuous) quantile function\footnote{The left-continuous quantile function of a cdf $F$ is defined as $\Finv(u) = \inf\{ y \in \mathbb{R}~|~ F(y) \geq u\}, ~ 0 < u <  1$.} by $\Finv$ and for a random variable $Y$, we write $F_Y$ and $\Finv_Y$ for its cdf and quantile function, respectively. Furthermore, we assume that the state-price density $\varphi _{T}$ is continuously distributed, satisfies $0<\varphi _{T}<+\infty$ $\P$-a.s.

As we work with non-negative {integrable} random variables, we define by $\mQ$ the set of quantile functions corresponding to non-negative { integrable} random variables. That is 
\begin{align*}
	\mQ := \big\{ \Ginv:(0,1) \to [0, \infty) \ | \ \Ginv \ \text{is non-decreasing, left-continuous, and }    { \int_0^1 \Ginv(u)\d u  < \infty} \big\}. 
\end{align*}

{The investor employs a utility function $u(\cdot)$ to attach a value $u(x)$ to each payoff outcome $x$ of $X_T$ and evaluates $X_T$ by its expected utility $\E[u(X_T)]$. Throughout, we assume that the utility function $u:[0,\infty)\to\mathbb R$ is continuous on $[0,\infty)$ and twice continuously differentiable on $(0,\infty)$, and satisfies:
\begin{enumerate}[label = $\roman*)$]
    \item $u$ is strictly increasing and strictly concave on $(0,\infty)$, i.e.,
    \begin{align*}
    u'(x)>0 \quad \text{and}\quad u''(x)<0,\qquad x>0; 
    \end{align*}
    \item $u$ satisfies the Inada conditions, i.e.,
    \begin{align*}
    \lim_{x\to 0}u'(x)=\infty \quad \text{and}\quad  \lim_{x \to \infty} u'(x) = 0\,. 
    \end{align*}
\end{enumerate}}

\subsection{Bregman-Wasserstein Uncertainty}
The investor aims to control the divergence from their payoff to a given benchmark. To formalise this we utilise the notion of Bregman-Wasserstein divergence, for which we first recall the Bregman divergence.   
\begin{definition}[Bregman divergence]\label{defBD}
Let $\phi\colon \mathbb{R} \to \mathbb{R}$ be a convex and continuously differentiable function, called a Bregman generator. Then, the Bregman divergence with generator $\phi$ is defined as
	\begin{equation*}
		B_\phi\big(z_1, z_2\big)
		:= \phi(z_1) - \phi(z_2) - \phi'(z_2) (z_1-z_2)
		\,,\quad z_1,z_2\in\mathbb{R}\,,
	\end{equation*}
	where $\phi'(z):= \frac{d}{dz} \phi(z)$ denotes the derivative of $\phi$.
\end{definition}

As we work with non-negative random variables, we typically consider Bregman generators $\phi \colon [0, \infty) \to \mathbb{R}$.
Note that for the choice $\phi(x)=x^2$, it holds that $B_{\phi}(z_1, z_2) = (z_1- z_2)^2$, i.e., we obtain the squared Euclidean distance. In general, however, the Bregman divergence lacks symmetry (and thus is not a distance).  
This asymmetry is of great interest in that it makes it possible to penalize positive deviations differently to negative ones. \Cref{BW_p1} illustrates the Bregman divergence for $\phi_1(x)=x^2$ (blue lines) and $\phi_2(x)=x\ln x$, $x>0$ (red lines). By symmetry it holds that $B_{\phi_1}(1.5, 0.8)= 0.49 = B_{\phi_1}(0.8, 1.5)$, while for $\phi_2(x)=x\ln x$, we have $B_{\phi_2}(1.5, 0.8)=0.2429>0.1971=B_{\phi_2}(0.8, 1.5)$.

\begin{figure}[h]
    \centering
    \begin{subfigure}[b]{0.49\textwidth}
        \centering
        \includegraphics[width=\textwidth]{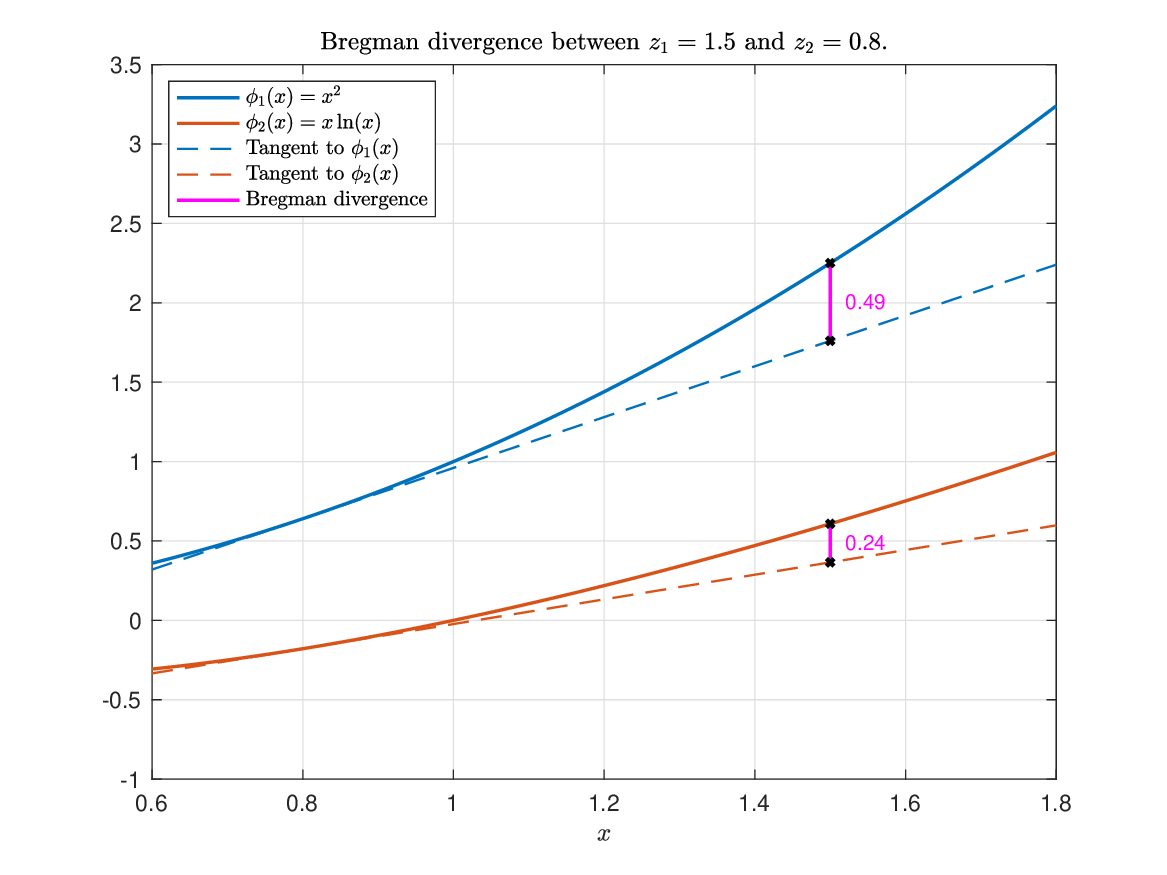}
        \label{BD_1}
    \end{subfigure}
    \hfill
    \begin{subfigure}[b]{0.49\textwidth}
        \centering
        \includegraphics[width=\textwidth]{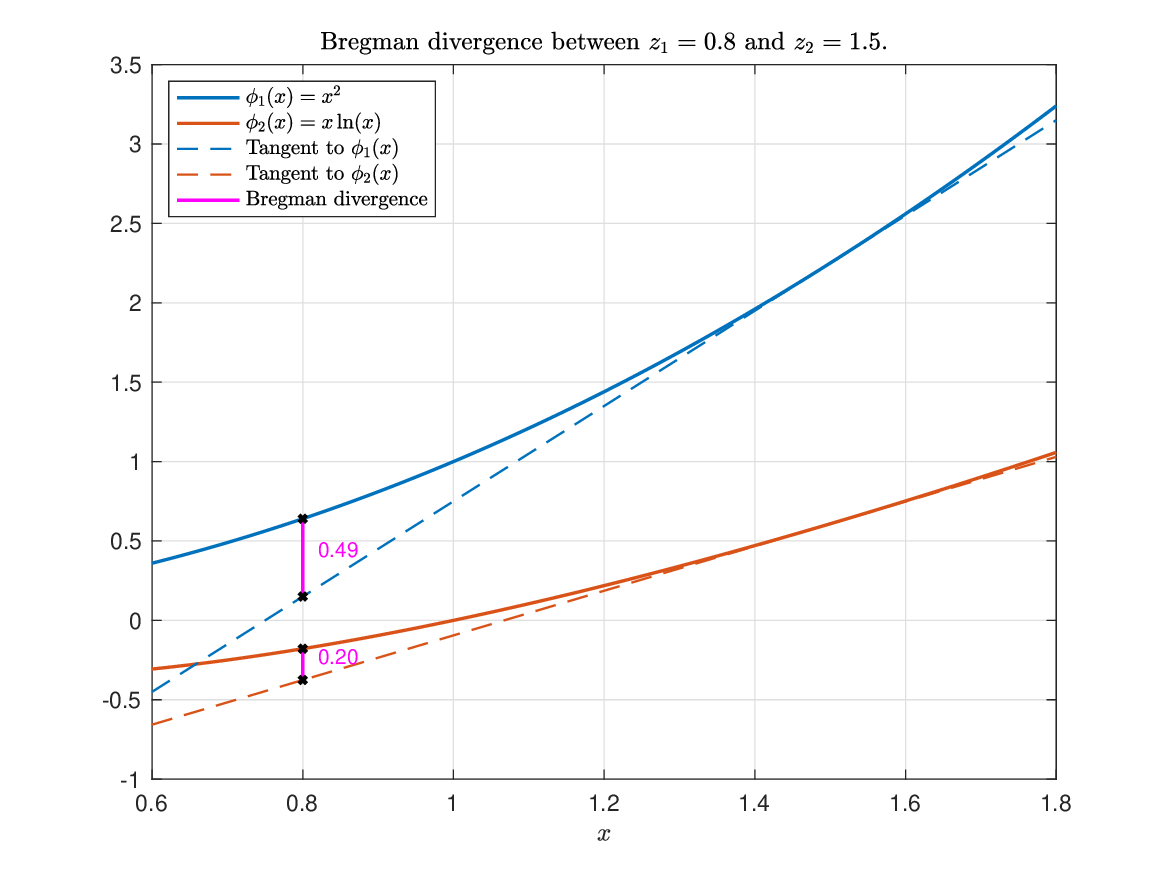}
        \label{BD_2}
    \end{subfigure}
	\caption{The Bregman divergence for the generators $\phi_1(x)=x^2$ (blue lines) and $\phi_2(x)=x\ln x$ (red lines). Left panel displays $B_{\phi_i}(1.5, 0.8)$ (purple vertical lines) and right panel displays $B_{\phi_i}(0.8, 1.5)$ (purple vertical lines), $i = 1,2$. 
 }
	\label{BW_p1}
\end{figure}

With the Bregman divergence, we can define a divergence on the space of cdfs as follows.
Let $\Pi(F_1,F_2)$ denote the set of all bivariate cdfs with marginal cdfs $F_1$ and $F_2$, respectively. \cite{Pesenti2023ORL} investigate the Monge-Kantorovich optimal transport problem 
\begin{equation*}
	\DB[F_1, F_2]:= \inf_{\pi\in\Pi(F_1,\,F_2)} \;\left\{\,\int_{\mathbb{R}^2} c(z_1, z_2)\,\pi(\d z_1,\d z_2)\, \right\},
\end{equation*}%
using as cost function $c(z_1, z_2)$ the Bregman divergence $ B_\phi(z_1, z_2)$ and show that the infimum
is attained for the comonotonic coupling, i.e., by $\big(\Finv_1(U), \Finv_2(U)\big)$, for any $U \sim U(0,1)$. As a consequence,  the so-called  Bregman-Wasserstein (BW) divergence from $F_1$ to $F_2$ has representation
\begin{align*}
	\DB_\phi [F_1, F_2]
         &=\int_0^1 B_\phi\big( \Finv_1(t), \Finv_2(t)\big)\,\d t \notag
    \\
    &=\int_0^1 \Big(   \phi\big(\Finv_1(t)\big)-\phi\big(\Finv_2(t)\big)
	-\phi'\big(\Finv_2(t)\big)\big(\Finv_1(t) -\Finv_2(t)\big)\, \Big)\,\d t\,, 
\end{align*}
which reduces to the 2-Wasserstein distance for $B_\phi$ being the squared loss, i.e., when $\phi(x) = x^2$. Since the BW divergence only depends on the marginal cdfs, we write $\DB_\phi [F_1, F_2]=\DB_\phi (\Finv_1, \Finv_2)$ and use square brackets for cdfs and round brackets for quantile functions.

{  
The asymmetry of the Bregman-Wasserstein divergence reflects how investors evaluate deviations relative to a benchmark. Since the seminal work of \cite{Tversky_Kahneman_92}, it has been well documented that individuals exhibit loss aversion, that is they weigh losses more heavily than gains of equal magnitude. For instance, pension funds or insurance companies prioritize downside protection to ensure solvency, while asset managers may tolerate some underperformance but seek to limit substantial drawdowns relative to benchmarks. The symmetric Wasserstein distance, which penalizes upward and downward deviations equally, cannot capture these asymmetric sensitivities. By contrast, the BW divergence permits a richer class of constraints where negative deviations can be penalized more severely than positive ones, aligning payoff design with observed investor behaviour and regulatory practice.
}

\subsection{Investor's Optimization Problem}
The investor aims to maximize their expected utility under the constraint that their terminal wealth does not diverge too much from the terminal wealth arising from a benchmark (reference) portfolio. Thus, we study the following 
EUT distance constrained optimization problem:
\begin{equation} \label{Yaari-o}
\underset{\begin{cases}
	\ c(X_{T}) \leq x_{0}\\[0.3em]
	\Finv_{X_T} \in \mQ_\ep
	\end{cases}%
	}{\max }\E[u(X_{T})]\,, \tag{P}
\end{equation}%
where $\mQ_\ep$ is the set of quantile functions that have a BW divergence of at most $\ep>0$ to the quantile function of the benchmark's terminal wealth $\Fbenchinv$, i.e., 
\begin{equation*}
    \mQ_\ep:= \Big\{ \Ginv ~ |~ \Ginv \in \mQ \quad \text {and} \quad \DB_\phi (\Ginv, \Fbenchinv) \leq \epsilon  \Big\}\,.
\end{equation*}

The value of $\epsilon>0$ is chosen by the investor and we refer to it as the \textit{tolerance level}. To solve problem \eqref{Yaari-o}, we first reduce the set containing all payoffs $X_{T}$ to the subset of payoffs that are anti-monotonic\footnote{Two random variables $X$ and $Y$ are called anti-monotonic if for every $\omega_1$, $\omega_2 \in \Omega$ it holds that $(X(\omega_1)-X(\omega_2))(Y(\omega_1)-Y(\omega_2)) \leq 0.$ By contrast, they are comonotonic if for every $\omega_1$, $\omega_2 \in \Omega$, $(X(\omega_1)-X(\omega_2))(Y(\omega_1)-Y(\omega_2)) \geq 0$.} to the state price density $\varphi _{T}$. To see this, assume that some payoff $X_T$ distributed with $F$ solves the EUT optimal payoff selection problem \eqref{Yaari-o}\footnote{  The existence of such a distribution $F$ can be ensured under \Cref{assumption_1} and \Cref{asm-well-posed}; see \Cref{Exist} for a detailed proof. } and is not anti-monotonic to the state price density $\varphi _{T}$. We can then obtain a payoff $\hat{X}_T$ that has the same cdf $F$ and is strictly cheaper. Indeed, $\hat X_{T}:=\Finv(1-F_{\varphi _{T}}(\varphi _{T}))$ has distribution $F$, thus the same expected utility, and satisfies $c(\hat{X}_T) < c(X_T)$ by the Hoeffding--Fr\'echet bounds. This reasoning was first established in \cite{Dybvig_88} and we refer to  \cite{ruschendorf2017} and many of the references therein for more detail. As a consequence, the optimal portfolio selection problem reduces to an optimization problem on real functions under monotonicity constraints (quantile functions).

\begin{theorem} [Anti-monotonicity of solution]
	\label{theo:1a} 
	If the EUT optimization problem \eqref{Yaari-o} has a solution, then there exists a solution $\hat X_{T}$ that is anti-monotonic with $\varphi _{T}$. 
\end{theorem}

Payoffs that are anti-monotonic with the state price density $\varphi _{T}$ are called cost-efficient, see e.g., \cite{bernard2014explicit}. 

Next, we define the set of quantile functions that have a BW-distance to the reference cdf of at most $\ep$ and whose corresponding cost-efficient payoff costs no more than $x_0$. Specifically,
\begin{equation*}
    \mQ_\ep(x_0) \coloneqq \bigl\{\Ginv \in \mQ~\big| 
    ~ \DB_\phi (\Ginv, \Fbenchinv) \leq \epsilon \; \text{ and } \; 
    \cost(\Ginv)\leq x_0 \bigr \}\subseteq \mQ_\ep\,,
\end{equation*}
where $\cost(\Ginv):= \int_0^1 \Ginv(t)\Finv_{\varphi_T}(1-t) \mathrm{d}t$.

The following result is now classical in the literature on portfolio choice and we omit its proof.\footnote{
This ``quantile reformulation'' of the optimal payoff selection problem has been first established in the pioneering work of \cite{Dybvig_88}. It has been further refined and applied in numerous research works including those by \cite{Foellmer_Schied_04}, \cite{Carlier_Dana_06}, \cite{Burgert_Rueschendorf_06}, \cite{jin2008behavioral}, \cite{kassberger2012path}, \cite{Xu_Zhou_13}, \cite{Bernard_Moraux_Rueschendorf_Vanduffel_13}, \cite{von2014optimality}, \cite{Xu_14}, \cite{xu2016note}, \cite{ruschendorf2017}, \cite{pesenti2020portfolio}, \cite{he2021optimal}, \cite{wei2021risk}, \cite{bi2021optimal},  \cite{magnani2022efficiency}, and \cite{bernard2024cost}.}

\begin{theorem}[Quantile reformulation]\label{theo:1b} 
	If $\Finv \in \mQ_\ep(x_0)$ is a solution to the optimization problem   
		\begin{equation} \label{Yaari2}\tag{\text{$\breve{P}$}}
		\underset{
				\Ginv \in {\mQ_\ep(x_0)}
		}{\max }  \int_0^1 u(\Ginv(t))\mathrm{d}t\, ,
		\end{equation}%
	then $\hat X_{T}={ {\Finv}}(1-F_{\varphi_{T}}(\varphi_{T}))$
	solves problem \eqref{Yaari-o}. 
\end{theorem}

Solving the EUT optimal payoff choice problem under a BW divergence constraint thus involves addressing an optimization problem concerning functions under monotonicity constraints. In other words, it entails finding a suitable quantile function.

\section{Optimal Payoff}\label{sec:3}
In this section we solve the EUT optimisation problem \eqref{Yaari2}. For this we first establish that problem \eqref{Yaari2} is  {  well-posed} and that there exists a unique solution. 

{We work under the following standard regularity assumption.
\begin{assumption}\label{assumption_1}
There exists a quantile function $\Ginv \in \mQ$ such that
\begin{align*}
\cost(\Ginv)<x_0
\quad\text{and}\quad
\DB_\phi(\Ginv,\Fbenchinv)<\varepsilon. 
\end{align*}
\end{assumption}
This assumption corresponds to a Slater-type condition (see \Cref{main_optimal_th}) for the constrained optimization problem.  In particular, it implies that the feasible set $\mQ_\ep(x_0)$ is non-empty.

Note that the Slater condition also requires that the expected utility of the quantile
function is finite. Since $\breve G\in\mathcal{Q}\subset L^1((0,1))$, the space of Lebesgue integrable functions on $(0,1)$, and $u$ is increasing and concave with $u(0)\in\mathbb{R}$, the expected utility of $\breve{G}$ is finite. For more details, please refer to \eqref{utility_finite}.
}

{ 
\begin{remark} \label{Qae}
For the convex-analytic arguments below, in line with \cite{ghossoub2025risk}, we work in $L^1((0,1))$, and identify functions that agree Lebesgue-a.e. Accordingly, we use
\begin{align*}
 \mQ^{ae}
:=
\bigl\{
\Ginv\in L^1((0,1)) \ \big|\
\Ginv\ge 0 \text{ a.e. and } \Ginv \text{ is a.e. non-decreasing}
\bigr\}.
\end{align*}
The set $\mQ^{ae}$ is the $L^1$-equivalence-class version of $\mQ$. More precisely, every element of $\mQ$ defines an element of $\mQ^{ae}$, and conversely every element of $\mQ^{ae}$ admits a unique non-decreasing left-continuous representative in $\mQ$.
\end{remark}
}

\subsection{Existence and Uniqueness}

To obtain existence and uniqueness of the optimal payoff, we first establish that {  problem \eqref{Yaari2} is well-posed, i.e., its value is finite. Before proving the well-posedness of problem \eqref{Yaari2}, we introduce the following definition.} 

{  
\begin{definition}[$\vartheta$-strongly convex]
A function $f: \mathbb{R} \to \mathbb{R}$ is $\vartheta$-strongly convex if there exist $\vartheta>0$ such that 
\begin{equation}
    f(z_1) \geq f(z_2) + f'(z_2)(z_1 - z_2) + \frac{\vartheta}{2}(z_1 - z_2)^2 \,,\quad z_1,z_2\in\mathbb{R}\,.
\end{equation}
\end{definition}

    For any Bregman generator $\phi$ that is $\vartheta$-strongly convex, the induced divergence is naturally linked to the squared Euclidean distance and to the Wasserstein divergence. Specifically, recalling \Cref{defBD}, we have
    \begin{equation*}
		B_\phi\big(z_1, z_2\big) \geq \frac{\vartheta}{2}(z_1 - z_2)^2 = \frac{\vartheta}{2} B_{\phi_1}\big(z_1, z_2\big) \,, \quad z_1,z_2\in\mathbb{R}\,,
	\end{equation*}
    where $\phi_1(x) = x^2$.
    This property plays a key role in establishing integrability and thus well-posedness in our framework. The next assumption provides three different cases, under which the problem \eqref{Yaari2} is well-posed. 
}

{  
\begin{assumption}\label{asm-well-posed}
    Let at least one of the following conditions hold:
   \begin{enumerate}[label = $\roman*)$]
      \item The state-price density satisfies $\operatorname*{ess\,inf} \varphi_T > 0$.
      \item The Bregman generator $\phi$ is $\vartheta$-strongly convex and the benchmark $\breve{F}_b$ belongs to $L^2$. 
      \item The Bregman generator $\phi$ is strictly convex and the benchmark $\breve{F}_b$ is bounded.
   \end{enumerate}
\end{assumption}

\begin{proposition}[Well-posedness]\label{Well-posedness}
   Under \Cref{asm-well-posed} problem \eqref{Yaari2} is well-posed.
\end{proposition}

\begin{proof}
{For any feasible solution $\breve{F} \in \mQ_\ep(x_0)$ of problem \eqref{Yaari2}, Jensen's inequality gives
\begin{equation*}
        \int_{0}^{1} u(\breve{F}(t))\mathrm{d}t \leq u\Big( \int_{0}^{1} \breve{F}(t)\mathrm{d}t \Big) \,.
    \end{equation*}
Therefore, to prove that the value $V$ of problem \eqref{Yaari2} is finite, it is sufficient to show that the feasible set is uniformly bounded in $L^1((0,1))$, namely, that there exists a constant $C>0$ such that
$\int_0^1 \breve F(t) \mathrm{d} t\le C$,
for every feasible $\breve{F}$. }

    For case $i)$, we have 
    \begin{align*}
        \int_{0}^{1} \breve{F}(t) \mathrm{d}t  
        \leq 
        \frac{1}{ \displaystyle\lim_{t \to 1} \breve{F}_{\varphi_{T}}(1-t) } \int_{0}^{1} \breve{F}(t) \breve{F}_{\varphi_{T}}(1-t) \mathrm{d}t
     \leq 
     \frac{x_0}{ \displaystyle\lim_{t \to 1} \breve{F}_{\varphi_{T}}(1-t) } 
     <  \infty \,.
    \end{align*}
    
    For case $ii)$, we note that as $\phi$ is $\vartheta$-strong convex, we have $B_\phi\big(z_1, z_2\big) \geq \frac{\vartheta}{2}(z_1-z_2)^2$ as well as $\DB_\phi (\breve{F}_1, \breve{F}_2) \geq \frac{\vartheta}{2} \DB_{\phi_{1}} (\breve{F}_1, \breve{F}_2)$, with $\phi_1(x)=x^2$.
Then, from the BW constraint, we obtain 
\begin{equation*}
   \frac{\vartheta}{2} \DB_{\phi_{1}} (\breve{F}, \breve{F}_b) \leq  \DB_{\phi} (\breve{F}, \breve{F}_b) \leq \varepsilon \,,
\end{equation*}
that is
\begin{equation*}
   \DB_{\phi_{1}} (\breve{F}, \breve{F}_b) = \int_{0}^{1} \big( \breve{F}(t) - \breve{F}_b(t) \big)^2  \mathrm{d}t \leq \frac{2 \varepsilon}{\vartheta}. 
\end{equation*}
Moreover, we have
\begin{align*}
    \int_0^1 \breve{F}(t) \mathrm{d}t  - \int_0^1 \breve{F}_b(t) \mathrm{d}t
    &\le 
   \Big| \int_0^1 \breve{F}(t) \mathrm{d}t  - \int_0^1 \breve{F}_b(t) \mathrm{d}t \Big|
   \\
   &\leq \sqrt{ \int_{0}^{1} \big( \breve{F}(t) - \breve{F}_b(t) \big)^2  \mathrm{d}t } \\
   &\leq \sqrt{\frac{2 \varepsilon}{\vartheta}}\,,
 \end{align*}
 where the second inequality follows from the Cauchy-Schwarz inequality (applied to the functions $\breve{F}(\cdot) - \breve{F}_b(\cdot)$ and $1$), and we obtain  
 \begin{equation*}
     \int_{0}^{1} \breve{F}(t) \mathrm{d}t  
     \leq  \int_0^1 \breve{F}_b(t) \mathrm{d}t +  \sqrt{\frac{2 \varepsilon}{\vartheta}}  
     <  \infty.
\end{equation*}

For case $iii)$, let the Bregman generator $\phi$ be strictly convex and the benchmark be bounded.
Since the benchmark is bounded, then for all $t \in (0,1)$, we have $m_b \leq \breve{F}_b(t) \leq M_b$, where $m_b$ and $M_b$ are two given constants. Consider $x>0$ and $y = \breve{F}_b(t)$, from the definition of Bregman divergence, we have
\begin{align*}
        B_\phi\big( x, y\big) =& \ \phi(x) - \phi(y) - \phi'(y) (x-y) \\
        =& \ \phi(x) - \phi'(y)x + \big( y \phi'(y) - \phi(y) \big),
\end{align*}
where the terms $\phi'(y)$ and $y \phi'(y) - \phi(y)$ are bounded because $y \in [m_b, M_b]$ and $\phi, \phi'$ are continuous on the compact set. 

Let $C_1= \phi'(M_b)$ and $C_2 = \sup_{y \in [m_b,M_b]} |y\phi'(y) - \phi(y)|$, we obtain 
\begin{equation*}
    B_\phi\big( x, y\big) \geq \phi(x) - C_1 x -C_2.
\end{equation*}
Note that $\phi$ is strictly convex, its derivative $\phi'$ is strictly increasing. Hence,
\begin{align*}
    \lim_{x \to \infty} \phi'(x)=L \in (C_1, \infty].
\end{align*}
Next, choose $\delta \in (0, L - C_1)$ (if $L = \infty$, take any $\delta \in (0, \infty)$). Then there exists $a_0 > M_b$ such that $\phi'(x) \ge C_1 + \delta$ for all $x \ge a_0$.
By the convexity of $\phi(x)$, we have
\begin{equation*}
    \phi(x) \geq \phi(a_0) + \phi'(a_0)(x-a_0) \geq \phi(a_0) + (C_1 + \delta)(x-a_0) \,,\quad x \geq a_0.
\end{equation*}
Therefore,
\begin{align*}
    \phi(x) - C_1 x - C_2 \geq & \ \phi(a_0) + (C_1 + \delta)(x-a_0) - C_1 x  - C_2 \\
    = & \ \delta x - ((C_1 + \delta )a_0 - \phi(a_0) ) - C_2  \,,\quad x \geq a_0.
\end{align*}
Taking $k_0 = \delta$ and $C_0 = \max \{C_2 + (C_1 + \delta)a_0 - \phi(a_0) ,0 \}$, we obtain that for all $x > a_0$ and all $y \in [m_b, M_b]$,
\begin{align*}
    B_\phi\big( x, y\big) \geq k_0 x -C_0.
\end{align*}

Let $\mathcal{S} = \{t \in (0,1) \ | \ \breve{F}(t) > a_0 \}$ and $\overline{\mathcal{S}} = \{t \in (0,1) \ | \ \breve{F}(t) \leq a_0 \}$, consider the BW constraint, 
\begin{align*}
     k_0 \int_{\mathcal{S}} \breve{F}(t) \mathrm{d}t - C_0
      &\leq \int_{\mathcal{S}} \big( k_0 \breve{F}(t) - C_0 \big) \mathrm{d}t  \\
    &\leq \int_{\mathcal{S}}B_{\phi}\big( \breve{F}(t), \breve{F}_b(t) \big) \mathrm{d}t \\
   &\leq   \DB_{\phi} (\breve{F}, \breve{F}_b) \\ 
   &\leq \varepsilon,
\end{align*}
where the first inequality arises from  the measure of $\mathcal{S}$ being at most 1. Thus we obtain
\begin{align*}
    \int_{\mathcal{S}} \breve{F}(t) \mathrm{d}t \leq \frac{ C_0 + \varepsilon }{k_0} \qquad \text{and}\qquad 
    \int_{\overline{\mathcal{S}}} \breve{F}(t) \mathrm{d}t \leq a_0.
\end{align*}
Thus, we derive that
\begin{align*}
    \int_{0}^{1} \breve{F}(t) \mathrm{d}t  
     = \int_{\mathcal{S}} \breve{F}(t) \mathrm{d}t + \int_{\overline{\mathcal{S}}} \breve{F}(t) \mathrm{d}t  
     \leq  \frac{ C_0 + \varepsilon }{k_0} + a_0  
     < \infty.
\end{align*} 

In conclusion, we showed that in all three cases, any feasible solution has a finite expected utility. Thus problem \eqref{Yaari2} is well-posed.    
\end{proof}

\begin{remark}\label{remark_convex}
    % For the classical expected-utility maximization problem (i.e., with only the budget constraint), well-posedness and the existence of an optimal solution requires  assumptions; see \cite{jin2008convex} for details. 
    If a Bregman generator $\phi$ does not satisfy strong convexity—for example, $\phi_{2}(x) = x\ln x$ or the threshold-based generator introduced in \Cref{defMBD} —strong convexity can be enforced by adding a regularization term of the form $\frac{\vartheta}{2} x^2$, where $\vartheta > 0$ is sufficiently small (e.g., $10^{-8}$), to the Bregman generator. For instance, for $\phi_{2}(x)$, one may consider the regularized version $\phi_{2,\vartheta}(x) = x\ln x + \frac{\vartheta}{2} x^2 $. Conceptually, this regularization adds a weighted Wasserstein-type constraint, $\frac{\vartheta}{2} \DB_{\phi_{1}} (\Finv_1, \Finv_2)$ with $\phi_{1}(x)=x^2$, to the original Bregman–Wasserstein constraint $\DB_{\phi_2} (\Finv_1, \Finv_2)$, i.e., 
    \begin{equation}\label{add_convex}
        \DB_{\phi_{2,\vartheta}} (\Finv_1, \Finv_2) = \DB_{\phi_2} (\Finv_1, \Finv_2) + \frac{\vartheta}{2} \DB_{\phi_{1}} (\Finv_1, \Finv_2),
    \end{equation}
    thereby ensuring the well-posedness of the problem. As observed from \eqref{add_convex}, when $\vartheta$ is sufficiently small, the deviation penalty is primarily governed by $\DB_{\phi_2} (\Finv_1, \Finv_2)$.
\end{remark}
}

{
\Cref{asm-well-posed} provides three explicit and verifiable sufficient conditions for the well-posedness of problem~\eqref{Yaari2}. Case 1 relies only on an assumption on the state price density \(\varphi_T\), a classical condition in the literature; see Theorem~5.1 in \cite{jin2008convex}. Cases 2 and 3 rely on the BW distance constraint, which restricts the feasible portfolio set to a neighbourhood of the benchmark. 
These cases cover a broad range of practically relevant settings, including both bounded and unbounded benchmarks, and contrast with works where well-posedness is imposed directly; see, e.g., \cite{wei2018risk,mostovyi2018utility}. 

The following example shows that, without such conditions, problem~\eqref{Yaari2} may fail to be well-posed.

\begin{example}[Ill-posedness under the budget and BW divergence constants]
    \label{ex:illposed_budget_BW}
Consider the problem
\begin{align*}
 \sup_{\Ginv \in \mathcal{Q}}\left\{
\int_0^1 u(\breve G(t))\,\mathrm{d}t \ \Big| \ \mathfrak c(\breve G)\le x_0,\ 
\DB_\phi(\breve G,\breve F_b)\le \varepsilon
\right\},
\end{align*}
where $u(x)=\sqrt{x}$ for $x\ge0$, the Bregman generator is $\phi(x)=x+e^{-x}$,
the benchmark is given by
\begin{align*}
 \breve F_b(t)=(1-t)^{-1/2},\quad t\in(0,1),
\end{align*}
and the quantile function of the state-price density is 
\begin{align*}
\breve{F}_{\varphi_T}(t)=e^{-t^{-\beta}},\quad t\in(0,1), 
\end{align*}
with $\beta=\frac35$. { We observe that none of the three cases in  \Cref{asm-well-posed} applies.}

For $N\in\mathbb N$, define
\begin{align*}
 A_N:=N^4,
\quad
\breve G_N(t):=\breve F_b(t)+A_N\,\mathbf 1_{(1-\frac1N,\,1)}(t),\quad t\in(0,1).
\end{align*}
Then each \(\breve G_N:(0,1)\to[0,\infty)\) is nondecreasing and left-continuous, hence a quantile function. Also, for every $y,a\ge0$,
\begin{align*}
B_\phi(y+a,y)
=
\phi(y+a)-\phi(y)-\phi'(y)a
=
e^{-y}(e^{-a}-1+a)
\le a e^{-y}. 
\end{align*}
Therefore,
\begin{align*}
\DB_\phi(\breve G_N,\breve F_b)
=
\int_{1-\frac1N}^1
B_\phi\!\bigl(\breve F_b(t)+A_N,\breve F_b(t)\bigr)\, \mathrm{d}t
\le
A_N\int_{1-\frac1N}^1 e^{-\breve F_b(t)}\,\mathrm{d}t. 
\end{align*}
With the change of variable \(s=1-t\), this yields
\begin{align*}
\DB_\phi(\breve G_N,\breve F_b)
&\le 
A_N\int_0^{ \frac{1}{N} } e^{-s^{-\frac{1}{2}}}\,\mathrm{d}s \\
&\le 
A_N\cdot \frac1N e^{-N^{\frac{1}{2}}} \\
 &=
N^3e^{-N^{\frac{1}{2}}}
\longrightarrow 0 \quad \text{as} \quad N \to \infty. 
\end{align*}
Similarly,
\begin{align*}
 \mathfrak c(\breve G_N)-\mathfrak c(\breve F_b)
&=
A_N\int_{1-\frac1N}^1 \breve{F}_{\varphi_T}(1-t)\,\mathrm{d}t \\
&=
A_N\int_0^{\frac{1}{N}} e^{-s^{-\beta}}\,\mathrm{d}s \\
& \le
A_N\cdot \frac1N e^{-N^\beta} \\
& =
N^3e^{-N^\beta}
\longrightarrow 0 \quad \text{as} \quad N \to \infty.
\end{align*}
Since
\begin{align*}
\mathfrak c(\breve F_b)
=
\int_0^1 (1-t)^{-1/2}e^{-(1-t)^{-\beta}}\,\mathrm{d}t =
\int_0^1 s^{-1/2}e^{-s^{-\beta}}\,\mathrm{d}s 
<\infty, 
\end{align*}
it follows that
\begin{align*}
x_0:=\sup_{N\ge1}\mathfrak c(\breve G_N)<\infty,
\quad
\varepsilon:=\sup_{N\ge1}\DB_\phi(\breve G_N,\breve F_b)<\infty. 
\end{align*}
Hence $(\breve G_N)_{N\ge1}$ is feasible for the above problem.

On the other hand, since $\breve G_N(t)\ge A_N$ on $(1-\frac1N,1)$,
\begin{align*}
\int_0^1 u(\breve G_N(t))\,\mathrm{d}t
&=
\int_0^1 \sqrt{\breve G_N(t)}\,\mathrm{d}t \\
&\ge
\int_{1-\frac1N}^1 \sqrt{A_N}\,\mathrm{d}t \\
& =
\frac1N\sqrt{A_N}
=
N
\longrightarrow\infty \quad \text{as} \quad N \to \infty.
\end{align*}
Consequently,
\begin{align*}
 \sup_{\Ginv \in \mathcal{Q}}\left\{
\int_0^1 u(\breve G(t))\,\mathrm{d}t \ \Big| \ \mathfrak c(\breve G)\le x_0,\ 
\DB_\phi(\breve G,\breve F_b)\le \varepsilon
\right\}
=+\infty, 
\end{align*}
so the problem is ill-posed.
\end{example}
}

\begin{theorem}[Existence and uniqueness]  \label{Exist}
{  Let Assumptions \ref{assumption_1} and \ref{asm-well-posed} be satisfied.} Then
there exists a unique solution, denoted by $\Finv^*$, of Problem \eqref{Yaari2}. 
\end{theorem}

\begin{proof}
The proof is broken into two parts. 	First, we show that the constraint set $\mQ_\ep(x_0)$ is convex. Second, we prove existence and uniqueness of the solution.

	\underline{Part 1.} Convexity of constraint set. Take $\Ginv_1,\Ginv_2 \in \mQ_\ep(x_0) $ and define the convex combination $\Ginv_{\omega}:=\omega \Ginv_1 + (1-\omega)\Ginv_2$ for $\omega \in [0,1]$.  {Since $\mathcal{Q}$ is convex, we have $\breve {G}_{\omega}\in \mathcal{Q}$. }It remains to verify the budget and Bregman–Wasserstein constraints.  Clearly, $\int_0^1 \Ginv_\omega(t)\Finv_{\varphi_T}(1-t) \mathrm{d}t  \leq x_0$. Furthermore, as the Bregman divergence is convex in its first argument, we have 
 \begin{align*}
     \DB_\phi (\Ginv_\omega, \Fbenchinv)
     &=
     \int_0^1 B_\phi\big( \Ginv_\omega(t), \Fbenchinv(t)\big)\,\d t
     \\
     &\leq 
     \omega \int_0^1 B_\phi\big( \Ginv_1(t), \Fbenchinv(t)\big)\,\d t 
     + (1-\omega)\int_0^1 B_\phi\big( \Ginv_2(t), \
     \Fbenchinv(t)\big)\,\d t
     \\
     &=\omega \,\DB_\phi (\Ginv_1, \Fbenchinv) + (1-\omega) \,\DB_\phi (\Ginv_2, \Fbenchinv) 
     \\
     &\leq \epsilon\,,
 \end{align*}
 where the last inequality follows since $\Ginv_1, \Ginv_2\in\mQ_\ep(x_0)$.
  Hence $\Ginv_{\omega} \in  \mQ_\ep(x_0)$ and  $\mQ_\ep(x_0)$ is a convex set.

 \underline{Part 2.} Existence and uniqueness. {
    Let $V$ denote the value of problem \eqref{Yaari2}. 
    Since $V < \infty$ (\Cref{Well-posedness}) and the feasible region is nonempty (\Cref{assumption_1}), there exists a sequence $(\breve{G}_{n})_{n \geq 1} \subseteq \mathcal{Q}_{\varepsilon}(x_0)$ such that 
     \begin{equation}\label{eq:maxseq}
     \int_{0}^{1} u(\breve{G}_{n}(t)) \mathrm{d} t  > V - \frac{1}{n}, \quad \forall n \geq 1.
     \end{equation}

By the proof of \Cref{Well-posedness}, there exists a constant $C>0$ such that
\begin{equation}\label{eq:L1bound}
\int_0^1 \breve G(t)\, \mathrm{d}t \le C,
\quad \forall\, \breve G\in \mathcal Q_\varepsilon(x_0).
\end{equation}
In particular,
\begin{align*}
\sup_{n\ge1}\|\breve G_n\|_{L^1((0,1))}<\infty. 
\end{align*}
Therefore, by Koml\'os' theorem \citep{komlos1967generalization}, see also Theorem 1.1 in \cite{gao2017uo}, there exist a subsequence \((\breve G_{n_k})_{k\ge1}\) and some \(\breve G^*\in L^1((0,1))\) such that the Ces\`aro means
\begin{equation}\label{eq:Cesaro}
\bar G_m(t):=\frac1m\sum_{k=1}^m \breve G_{n_k}(t),\quad m\ge1,
\end{equation}
satisfy
\begin{equation}\label{eq:aeconv}
\bar G_m(t)\longrightarrow \breve G^*(t) \quad \text{as} \quad m \to \infty
\quad \text{for a.e. }t\in(0,1).
\end{equation}

We next show that $\breve G^*$ is feasible.
Since each $\breve G_{n_k} \in \mQ^{ae}$ and $\mQ^{ae}$ is convex, every Ces\`aro mean $\bar G_m$ also belongs to $\mQ^{ae}$. 
In particular, $\bar G_m\ge 0$ a.e. for every $m$, and hence $\breve G^*\ge0$ a.e. by \eqref{eq:aeconv}. 
Moreover, since $\bar G_m\in \mQ^{ae}$ for all $m$, there exists a measurable set $E\subset(0,1)$ with Lebesgue measure one such that, for every $m$, $\bar G_m$ is non-decreasing on $E$, and $\bar G_m(t)\to \breve G^*(t)$ for every $t\in E$. Hence, for any $s,t\in E$ with $s<t$, we have
\begin{align*}
\bar G_m(s)\le \bar G_m(t),\quad \forall m\ge1. 
\end{align*}
Passing to the limit as $m\to\infty$, we obtain
\begin{align*}
\breve G^*(s)\le \breve G^*(t). 
\end{align*}
Therefore, $\breve G^*$ is non-decreasing on $E$, and thus $\breve G^*$ is a.e. non-decreasing. Then, we obtain $\breve G^*\in \mQ^{ae}$.

Moreover, since the budget functional is linear,
\begin{align*}
\mathfrak{c}(\bar G_m)
=\frac1m\sum_{k=1}^m \mathfrak{c}(\breve G_{n_k})
\le x_0,
\quad m\ge1. 
\end{align*}
Because $\breve F_{\varphi_T}(1-\cdot)\ge0$, Fatou's lemma and \eqref{eq:aeconv} imply
\begin{equation}\label{eq:budgetlimit}
\mathfrak{c}(\breve G^*)
=
\int_0^1 \breve G^*(t)\,\breve F_{\varphi_T}(1-t)\, \mathrm{d}t
\le
\liminf_{m\to\infty}
\int_0^1 \bar G_m(t)\,\breve F_{\varphi_T}(1-t)\,\mathrm{d}t
\le x_0.
\end{equation}

Next, since $\DB_\phi(\cdot,\breve F_b)$ is convex in its first argument,
\begin{align*}
\DB_\phi(\bar G_m,\breve F_b)
\le
\frac1m\sum_{k=1}^m \DB_\phi(\breve G_{n_k},\breve F_b)
\le \varepsilon,
\quad m\ge1. 
\end{align*}
Hence \eqref{eq:aeconv} yields
\begin{align*}
B_\phi(\breve G^*(t),\breve F_b(t))
\le
\liminf_{m\to\infty} B_\phi(\bar G_m(t),\breve F_b(t))
\quad \text{for a.e. }t\in(0,1). 
\end{align*}
Since the Bregman integrand is non-negative, Fatou's lemma gives
\begin{align}
\DB_\phi(\breve G^*,\breve F_b)
&=
\int_0^1 B_\phi(\breve G^*(t),\breve F_b(t))\, \mathrm{d} t \notag\\
&\le
\int_0^1 \liminf_{m\to\infty} B_\phi(\bar G_m(t),\breve F_b(t))\,\mathrm{d}t \notag\\
&\le
\liminf_{m\to\infty}\int_0^1 B_\phi(\bar G_m(t),\breve F_b(t))\,\mathrm{d}t \notag\\
&=
\liminf_{m\to\infty}\DB_\phi(\bar G_m,\breve F_b)
\le \varepsilon.
\label{eq:Bregmanlimit}
\end{align}
Therefore, $\breve G^*\in \mathcal Q_\varepsilon(x_0)$.

We now show that $\breve G^*$ is optimal. Since $u$ is concave, Jensen's inequality gives, for every $m\ge1$ and every $t\in(0,1)$,
\begin{align*}
u(\bar G_m(t))
=
u\!\left(\frac1m\sum_{k=1}^m \breve G_{n_k}(t)\right)
\ge
\frac1m\sum_{k=1}^m u(\breve G_{n_k}(t)). 
\end{align*}
Integrating over $(0,1)$, we obtain
\begin{equation}\label{eq:Jensen}
\int_0^1 u(\bar G_m(t))\,\mathrm{d}t
\ge
\frac1m\sum_{k=1}^m \int_0^1 u(\breve G_{n_k}(t))\,\mathrm{d}t.
\end{equation}
According to \eqref{eq:maxseq},
\begin{align*}
V-\frac1{n_k}
<
\int_0^1 u(\breve G_{n_k}(t))\,\mathrm{d}t
\le V,
\quad k\ge1. 
\end{align*}
Hence
\begin{align*}
\int_0^1 u(\breve G_{n_k}(t))\,\mathrm{d}t \longrightarrow V
\quad \text{as} \quad  k\to\infty, 
\end{align*}
and therefore, by Ces\`aro mean convergence theorem \citep[c.f. Corollary 1.1 in][]{linero2013toeplitz},
\begin{equation}\label{eq:CesaroV}
\frac1m\sum_{k=1}^m \int_0^1 u(\breve G_{n_k}(t))\,\mathrm{d}t
\longrightarrow V
\quad \text{as} \quad m\to\infty.
\end{equation}
Combining \eqref{eq:Jensen} and \eqref{eq:CesaroV}, we obtain
\begin{equation}\label{eq:liminfV}
\liminf_{m\to\infty}\int_0^1 u(\bar G_m(t))\,\mathrm{d}t \ge V.
\end{equation}

It remains to pass to the limit in the objective. Note that
\begin{align*}
u(0)=u_0:=\lim_{x\downarrow0}u(x)\in\mathbb R, 
\end{align*}
and $u$ is a continuous function on $[0,\infty)$. By \eqref{eq:L1bound},
\begin{align*}
\int_0^1 \bar G_m(t)\,\mathrm{d}t
=
\frac1m\sum_{k=1}^m \int_0^1 \breve G_{n_k}(t)\,\mathrm{d}t
\le C,
\quad m\ge1. 
\end{align*}

If $u$ is bounded above, then $\{u(\bar G_m)\}_{m\ge1}$ is uniformly bounded, hence uniformly integrable. Assume now that $u$ is unbounded above. Define
\begin{align*}
Y_m(t):=u(\bar G_m(t))-u_0\ge0. 
\end{align*}
Since $u$ is strictly increasing and concave on $[0,\infty)$, its inverse
\begin{align*}
u^{-1}:[u_0,\infty)\to[0,\infty) 
\end{align*}
is increasing and convex. Let
\begin{align*}
\Psi(y):=u^{-1}(y+u_0),\quad y\ge0. 
\end{align*}
Then $\Psi$ is increasing and convex. Moreover, by the Inada condition
\begin{align*}
\lim_{x\to\infty}u'(x)=0, 
\end{align*}
concavity implies
\begin{align*}
 \frac{u(x)}{x} \longrightarrow 0
\quad\text{as} \quad x\to    \infty.
\end{align*}
Let $x=\Psi(y)$. Then $y=u(x)-u_0$, and therefore
\begin{align*}
\frac{\Psi(y)}{y}
=
\frac{x}{u(x)-u_0}. 
\end{align*}
Since
\begin{align*}
\frac{u(x)-u_0}{x}
=
\frac{u(x)}{x}-\frac{u_0}{x}
\longrightarrow 0
\quad \text{as} \quad x\to\infty, 
\end{align*}
it follows that
\begin{align*}
\frac{\Psi(y)}{y}\longrightarrow\infty
\quad \text{as} \quad y\to\infty. 
\end{align*}
Furthermore,
\begin{align*}
\int_0^1 \Psi(Y_m(t))\,\mathrm{d} t
=
\int_0^1 u^{-1}(u(\bar G_m(t)))\,\mathrm{d}t
=
\int_0^1 \bar G_m(t)\,\mathrm{d}t
\le C,
\quad m\ge1. 
\end{align*}
By the de la Vall\'ee--Poussin criterion \citep[see Theorem 1.1 in][]{hu2011note}, $(Y_m)_{m\ge1}$ is uniformly integrable. Since
\begin{align*}
 u(\bar G_m)=Y_m+u_0,
\end{align*}
it follows that $\{u(\bar G_m)\}_{m\ge1}$ is uniformly integrable.

Combining with \eqref{eq:aeconv} and continuity of $u$ on $[0,\infty)$, we have
\begin{align*}
 u(\bar G_m(t))\longrightarrow u(\breve G^*(t))
\quad\text{as} \quad m\to\infty \quad \text{for a.e. }t\in(0,1).
\end{align*}
According to Theorem 4.6.3 in 
\cite{durrett2019probability}, we have
\begin{equation}\label{eq:Vitali}
\int_0^1 u(\bar G_m(t))\,\mathrm{d}t
\longrightarrow
\int_0^1 u(\breve G^*(t))\,\mathrm{d}t.
\end{equation}
Combining \eqref{eq:liminfV} and \eqref{eq:Vitali}, we conclude that
\begin{align*}
\int_0^1 u(\breve G^*(t))\,\mathrm{d}t \ge V. 
\end{align*}

Since $\breve G^*\in \mQ^{ae}$, we can choose a non-decreasing left-continuous representative, still denoted by $\breve G^*$, of its $L^1((0,1))$-equivalence class. Then $\breve G^*\in \mQ$. Moreover, the objective functional, the budget functional, and the BW-divergence constraint are all invariant under modifications on sets of Lebesgue measure zero. Therefore, this representative also belongs to $\mathcal Q_\varepsilon(x_0)$ and satisfies
\begin{align*}
\int_0^1 u(\breve G^*(t))\,\mathrm{d}t = V. 
\end{align*}
Hence $\breve G^*$ is an optimal solution of Problem \eqref{Yaari2} in the original admissible class $\mQ$.

Finally, uniqueness follows from the strict concavity of $u$ and the convexity of the feasible set $\mathcal Q_\varepsilon(x_0)$. More precisely, if $\breve G_1,\breve G_2\in\mathcal Q_\varepsilon(x_0)$ are both optimal and differ on a set of positive measure, then for every $\theta\in(0,1)$,
\begin{align*}
\breve G_\theta:=\theta \breve G_1+(1-\theta)\breve G_2
\in \mathcal Q_\varepsilon(x_0), 
\end{align*}
and strict concavity yields
\begin{align*}
u(\breve G_\theta(t))
>
\theta u(\breve G_1(t))+(1-\theta)u(\breve G_2(t)) 
\end{align*}
on a set of positive measure. Integrating gives
\begin{align*}
\int_0^1 u(\breve G_\theta(t))\,\mathrm{d}t
>
\theta \int_0^1 u(\breve G_1(t))\,\mathrm{d}t
+(1-\theta)\int_0^1 u(\breve G_2(t))\,\mathrm{d}t
=V, 
\end{align*}
a contradiction. Hence the optimal solution is unique up to equality almost everywhere in $L^1((0,1))$. Consequently, its non-decreasing left-continuous representative in $\mQ$ is unique.
}
\end{proof}

\subsection{Optimal Quantile Function}

{
As discussed in \Cref{Qae}, we work in $L^1((0,1))$ and consider set $\mathcal{Q}^{ae}$ which is a convex subset of $L^1((0,1))$. Endowed with the $L^1$-norm, $L^1((0,1))$ is a Banach space. This setting allows us to apply tools from convex optimization to problem~\eqref{Yaari2}, provided that the objective and constraint functionals satisfy the required regularity conditions. A concise review of the relevant results is given in \Cref{convex opti}.
}

In this section, {we work under Assumptions \ref{assumption_1} and \ref{asm-well-posed}} and rewrite problem \eqref{Yaari2} as follows:
\begin{align}\label{new_problem}
\inf_{\Ginv \in \mQ^{ae}} \left\{\int_{0}^{1} -u\big(\Ginv(t)\big)\d t \ \Big| \ 
\cost(\Ginv)- x_0 \leq 0, \ \DB_\phi (\Ginv, \Fbenchinv)-\varepsilon \leq 0 \right\}.
\end{align}

{
In the proof of \Cref{Exist} (Part 1), we established the convexity of the feasible set. Moreover, the objective function is convex, the budget constraint is linear, and the BW-distance constraint is convex in its first argument. To apply \Cref{main_optimal_th}, we extend the objective and constraint functions to the whole space $L^1((0,1))$.

Let $\bar{\mathbb{R}} := \mathbb{R} \cup \{+\infty\}$, and define the objective functional $\bar{\mathfrak{u}}(\cdot): L^1((0,1)) \to \bar{\mathbb{R}}$ by
\begin{align*}
\bar{\mathfrak{u}}(\Ginv):=
\begin{cases}
-\int_0^1 u(\Ginv(t))\, \mathrm{d}t, & \Ginv\in\mQ^{ae},\\
+\infty, & \Ginv\notin\mQ^{ae}.
\end{cases} 
\end{align*}
Next, define the constraint functionals $\bar{\cost}(\cdot),\bar{\DB}_{\phi,\Fbenchinv}(\cdot)$ by
\begin{align*}
\bar{\cost}(\Ginv):=
\begin{cases}
\cost(\Ginv)-x_0, & \text{if }\cost(\Ginv)<\infty,\\
+\infty, & \text{otherwise},
\end{cases} 
\end{align*}
and
\begin{align*}
\bar{\DB}_{\phi,\Fbenchinv}(\Ginv):=
\begin{cases}
\DB_\phi(\Ginv,\Fbenchinv)-\varepsilon,
& \text{if }\DB_\phi(\Ginv,\Fbenchinv)<\infty,\\
+\infty, & \text{otherwise}.
\end{cases} 
\end{align*}

Hence problem \eqref{new_problem} is equivalent to
\begin{equation}\label{new_problem_ext}
\inf_{\Ginv \in L^{1}((0,1))} \bigg\{ \bar{\mathfrak{u}}(\Ginv) \ \Big| \ 
\bar{\cost}(\Ginv) \le 0,\ 
\bar{\DB}_{\phi,\Fbenchinv}(\Ginv) \le 0 \bigg\}.
\end{equation}

\begin{proposition}\label{existlag}
Suppose that Assumptions~\ref{assumption_1} and~\ref{asm-well-posed} hold.
If $\Ginv^*$ is  a solution to problem \eqref{new_problem_ext}, there exist Lagrange multipliers $\lambda^*,\mu^*\ge0$ such that
\begin{align*}
\lambda^*\bar{\cost}(\Ginv^*)=0,
\quad \text{and} \quad
\mu^*\bar{\DB}_{\phi,\Fbenchinv}(\Ginv^*)=0, 
\end{align*}
and 
\begin{align}\label{theLagrangian}
\Ginv^*
 \in  \argmin_{\Ginv\in L^{1}((0,1))}
\bigg\{ \bar{\mathfrak{u}}(\Ginv) + \lambda^*  
\bar{\cost}(\Ginv) + \mu^*
\bar{\DB}_{\phi,\Fbenchinv}(\Ginv) \bigg\}.
\end{align}
\end{proposition}

\begin{proof}
    We verify the conditions of \Cref{main_optimal_th}. First, $X=L^{1}((0,1))$, endowed with the $L^1$-norm, is a Banach space and hence a separated locally convex space.

Second, we show that $\bar{\mathfrak u}$ is proper (see \Cref{app_def}). 
Taking $\Ginv(t)=0$ for $t \in (0,1)$, we have $\Ginv\in\mQ^{ae}$ and
\begin{align*}
\bar{\mathfrak u}(\Ginv)
=
-\int_0^1 u(0)\, \mathrm{d}t
=
-u(0)<+\infty. 
\end{align*}
Hence $\operatorname{dom}\bar{\mathfrak u}\neq\emptyset$.
Moreover, let $\breve G\in\mQ^{ae}$. Then $\breve G\ge0$ a.e. and $\breve G\in L^1((0,1))$. Since \(u\) is increasing and $u(0)\in\mathbb R$, we have
\begin{align*}
u(\breve G(t))\ge u(0)
\quad \text{for a.e. }t\in(0,1), 
\end{align*}
and therefore $\int_0^1 u(\breve G(t)) \mathrm{d}t>-\infty$.
On the other hand, since $u$ is concave on $[0,\infty)$, for any fixed $a>0$ and all $x\ge0$,
\begin{align*}
u(x)\le u(a)+u'(a)(x-a). 
\end{align*}
Then there exist constants $C_1\in\mathbb R$ and $C_2>0$ such that
\begin{align*}
u(x)\le C_1+C_2x,\quad x\ge0. 
\end{align*}
Therefore,
\begin{align*}
u(\breve G(t))\le C_1+C_2\breve G(t)
\quad \text{for a.e. }t\in(0,1). 
\end{align*}
Since $\breve G\in L^1((0,1))$, it follows that $\int_0^1 u(\breve G(t)) \mathrm{d}t<+\infty$.
Consequently,
\begin{align} \label{utility_finite}
- \infty < \int_0^1 u(\breve G(t))\mathrm{d}t < +\infty
\quad \text{for every } \breve G\in\mQ^{ae}. 
\end{align}
Thus $\bar{\mathfrak u}$ never takes the value $-\infty$ and $\bar{\mathfrak u}$ is proper. Convexity of $\bar{\mathfrak u}$ follows from the convexity of $-u$ and the convexity of $\mQ^{ae}$. { Hence $\bar{\mathfrak u}(\cdot)\in\Lambda(X)$, the class of proper convex functions on $X$, see \Cref{app_def}.} Similarly, by \Cref{assumption_1} and the convexity properties, we have $\bar{\cost}(\cdot), \bar{\DB}_{\phi,\Fbenchinv}(\cdot)\in \Lambda(X)$.

Finally, by Assumption \ref{assumption_1}, there exists $\Ginv \in\mQ$, identified with an element of $\mQ^{ae}$, such that
\begin{align*}
\bar{\cost}(\Ginv)<0
\quad\text{and}\quad
\bar{\DB}_{\phi,\Fbenchinv}(\Ginv)<0. 
\end{align*}
Thus the Slater condition in \Cref{main_optimal_th} holds.

Let $\Ginv^*$ be the solution of \eqref{new_problem_ext}. By \Cref{main_optimal_th}, there exist Lagrange multipliers $\lambda^*,\mu^*\ge0$ such that
\begin{align*}
\lambda^*\bar{\cost}(\Ginv^*)=0,
\quad \text{and} \quad
\mu^*\bar{\DB}_{\phi,\Fbenchinv}(\Ginv^*)=0, 
\end{align*}
and $0\in
\partial  \mathcal L_{\lambda^*,\mu^*} (\Ginv^*)$, where
\begin{align*}
 \mathcal L_{\lambda^*,\mu^*} := \bar{\mathfrak u}
+\lambda^*\bar{\cost}
+\mu^*\bar{\DB}_{\phi,\Fbenchinv}.
\end{align*}

Since $\bar{\mathfrak u},\bar{\cost},\bar{\DB}_{\phi,\Fbenchinv}\in\Lambda(X)$, $\lambda^*,\mu^*\ge0$, and $\Ginv^*$ is feasible for \eqref{new_problem_ext}, the functional $\mathcal L_{\lambda^*,\mu^*}$ is proper and convex, i.e., $\mathcal L_{\lambda^*,\mu^*}(\cdot) \in \Lambda(X)$.  \Cref{app_partial} yields
\begin{align*}
\Ginv^*
& \in \argmin_{\Ginv\in L^{1}((0,1))}
 \mathcal L_{\lambda^*,\mu^*}(\Ginv)   \\
& = \argmin_{\Ginv\in L^{1}((0,1))}
\bigg\{ \bar{\mathfrak{u}}(\Ginv) + \lambda^*  
\bar{\cost}(\Ginv) + \mu^*
\bar{\DB}_{\phi,\Fbenchinv}(\Ginv) \bigg\}. 
\end{align*}
\end{proof}

\begin{corollary}[At least one binding constraint]
\label{cor_binding}
Suppose that Assumptions~\ref{assumption_1} and~\ref{asm-well-posed} hold, and let
$\Ginv^*$ be a solution of problem \eqref{new_problem_ext}. Then at least one of the two constraints is binding at $\Ginv^*$, that is,
\begin{align*}
\cost(\Ginv^*)=x_0
\quad\text{or}\quad
\DB_\phi(\Ginv^*,\Fbenchinv)=\varepsilon. 
\end{align*}
\end{corollary}

\begin{proof}
By \Cref{existlag}, there exist $\lambda^*,\mu^*\ge0$ such that
\begin{align*}
\lambda^*\bar{\cost}(\Ginv^*)=0,
\quad
\mu^*\bar{\DB}_{\phi,\Fbenchinv}(\Ginv^*)=0, 
\end{align*}
and
\begin{align*}
\Ginv^*
\in
\argmin_{\Ginv\in L^1((0,1))}
\mathcal L_{\lambda^*,\mu^*}(\Ginv). 
\end{align*}

We first show that $(\lambda^*,\mu^*)\neq(0,0)$. Suppose, to the contrary, that
$\lambda^*=\mu^*=0$. Then $\mathcal L_{\lambda^*,\mu^*}=\bar{\mathfrak u}$
and hence
\begin{align*}
\Ginv^*
\in
\argmin_{\Ginv\in L^1((0,1))}
\bar{\mathfrak u}(\Ginv). 
\end{align*}
For any $c>0$, define
\begin{align*}
\Ginv_c(t):=\Ginv^*(t)+c,\quad t\in(0,1). 
\end{align*}
Since $\Ginv^*\in\mQ^{ae}$, we have $\Ginv_c\in\mQ^{ae}$. Moreover, by the strict monotonicity of $u$,
\begin{align}
u(\Ginv_c(t))>u(\Ginv^*(t))
\quad\text{for a.e. }t\in(0,1). 
\end{align}
Therefore,
\begin{align*}
\bar{\mathfrak u}(\Ginv_c)
=
-\int_0^1 u(\Ginv_c(t))\,dt
<
-\int_0^1 u(\Ginv^*(t))\,dt
=
\bar{\mathfrak u}(\Ginv^*), 
\end{align*}
which contradicts the minimality of $\Ginv^*$ for $\bar{\mathfrak u}$. Hence $(\lambda^*,\mu^*)\neq(0,0)$.

Now suppose that both constraints are slack
\begin{align*}
\cost(\Ginv^*)<x_0,
\quad
\DB_\phi(\Ginv^*,\Fbenchinv)<\varepsilon. 
\end{align*}
Equivalently,
\begin{align*}
\bar{\cost}(\Ginv^*)<0,
\quad
\bar{\DB}_{\phi,\Fbenchinv}(\Ginv^*)<0. 
\end{align*}
Since $\lambda^*,\mu^*\ge0$, the complementary slackness conditions imply $(\lambda^*,\mu^*)=(0,0)$,
which contradicts $(\lambda^*,\mu^*)\neq(0,0)$. Therefore, both constraints cannot be slack simultaneously. Hence at least one constraint is binding.
\end{proof}

We now express the Lagrangian of \eqref{theLagrangian} on its effective domain.
For any \(\Ginv\in\operatorname{dom}\mathcal L_{\lambda^*,\mu^*}\), we have
\(\Ginv\in\mQ^{ae}\), and the extended-value functionals reduce to their ordinary integral forms. Hence
\begin{align}
\mathcal L_{\lambda^*,\mu^*}(\Ginv) =& \int_{0}^{1} - u(\Ginv(t))\d t  + \lambda^{*} \cost(\Ginv)  + \mu^{*}  \DB_\phi(\Ginv, \Fbenchinv) -\lambda^*x_0-\mu^*\varepsilon \,.
\label{ordinary_int} 
\end{align}
If \(\Ginv\notin\operatorname{dom}\mathcal L_{\lambda^*,\mu^*}\), then $\mathcal L_{\lambda^*,\mu^*}(\Ginv)=+\infty$.
Therefore, the minimization of the extended Lagrangian over $X=L^1((0,1))$ is equivalent to minimizing \eqref{ordinary_int} over the effective domain of $\mathcal L_{\lambda^*,\mu^*}$.

Thus, for general $\lambda,\mu\ge0$, we introduce the pointwise candidate associated with the Lagrangian integrand. For fixed $t\in(0,1)$, the relevant ordinary integrand is
\begin{align}
-u(\Ginv(t))
+\lambda \Ginv(t)\breve{F}_{\varphi_T}(1-t)
+\mu B_\phi(\Ginv(t),\Fbenchinv(t)). 
\end{align}
Since
\begin{align*}
B_\phi(y,\Fbenchinv(t))
=
\phi(y)-\phi(\Fbenchinv(t))-\phi'(\Fbenchinv(t))(y-\Fbenchinv(t)), 
\end{align*}
the terms independent of $y$ do not affect the pointwise minimizer. Hence, for fixed $t\in(0,1)$, we consider
\begin{equation}
    h_{t}(y) :=\  m(y) + \lambda y \breve{F}_{\varphi_T}(1-t) - \mu \phi'\big(\Fbenchinv(t)  \big) y\,, \label{integrand}
\end{equation}
where we set $m(y):=-u(y)+\mu \phi(y)$.
Then, we consider the pointwise optimiser of $h_t(\cdot)$, that is for each $t\in(0,1)$ we solve 
\begin{equation}\label{G*} 
 \argmin_{y\in[0, \infty)} \; h_t(y)\,.    
 \end{equation}
The pointwise minimizer, when finite, gives a candidate for the Lagrangian minimizer.

\begin{proposition}[Pointwise candidate]
\label{mini}
For fixed $\lambda,\mu\ge0$, define the extended-valued function $\Ginv_{\lambda,\mu}:(0,1) \to [0,\infty]$ by
\begin{align}\label{opt_sol}
\Ginv_{\lambda,\mu}(t)
:=
\inf\left\{
y\ge0
\ \middle|\ 
\frac{\partial}{\partial y}h_t(y)\ge0
\right\}, 
\end{align}
with the convention $\inf\emptyset:=+\infty$. If
$\Ginv_{\lambda,\mu}(t)<+\infty$, then
$\Ginv_{\lambda,\mu}(t)$ is the unique minimizer of $h_t$ over
$[0,\infty)$. If $\Ginv_{\lambda,\mu}(t)=+\infty$, then $h_t$ has no finite minimizer over $[0,\infty)$, and its infimum is approached as $y\to\infty$.
\end{proposition}

\begin{proof}
Fix $t\in(0,1)$. Since $u$ is strictly concave and $\phi$ is convex, the function
\begin{align*}
h_t(y)
=
-u(y)
+\lambda y\breve F_{\varphi_T}(1-t)
+\mu\bigl(\phi(y)-\phi'(\Fbenchinv(t))y\bigr),
\quad y\ge0, 
\end{align*}
is strictly convex in $y$. Moreover, $h_t$ is differentiable on $(0,\infty)$, and
\begin{align*}
\frac{\partial}{\partial y}h_t(y)
=
-u'(y)
+\lambda \breve F_{\varphi_T}(1-t)
+\mu\bigl(\phi'(y)-\phi'(\Fbenchinv(t))\bigr) 
\end{align*}
is continuous and strictly increasing in $y$.

By the Inada condition $\lim_{y \to 0}u'(y)=\infty$, we have
\begin{align*}
\lim_{y\downarrow0}\frac{\partial}{\partial y}h_t(y)<0. 
\end{align*}
Then the minimizer, if finite, cannot occur at $0$.

If there exists $\hat y\in(0,\infty)$ such that
\begin{align*}
\frac{\partial}{\partial y}h_t(\hat y)=0, 
\end{align*}
then the strict convexity of $h_t$ implies that $\hat y$ is the unique minimizer of $h_t$ over $[0,\infty)$. In this case, by the definition of $\Ginv_{\lambda,\mu}(t)$ and the continuity and monotonicity of $\frac{\partial}{\partial y}h_t$, we have
\begin{align*}
\Ginv_{\lambda,\mu}(t)=\hat y. 
\end{align*}

If no such finite $\hat y$ exists, then, since $\frac{\partial}{\partial y}h_t$ is increasing and negative near zero, either $\frac{\partial}{\partial y}h_t(y)<0$ for all $y>0$. Hence $h_t$ is decreasing on $[0,\infty)$ and has no finite minimizer. In this case the set
\begin{align*}
\left\{
y\ge0
\ \middle|\ 
\frac{\partial}{\partial y}h_t(y)\ge0
\right\} 
\end{align*}
is empty, and by convention $\Ginv_{\lambda,\mu}(t)=+\infty$.
This value records that the pointwise problem has no finite minimizer.
\end{proof}

The next proposition establishes the monotonicity of the pointwise candidate $\Ginv_{\lambda, \mu}$.

\begin{proposition}\label{non-decreasing}
For all $\lambda , \mu \ge 0$, $ \Ginv_{\lambda,\mu}(\cdot)$ is non-decreasing. 
\end{proposition}
\begin{proof}
First, we show that if there exists a $t_1 \in (0,1)$ such that $\Ginv_{\lambda,\mu}(t_1)=+\infty$, then $\Ginv_{\lambda,\mu}(t_2)=+\infty$ for all $t_2 \geq t_1$. From the proof of \Cref{mini}, we know that $\Ginv_{\lambda,\mu}(t_1)=+\infty$ is equivalent to
\begin{align*}
     \frac{\partial}{\partial y }h_{t_1}(y) = m'(y) + \lambda  \breve{F}_{\varphi_T}(1-t_1) - \mu \phi'\big(\Fbenchinv(t_1)  \big)  <0,
 \end{align*}
 for all $y \in (0, \infty)$. Since $\breve{F}_{\varphi_T}(1-t)$ is non-increasing and $\phi'\big(\Fbenchinv(t)  \big)$ is non-decreasing with respect to $t$, we also obtain that 
 \begin{align*}
     \frac{\partial}{\partial y }h_{t_2}(y)
     =
     m'(y) + \lambda  \breve{F}_{\varphi_T}(1-t_2) - \mu \phi'\big(\Fbenchinv(t_2)  \big)  <0,
 \end{align*}
for all $y \in (0,\infty)$, and thus $\Ginv_{\lambda,\mu}(t_2)=+\infty$. Thus, we conclude that $\Ginv_{\lambda,\mu}(t_1)=+\infty$, implies that $\Ginv_{\lambda,\mu}(t_2)=+\infty$ for all $t_2 \geq t_1$.

	Next, we  verify that $\Ginv_{\lambda,\mu}(t)$ is non-decreasing within the set $\{ t \in (0,1) \ | \ \Ginv_{\lambda,\mu}(t) < + \infty  \}$.
    By \Cref{mini}, these finite values are the unique pointwise minimizers. For $t_1 < t_2$,  we have  
  \begin{align*}
  &\tfrac{\partial}{\partial y} h_{t_1}(\Ginv_{\lambda,\mu}(t_1)) = m'(\Ginv_{\lambda,\mu}(t_1)) + \lambda  \breve{F}_{\varphi_T}(1-t_1) - \mu \phi'\big(\Fbenchinv(t_1)  \big)  =0\,, 
         \quad \text{and}\\
  &\tfrac{\partial}{\partial y} h_{t_2}(\Ginv_{\lambda,\mu}(t_2))
  =m'(\Ginv_{\lambda,\mu}(t_2)) + \lambda  \breve{F}_{\varphi_T}(1-t_2) - \mu \phi'\big(\Fbenchinv(t_2)  \big)  =0.
 \end{align*}
    Subtracting the two equations and rearranging, we have
\begin{equation}\label{eq:pf-monotone}
\begin{split}
        m'(\Ginv_{\lambda,\mu}(t_2)) -m'(\Ginv_{\lambda,\mu}(t_1))= 
    -\lambda \left( \breve{F}_{\varphi_T}(1-t_2) - \breve{F}_{\varphi_T}(1-t_1) \right) 
    \\
    + \mu \left( \phi'\big(\Fbenchinv(t_2)  \big) - \phi'\big(\Fbenchinv(t_1)  \big)  \right)  \,.
\end{split}
\end{equation}
Note that the right hand side of \eqref{eq:pf-monotone} is non-negative by the monotonicity of the quantile functions $\breve{F}_{\varphi_T}(\cdot)$ and $\Fbenchinv(\cdot)$. Thus 
\begin{equation*}
    m'(\Ginv_{\lambda,\mu}(t_2)) -m'(\Ginv_{\lambda,\mu}(t_1)) \ge 0\,,
\end{equation*}
which implies that $\Ginv_{\lambda,\mu}(t_2)\ge \Ginv_{\lambda,\mu}(t_1) $, since $m'(\cdot)$ is increasing. 
Hence, $\Ginv_{\lambda,\mu}$ is indeed non-decreasing.
\end{proof}

Finally, we identify the optimal quantile function with the pointwise candidate associated with the optimal Lagrange multipliers.

\begin{theorem}[Optimal quantile function]
\label{mainth}
Let Assumptions \ref{assumption_1} and \ref{asm-well-posed} hold, and let
$\breve F^*\in\mQ^{ae}$ be the unique solution of problem \eqref{new_problem_ext}. 
Let $\lambda^*,\mu^*\ge0$ be the Lagrange multipliers, and define
$\Ginv_{\lambda^*,\mu^*}$ by \eqref{opt_sol}. Suppose that $\Ginv_{\lambda^*,\mu^*}\in
\operatorname{dom}\mathcal L_{\lambda^*,\mu^*}$.
Then
\begin{align*}
\breve F^*(t)
=
\Ginv_{\lambda^*,\mu^*}(t)
\quad\text{for a.e. }t\in(0,1). 
\end{align*}
Moreover, the multipliers satisfy the complementary slackness conditions
\begin{align*}
\lambda^*\bigl(\cost(\breve F^*)-x_0\bigr)=0,
\qquad
\mu^*\bigl(\DB_\phi(\breve F^*,\Fbenchinv)-\varepsilon\bigr)=0. 
\end{align*}
Consequently, the non-decreasing left-continuous representative of this $L^1$-equivalence class is the optimal quantile function for \eqref{Yaari2}.
\end{theorem}

\begin{proof}
By \Cref{Exist}, problem \eqref{new_problem_ext} admits a unique solution
$\breve F^*\in\mQ^{ae}$. According to \Cref{existlag}, there exist
$\lambda^*,\mu^*\ge0$ such that
\begin{align*}
\lambda^*\bar{\cost}(\Ginv^*)=0,
\qquad 
\mu^*\bar{\DB}_{\phi,\Fbenchinv}(\Ginv^*)=0, 
\end{align*}
and
\begin{align*}
\breve F^*
\in
\argmin_{\Ginv\in L^1((0,1))}
\mathcal L_{\lambda^*,\mu^*}(\Ginv), 
\end{align*}
where $\mathcal L_{\lambda^*,\mu^*}=\bar{\mathfrak u}
+\lambda^*\bar{\cost}
+\mu^*\bar{\DB}_{\phi,\Fbenchinv}$. Since $\breve F^*$ is feasible, the complementary slackness conditions are equivalently
\begin{align*}
\lambda^*\bigl(\cost(\breve F^*)-x_0\bigr)=0,
\qquad
\mu^*\bigl(\DB_\phi(\breve F^*,\Fbenchinv)-\varepsilon\bigr)=0. 
\end{align*}

For $t\in(0,1)$ and $y\ge0$, define
\begin{align}
\ell_{\lambda^*,\mu^*}(t,y)
:=
-u(y)
+\lambda^*y\breve F_{\varphi_T}(1-t)
+\mu^*B_\phi(y,\Fbenchinv(t)). 
\end{align}
On $\operatorname{dom}\mathcal L_{\lambda^*,\mu^*}$, the extended-value functionals reduce to their ordinary integral forms. Thus, for every
$\Ginv\in\operatorname{dom}\mathcal L_{\lambda^*,\mu^*}$,
\begin{align*}
\mathcal L_{\lambda^*,\mu^*}(\Ginv)
=
\int_0^1
\ell_{\lambda^*,\mu^*}(t,\Ginv(t))\,\d t
-\lambda^*x_0-\mu^*\varepsilon. 
\end{align*}

For each fixed $t \in (0,1)$, the functions $y \mapsto \ell_{\lambda^*,\mu^*}(t,y)$ and $
y \mapsto h_t^*(y)$ (given in \eqref{integrand} with multipliers $\lambda^*$ and $\mu^*$) differ only by a term independent of $y$. Hence they have the same pointwise minimizers. By \Cref{mini}, whenever the pointwise candidate is finite, it is the unique minimizer of $h_t$, and hence also of $\ell_{\lambda^*,\mu^*}(t,\cdot)$.

By the admissibility assumption, $\Ginv_{\lambda^*,\mu^*}
\in
\operatorname{dom}\mathcal L_{\lambda^*,\mu^*}$.
In particular, $\Ginv_{\lambda^*,\mu^*}$ is finite a.e. and is an admissible competitor for the Lagrangian minimization problem. Therefore, for every $\Ginv\in\operatorname{dom}\mathcal L_{\lambda^*,\mu^*}$,
\begin{align*}
\ell_{\lambda^*,\mu^*}
\bigl(t,\Ginv_{\lambda^*,\mu^*}(t)\bigr)
\le
\ell_{\lambda^*,\mu^*}
\bigl(t,\Ginv(t)\bigr)
\quad\text{for a.e. }t\in(0,1). 
\end{align*}
Since both functions belong to the effective domain, the corresponding integrals are well defined. Integrating the above inequality gives
\begin{align*}
\mathcal L_{\lambda^*,\mu^*}
(\Ginv_{\lambda^*,\mu^*})
\le
\mathcal L_{\lambda^*,\mu^*}(\Ginv). 
\end{align*}
For $\Ginv\notin\operatorname{dom}\mathcal L_{\lambda^*,\mu^*}$, the inequality is trivial because $\mathcal L_{\lambda^*,\mu^*}(\Ginv)=+\infty$. Hence
\begin{align*}
\Ginv_{\lambda^*,\mu^*}
\in
\argmin_{\Ginv\in L^1((0,1))}
\mathcal L_{\lambda^*,\mu^*}(\Ginv). 
\end{align*}

Since $\breve F^*$ is also a minimizer of the same Lagrangian, suppose that
\begin{align*}
\mathcal I:=
\bigl\{
t\in(0,1):
\breve F^*(t)\neq \Ginv_{\lambda^*,\mu^*}(t)
\bigr\} 
\end{align*}
has positive measure. By the strict convexity of $h_t$, the finite pointwise minimizer is unique. Hence, for a.e. $t\in \mathcal I$,
\begin{align*}
\ell_{\lambda^*,\mu^*}
\bigl(t,\Ginv_{\lambda^*,\mu^*}(t)\bigr)
<
\ell_{\lambda^*,\mu^*}
\bigl(t,\breve F^*(t)\bigr). 
\end{align*}
Integrating yields
\begin{align*}
\mathcal L_{\lambda^*,\mu^*}
(\Ginv_{\lambda^*,\mu^*})
<
\mathcal L_{\lambda^*,\mu^*}(\breve F^*), 
\end{align*}
which contradicts the Lagrangian minimality of $\breve F^*$. Therefore,
\begin{align*}
\breve F^*(t)
=
\Ginv_{\lambda^*,\mu^*}(t)
\quad\text{for a.e. }t\in(0,1). 
\end{align*}

Finally, since this $L^1$-equivalence class belongs to $\mQ^{ae}$, we choose its non-decreasing left-continuous representative. This representative belongs to $\mQ$ and gives the optimal quantile function for \eqref{Yaari2}.
\end{proof}

Next, we show that the optimal quantile function has a simple representation. 

\begin{corollary}
    Under the assumptions of \Cref{mainth}, the unique quantile function that solves \eqref{Yaari2} has representation
    \begin{equation*}
        \breve{F}^*(t)
        =
       (m')^{-1}\Big( -\lambda^*  \breve{F}_{\varphi_T}(1-t) + \mu^* \phi'\big(\Fbenchinv(t)  \big)\Big)\,, \quad\text{for a.e. }t\in(0,1),
    \end{equation*}
    where $(m')^{-1}$ is the inverse of $m'(y) = -u'(y) + \mu^* \phi'(y)$, and $\lambda^*$ and $\mu^*$ are given in \Cref{mainth}.
\end{corollary}

\begin{proof}
    By \Cref{mainth},  $\breve{F}^*(t) =\Ginv_{\lambda^{*},\mu^{*}}(t)$, for a.e. $t \in (0,1)$. Moreover, the admissibility condition implies that $\Ginv_{\lambda^*,\mu^*}$ is finite a.e. Thus, from \Cref{mini} we have that for each $t \in (0,1)$, $\Ginv_{\lambda^{*},\mu^{*}}(t)$ is the solution to
    \begin{equation}\label{eq:equation-F-star}
         \frac{\partial}{\partial y}\;h^*_t(y) 
         =
           m'(y)+ \lambda^{*} \breve{F}_{\varphi_T}(1-t)  -\mu^{*}  \phi'\big(\Fbenchinv(t)  \big)  = 0\,,
    \end{equation}
    where $h^*_t(\cdot)$ is given in \eqref{integrand} with multipliers $\lambda^*$ and $\mu^*$. Solving \eqref{eq:equation-F-star} for $y$ concludes the proof.
\end{proof}

Combining Theorems \ref{Exist} and \ref{mainth} with \Cref{cor_binding}, we observe that there are three possible cases for the optimal quantile function:
	\begin{enumerate}[label = $\roman*)$]
		\item Both constraints are binding, i.e., $\lambda^*,\mu^*>0$. Then, $\breve{F}^{*}(t) =  \Ginv_{\lambda^{*},\mu^{*}}(t)$ and 
		\begin{align*}
			  \cost\big( \breve{F}^{*}\big)  = x_0  \quad \text{and} \quad
		 \DB_\phi\big(\breve{F}^*, \Fbenchinv\big) =\varepsilon\,.
		\end{align*}
		\item Only the budget constraint is binding, i.e., $\lambda^*>0$ and $\mu^*=0$. Then, $\breve{F}^{*}(t) =  \Ginv_{\lambda^{*},0}(t)$ and 
		\begin{align*}
			\cost\big( \breve{F}^{*}\big) = x_0 \,.
	   \end{align*}
	   \item Only the BW constraint is binding, i.e., $\lambda^*=0$ and $\mu^*>0$. Then, $\breve{F}^{*}(t) =  \Ginv_{0,\mu^{*}}(t)$ and 
       \begin{align*}
		\DB_\phi\big(\breve{F}^*, \Fbenchinv\big) = \varepsilon\,.
   \end{align*}
\end{enumerate}

}

\section{Examples}\label{sec:4}

In this section we illustrate the theoretical results. We derive optimal payoffs and evaluate how these depend on the choice of the BW divergence, that is, on how investors reward positive deviations and penalize negative ones. We first lay out the market model and then discuss how to choose the tolerance level $\epsilon$ and the Bregman generating function $\phi$.  
\subsection{Market model}
We provide examples within the Geometric Brownian Motion (GBM) market framework. In this market, there is a stock $(S_t)_{t\in [0,T]}$ and a risk-free asset $(B_t)_{t \in [0,T]}$ satisfying the stochastic differential equations 
\begin{align*}
\d S_t =  \mu_{s} S_t \d t + \sigma_s S_t \d W_t, \quad \text{and} \quad \d B_t = rB_t \d t, 
\end{align*}
respectively, where $(W_t)_{t \in [0,T]}$ is a standard Brownian motion, and the stock has drift $\mu_s\in\mathbb{R}$ and instantaneous volatility $\sigma_s>0$. Under the no-arbitrage condition the state price density $\varphi_T$ is uniquely given by
\begin{equation*}
	\varphi_T=
 e^{- \left(r+\frac{1}{2}\theta^2  \right)T - \theta W_{T} },
\end{equation*}
in which $\theta=\frac{\mu_{s}-r}{\sigma_s}$. Hence, $\varphi_T$ follows a log-normal distribution with mean $\mu_{\varphi}:=-(r+ \frac{1}{2}\theta^2)T$ and standard deviation $\sigma_\varphi := \theta\sqrt{T}$, denoted as $\varphi_T \sim \text{LN}(\mu_{\varphi}, \sigma_\varphi^2)$. Its cdf and quantile function are 
\begin{equation*}
	F_{\varphi _{T}}(x)= \Phi \left(\frac{\ln x-\mu_{\varphi}}{\sigma_{\varphi}} \right) \quad \text{and} \quad  \breve{F}_{\varphi_{T}}(t)=e^{\mu_{\varphi}+\sigma_{\varphi} \breve{\Phi}(t)}, 
	\end{equation*}
	where $\Phi(\cdot)$ and $\breve{\Phi}(\cdot)$ are the cdf and quantile function of the standard normal distribution, respectively.

In all our numerical experiments,  we assume the following parameters for the market model
\begin{equation*}
  T=5\,,\quad r=0\,, \quad \mu_s=0.05\, , \quad \sigma_s=0.1\,, 
  \quad \text{and} \quad S_0=1\,.
   \end{equation*}

{Furthermore, we assume that the investor makes their choices using a Constant Relative Risk Aversion (CRRA) utility function
\begin{align*}
	  u(x)=\frac{x^{1-\gamma}-1}{1-\gamma}  \quad 0<\gamma<1,
\end{align*}
in which the parameter $\gamma$ reflects their degree of risk aversion. In the numerical  experiments we choose the parameter values $\gamma=0.5$ and $\gamma=0.8$. Moreover, without any loss of generality, we assume that the investor has a budget $x_0=1$. }

\subsection{Choice of tolerance level $\varepsilon$}\label{choice_BW}
When determining their optimal payoff as a solution to problem \eqref{Yaari-o}, the investor needs to choose the benchmark quantile function $\Fbenchinv$, the tolerance level $\varepsilon$, and the Bregman generating function $\phi$ that determines how deviations from the benchmark $\Fbenchinv$ are assessed.

{
For a given auxiliary budget level $\bar{x}_0>0$, we consider the smallest BW distance attainable under this budget:
\begin{equation}
  \min_{c(X_T)\le \bar{x}_0}
  \DB_\phi(\breve F_{X_T},\Fbenchinv).
\end{equation}
Since $\DB_\phi(\breve F_{X_T},\Fbenchinv)$ depends only on the distribution of $X_T$, and cost-efficiency yields the least expensive payoff among all payoffs with the same distribution; see \cite[Corollary~2]{bernard2014explicit}, it is sufficient to restrict attention to payoffs that are anti-monotonic with $\varphi_T$. This leads to the quantile problem
\begin{equation}\label{eq:min-BW}
    \min_{\Ginv\in\mathcal Q,\ \cost(\Ginv)\le \bar{x}_0}
    \DB_\phi(\Ginv,\Fbenchinv).
\end{equation}
We denote the value of \eqref{eq:min-BW} by $\varepsilon_{\min}(\bar{x}_0)$.

The value $\varepsilon_{\min}(\bar{x}_0)$ provides a natural lower bound for the tolerance level. In problem~\eqref{Yaari2}, if the initial budget satisfies $x_0>\bar{x}_0$, then every choice $\varepsilon>\varepsilon_{\min}(\bar{x}_0)$
satisfies the Slater-type feasibility condition in \Cref{assumption_1}. Indeed, any minimizer of \eqref{eq:min-BW} has cost at most $\bar{x}_0<x_0$, while its BW distance is equal to $\varepsilon_{\min}(\bar{x}_0)<\varepsilon$. Together with \Cref{asm-well-posed}, this yields the existence of a solution to problem~\eqref{Yaari2}.

The following proposition gives the expression for $\varepsilon_{\min}(\bar{x}_0)$.

\begin{proposition}\label{dis_min}
    For a budget $\bar{x}_0$ and a benchmark $\Fbenchinv$, the minimal BW distance $\varepsilon_{\min}(\bar{x}_0)$ is given by
    \begin{align*}
\varepsilon_{\min}(\bar{x}_0) =
    \begin{cases}
        0 \qquad & \text{if} \qquad \cost(\Fbenchinv) \le \bar{x}_0\\[0.5em]
        \DB_\phi (\Ginv^*, \Fbenchinv) \quad & \text{if} \qquad \cost(\Fbenchinv) > \bar{x}_0\,,
    \end{cases}
    \end{align*}
    where
    \begin{align*}
    \Ginv^*(t)=(\phi')^{-1}\left(\phi'(\Fbenchinv(t))-\eta^*\breve F_{\varphi_T}(1-t)\right),  
    \end{align*}
    and $\eta^* > 0$ is the unique solution to $\cost\big( \Ginv^{*}\big) = \bar{x}_0$.
\end{proposition}

\begin{proof}
    If $\cost(\Fbenchinv) \le \bar{x}_0$, then $\ep_{\min}(\bar{x}_0) = 0$. To see this note that the Bregman divergence $B_\phi(z_1, z_2)$ is zero, if and only if $z_1 = z_2$, thus $\DB_\phi(\Ginv, \breve{F}) = 0$ if and only if $\Ginv(u) = \breve{F}(u)$ for all $u\in(0,1)$. 
    
    For the second case, assume that $\cost(\Fbenchinv) > \bar{x}_0$. 
    Then the associated Lagrangian to \eqref{eq:min-BW} with  Lagrange multiplier $\eta \geq 0$ is
        \begin{align*}
        \int_{0}^{1}\phi \big(\Ginv(t) \big)  - \phi \big(\Fbenchinv(t) \big)  - \phi'\big(\Fbenchinv(t)  \big) \Big(\Ginv(t) -\Fbenchinv(t) \Big) \d t + \eta \Big(   \int_{0}^{1} \Ginv(t) \breve{F}_{\varphi_{T}}(1-t)\d t - \bar{x}_0 \Big)\,.   
        \end{align*}
        The remainder of the proof follows using similar arguments as in Propositions \ref{mini}, \ref{non-decreasing}, and \Cref{mainth}.
\end{proof}

In applications, it is often natural to take the initial budget $x_0$ equal to the cost of the benchmark, i.e.,
\begin{align*}
x_0=\cost(\Fbenchinv). 
\end{align*}
In this case, $\varepsilon_{\min}(x_0)=0$. Note that
\begin{align*}
\DB_\phi((1-\delta)\Fbenchinv,\Fbenchinv)
\longrightarrow 0
\quad\text{as }\delta\downarrow0. 
\end{align*}
Then, for every $\varepsilon>0$, there exists $\delta>0$ such that
\begin{align*}
\cost((1-\delta)\Fbenchinv)
=
(1-\delta)x_0
<
x_0,
\qquad
\DB_\phi((1-\delta)\Fbenchinv,\Fbenchinv)<\varepsilon. 
\end{align*}
Therefore, \Cref{assumption_1} is satisfied.}

\subsection{Choice of Bregman generator $\phi$}
With the BW divergence, an investor can penalize underperformance and outperformance relative to the benchmark asymmetrically. Moreover, as investors may not want to penalize gains when measured relative to the terminal wealth of the benchmark, see e.g., \cite{harlow1991asset} and \cite{klebaner2017optimal}, we introduce a family of Bregman generators who assess deviations in relation to a wealth threshold $\alpha \in [0,\infty)$. 

\begin{definition}[Bregman divergence with threshold $\alpha$]\label{defMBD} 
    Let $\phi\colon \mathbb{R} \to \mathbb{R}$ be a convex and continuously differentiable Bregman generator. For a given threshold $\alpha$, we define the Bregman generator $\tilde{\phi}(\cdot ; \alpha)$ associated with $\phi$ as 
    \begin{align}\label{MBG}
     \tilde{\phi}(x;\alpha) = \phi(x)\, \mathbb{I}_{\{ x \leq \alpha \}} +  \left(\phi'(\alpha) (x-\alpha)+ \phi(\alpha) \right)  \, \mathbb{I}_{\{ x > \alpha \}},
    \end{align}
     where $\mathbb{I}_{ \{\cdot \} }$ denotes the indicator function.
     Then the Bregman divergence with generator $\tilde{\phi}(\cdot ; \alpha)$ becomes
    \begin{align*}
        B_{\tilde{\phi}} (z_1, z_2; \alpha)
        &=
        B_{\phi} (z_1, z_2) \,\mathbb{I}_{\{ z_1, z_2 \leq \alpha \}} 
        + B_{\phi} (z_1, \alpha ) \,\mathbb{I}_{\{ z_1 \leq \alpha < z_2\}}
        \\[0.5em]
        & \quad 
        + \left\{B_{\phi} (\alpha, z_2)  + \big(\phi'(\alpha) - \phi'(z_2)\big) (z_1 - \alpha)\right\}
        \,\mathbb{I}_{\{z_2\le \alpha < z_1\}}\,.
    \end{align*}
    \end{definition}
    Note that $\tilde{\phi}$ is continuously differentiable and convex, thus all relevant results of problem \eqref{Yaari-o} apply to the Bregman divergence with threshold $\alpha$. 

We observe that if $z_1, z_2 \le \alpha$, then $B_{\tilde{\phi}}$ reduces to the Bregman divergence with $\phi$ and if $z_1, z_2 > \alpha$, then $B_{\tilde{\phi}} \big(z_1, z_2; \alpha\big)=0$, meaning that the Bregman divergence no longer accounts for the distance between points that both exceed the threshold $\alpha$. If $z_1 \le\alpha < z_2$, the divergence equals $B_{\phi} (z_1, \alpha )$, that is only deviation of $z_1$ to the reference point $\alpha$ are taken into account, independent of the value of $z_2$. Moreover, due to the asymmetry of the Bregman divergence, if $z_2\le \alpha < z_1$, the penalization is not only $B_{\phi} (\alpha, z_2 )$ but also via the derivative of $\phi$. 
The Bregman divergence with threshold $\alpha$ thus combines the tractability of Bregman divergence with the ability to measure downside risk relative to the wealth threshold $\alpha$ and to moreover penalize underperformance relative to the benchmark. 
To gain a deeper understanding of the role of the threshold $\alpha$, \Cref{prop_MBW} establishes the monotonicity and limit properties of $\DB_{\tilde{\phi}} (\breve{F}_1, \breve{F}_2; \alpha)$ with respect to $\alpha$. Its proof is relegated to \Cref{app:proofs}.

\begin{proposition}\label{prop_MBW}
    For $\breve{F}_1,\breve{F}_2 \in \mQ$, the function $\alpha \mapsto \DB_{\tilde{\phi}} (\breve{F}_1, \breve{F}_2; \alpha)$ is non-decreasing. Moreover, it holds that
    \begin{align}\label{lima}
        \lim_{\alpha \to 0^+} \DB_{\tilde{\phi}} (\breve{F}_1, \breve{F}_2; \alpha) = 0, \quad \text{and} \quad \lim_{\alpha \to \infty} \DB_{\tilde{\phi}} (\breve{F}_1, \breve{F}_2; \alpha) = \DB_\phi (\breve{F}_1, \breve{F}_2).
    \end{align}
\end{proposition}

The key emphasis of the divergence $\DB_{\tilde{\phi}} (\breve{F}_1, \breve{F}_2; \alpha)$ is when $\breve{F}_1$ and $\breve{F}_2$ do not simultaneously exceed the threshold $\alpha$. Specifically, for $t \in (0,1)$, if both $\breve{F}_1(t)$ and $\breve{F}_2(t)$ are less than $\alpha$, then $B_{\tilde{\phi}} \big(\breve{F}_1(t), \breve{F}_2(t); \alpha \big) = B_{\phi} \big(\breve{F}_1(t), \breve{F}_2(t) \big)$, indicating that the divergence is measured using the Bregman divergence with $\phi$. However, if $\breve{F}_1(t)$ and $\breve{F}_2(t)$ lie on opposite sides of $\alpha$, then $B_{\tilde{\phi}} \big(\breve{F}_1(t), \breve{F}_2(t); \alpha \big) \leq B_{\phi} \big(\breve{F}_1(t), \breve{F}_2(t) \big)$, thereby penalizing less than with $B_\phi$, and if both $\breve{F}_1(t)$ and $\breve{F}_2(t)$ are larger than the wealth threshold $\alpha$ the divergence is zero. 
    
Alternatively, we can view $\DB_{\tilde{\phi}}$ as a function of $\alpha$. If $\alpha = 0$, the investor's problem \eqref{Yaari-o} reduces to the case without a divergence constraint, that is the classical Merton problem. As $\alpha$ increases to infinity, $\DB_{\tilde{\phi}}$ converges to the BW divergence with $\phi$. Therefore, $\alpha$ can be regarded as a weighting parameter that interpolates between these two extreme cases.

In the numerical examples below, we use the following two convex functions as  Bregman generators 
\begin{align*}
&\phi_{1}(x)=x^{2} \quad \text{and} \quad \phi_{2}(x)=x \ln x, \quad \text{for} \quad x >0\,.
\end{align*}
The choice $\phi_{1}$ corresponds to the 2-Wasserstein distance where positive and negative deviations from the benchmark quantiles are penalized equally. By contrast, under $\phi_{2}$ deviations are penalized in an asymmetric manner in that negative deviations (underperforming the benchmark) receive more weight than positive ones (outperforming the benchmark). Further, note that the convexity of $\phi_2$ becomes less strong when $x$ increases, thus underperforming the benchmark is penalised significantly more. We further use these two convex functions as generators for the Bregman divergence with given threshold $\alpha$, which then become
  \begin{align*}
      \tilde{\phi}_1(x;\alpha) 
      & =  x^2  \mathbb{I}_{\{ x \leq \alpha \}} + 2\alpha (x - \frac{\alpha}{2}) \mathbb{I}_{\{ x > \alpha \}}\,,
      \\    
      \tilde{\phi}_2(x;\alpha) 
      &= x \ln x  \mathbb{I}_{\{ x \leq \alpha \}} + \big( (\ln \alpha + 1)x - \alpha \big)   \mathbb{I}_{\{ x > \alpha \}}\,.
  \end{align*} 
Note that both generators are linear in $x$ on  ($\alpha$,$\infty$). 

{ 
As discussed in \Cref{remark_convex}, in the numerical examples below we incorporate a regularization term of the form $\frac{\vartheta}{2}x^2$ to $\phi_{2}(x)$, $\tilde{\phi}_1(x;\alpha)$, and $\tilde{\phi}_2(x;\alpha)$. Unless otherwise stated, we set $\vartheta = 10^{-8}$ by default. For notational convenience, we retain the original notation to denote the corresponding regularized functions.}

\subsection{Example 1: Optimal payoff when the benchmark is a constant}\label{sec:ex-const-bench}
 
We study optimal payoff choice for an investor who aims to maximize his expected utility under the constraint that their terminal wealth does not diverge too much from the terminal wealth when pursuing an investment bearing a fixed rate of return $\kappa>0$. That is, we assume in problem \eqref{Yaari-o} that the benchmark's quantile function is $\breve{F}_{b} \equiv e^{\kappa T}$. 
This situation corresponds to the case of pension fund managers whose   performance and solvency position might be assessed with reference to such a benchmark. For instance, under a fair value approach, the supervisory authorities may require that $\breve{F}_{b} \equiv e^{\kappa T}$, where $\kappa$ is equal to the risk free rate $r$. 

In order to determine the appropriate tolerance level  $\varepsilon$, the investor considers three investment strategies that he deems \textit{acceptable} (although not necessarily optimal). For each of these strategies, the investor calculates their BW divergence to the benchmark's and defines $\varepsilon$ as the maximum among them (one for each of the Bregman generator $\phi$ we study). The three \textit{acceptable} strategies are as follows:

\textbf{Strategy 1:} The constant mix strategy is a strategy where one maintains, through continuous rebalancing, a portfolio which during the entire investment period $[0,T]$ keeps the proportions invested in the riskless asset respectively in the stock constant. We consider $82.5\%$ invested in the risk-free asset and the remaining proportion of $17.5\%$ is invested in stock, in which case, the quantile function of the terminal wealth is 
\begin{equation*}
		\breve{F}_T(u) = e^{\mu_1 + \sigma_1 \breve{\Phi}(u)}, \quad u \in (0,1)\,,
\end{equation*}
where $\mu_1 := \left(r+ \left(\mu_s -r \right) 0.175 - \frac{1}{2} 0.175^2{\sigma^2_s}  \right) T $ and $\sigma_1:=0.175\, \sigma_s  \,\sqrt{T}$.

\textbf{Strategy 2:} For this strategy one invests at $t=0$ a proportion of the initial wealth in the risk-free asset and the remaining proportion in the stock. This is a buy and hold strategy, meaning that during $(0, T]$ there are no trades, thus the relative exposures to the risk-free asset and to the stock evolves over time. We consider $85\%$ of the initial wealth invested in the risk-free asset and the remaining proportion of  $15\%$ is invested in stock. The terminal wealth's quantile function is 
\begin{align*}
    \breve{F}_T(u) = \left( 0.15e^{\mu_2+\sigma_2 \breve{\Phi}(u)} + 0.85e^{rT} \right), \quad u \in (0,1)\,,
\end{align*}
where $\mu_2:=\left(\mu_s - \frac{1}{2}\sigma^2_s \right)T
$ and $\sigma_2:=\sigma_s \sqrt{T}$.

\textbf{Strategy 3:} This strategy yields a so-called digital payoff $Y_T$ on the stock $S_T$, given by  
\begin{align*}
	Y_T = 0.9\, \mathbb{I}_{ \{ S_T\leq c\} } +  \tfrac{1}{0.95}\,(e^{rT}-0.045)\, \mathbb{I}_{ \{ S_T > c\} }  \,,
	\end{align*}
where $c >0$ is such that $\mathbb{Q}(S_T<c)=0.05$. Note that the cost of $Y_T$ is indeed equal to $x_0=1$. Its quantile function is given as 
\begin{align*}
    \breve{F}_T(u) = 0.9 \,\mathbb{I}_{ \{ u\leq \alpha_c\} } +  \tfrac{1}{0.95}(e^{rT}-0.045)\, \mathbb{I}_{ \{ u > \alpha_c\} }, \quad  u \in (0,1)\,, 
\end{align*}
where $\alpha_c:= \mathbb{P}(S_T<c)$.

\begin{figure}[H]
    \centering
    \begin{subfigure}[b]{0.49\textwidth}
        \centering
        \includegraphics[width=\textwidth]{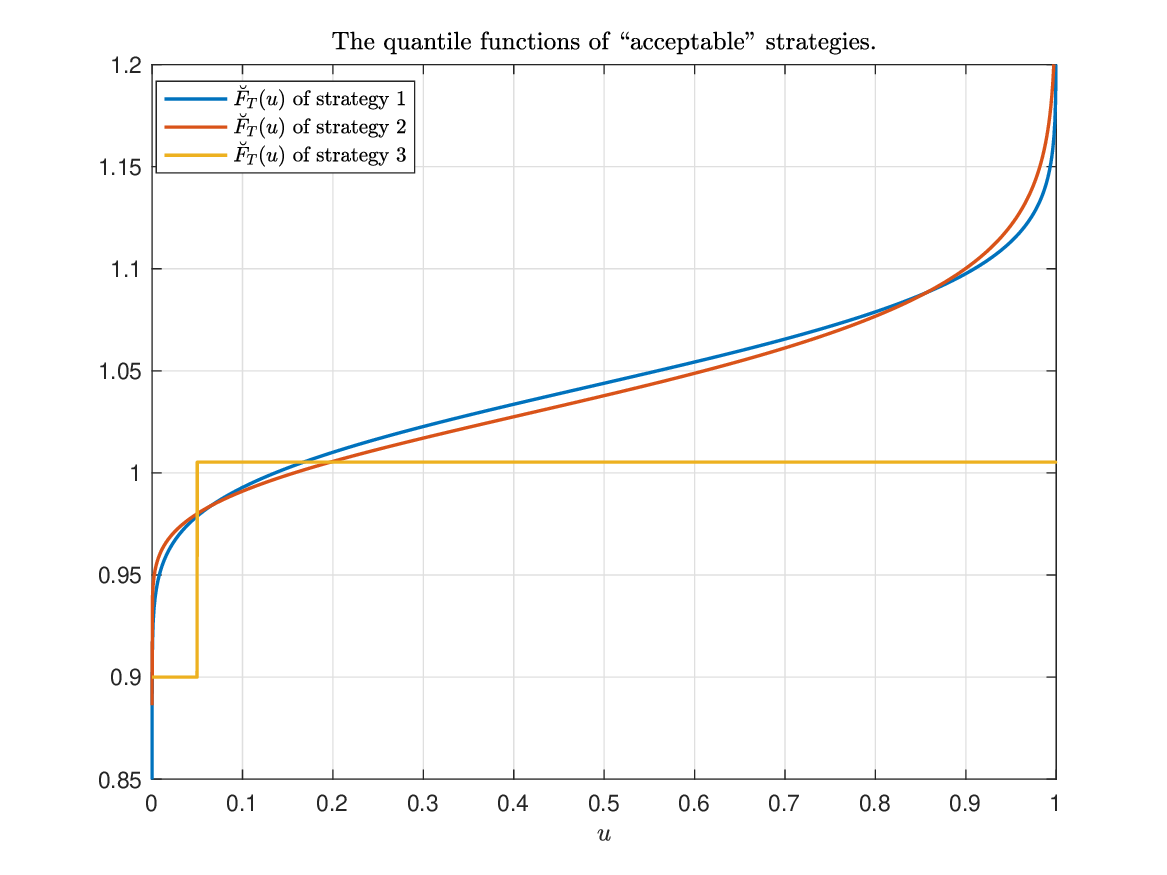}
    \end{subfigure}
    \hfill 
    \begin{subfigure}[b]{0.49\textwidth}
        \centering
        \includegraphics[width=\textwidth]{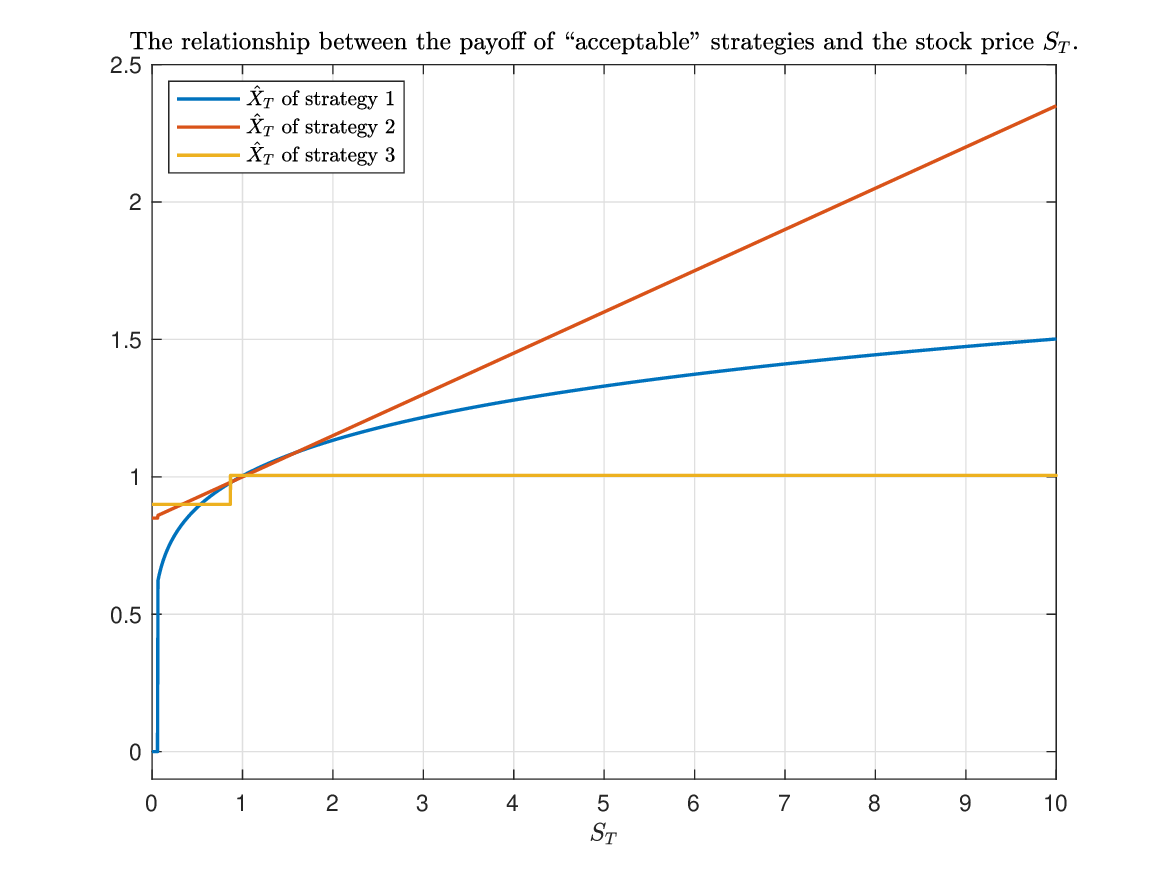}
    \end{subfigure}
	\caption{The quantile functions of the three \textit{acceptable} strategies (left panel) and the  relationship between their payoffs and the stock price (right panel). }
	\label{acceptable_qr}
\end{figure}
\Cref{acceptable_qr} illustrates the quantile functions of the three \textit{acceptable} strategies (left panel) and the relationship between their payoffs and the stock price (right panel). Notably, the quantile functions of strategies 1 and 2 have similar shapes, however their lower bounds $\breve{F}_T(0^+)$ differ significantly. For strategy 1 the lower bound is $0$, whereas for strategy 2 the lower bound is $0.85$. Thus, strategy 1 carries more risk in the event of extreme stock price losses, see also the right panel. As compared with the first two strategies, strategy 3 provides the best protection against downside risk, which however comes at the cost of limited upside potential.

\begin{table}[thbp]
  \centering
  \caption{BW divergences between each \textit{acceptable} strategy's
   terminal wealth and that of the benchmark's. The top row corresponds to the BW divergence with generator $\phi_1$, while the second row to the generator $\phi_2$. The last column provides the chosen tolerance level $\varepsilon$, that is the maximal BW divergences of each row.   }
\begin{tabular}{c@{\hskip 0.25in}ccc c}
     BW generator    & \multicolumn{3}{c}{BW divergence} & \hspace{0.25em}chosen $\varepsilon$\\
     \cmidrule{2-4}
     & Strategy 1    & Strategy 2    & Strategy 3    &   \\
    \midrule
    $\phi_1$  & $0.003673 $& $0.003717$ & $0.000526 $ & $0.003717 $ \\[0.5em]
    $\phi_2$  & $0.001785 $ & $0.001799  $ & $0.000272  $& $0.001799 $\\
    \bottomrule
    \end{tabular}%
  \label{tab:1}%
\end{table}%

\Cref{tab:1} reports for each \textit{admissible} strategy the BW divergence of its terminal wealth to that of the benchmark's. The last column displays the maximal $\varepsilon$ for each row, which is the  tolerance level the investor chooses. The first row corresponds to the Bregman generator $\phi_1$ and the second to the generator $\phi_2$.

We first study the optimal payoff when the investor does not care about the benchmark, that is when the BW constraint is absent and $\varepsilon = + \infty$ in problem \eqref{Yaari-o}. Figure \ref{Test_gamma_b} displays the optimal quantile functions for risk aversion parameters $\gamma=0.5$ and $\gamma=0.8$. Note that when $\gamma=0.8$ the quantile function is, as expected, flatter. We observe that these quantile functions are far away from the benchmark.    

\begin{figure}[tt]
    \centering
    \begin{subfigure}[b]{0.49\textwidth}
        \centering
 \includegraphics[width=\textwidth]{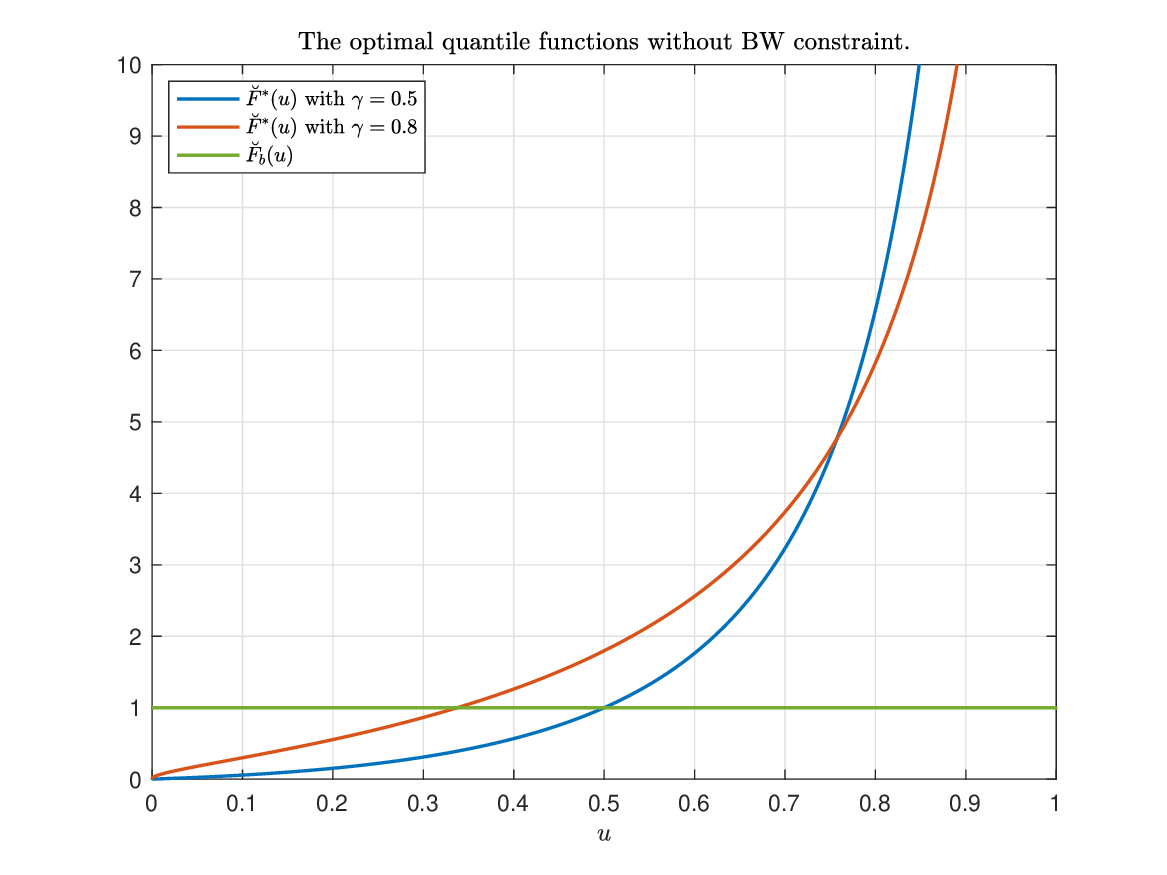}
    \end{subfigure}
	\caption{The optimal quantile function $\breve{F}^*(\cdot)$ in the absence of a BW constraint, i.e., $\varepsilon = + \infty$ in problem \eqref{Yaari2}, for $\gamma=0.5$ (blue line) and $\gamma=0.8$ (red line), respectively. The green curve depicts the quantiles of the constant benchmark. }
	\label{Test_gamma_b}
\end{figure}

Next, \Cref{Test_gamma_a} displays the optimal quantile function for problem \eqref{Yaari2} when the BW constraint is applied. We observe that the shape of the optimal quantile functions changes in a significant manner and the variation in $\gamma$ has a smaller impact on $\breve{F}^*$, indicating that the BW constraint becomes the dominant factor. We make the following additional observations. First, the similarity of $\breve{F}^*$ under the different $\phi_i$ suggests that our method for choosing  $\varepsilon$ is not unreasonable. Second, in both the left and right panels we see that - to satisfy the BW constraint - the optimal payoffs are less likely to fall below the benchmark payoff and tend to stay near it. Specifically, the optimal payoff has about $5$\% probability of falling below that of the benchmark's, while the upper bound of the payoff, $\breve{F}^*(1^-)$, is around $1.073$.

\begin{figure}[t]
    \centering
    \begin{subfigure}[b]{0.49\textwidth}
        \centering
        \includegraphics[width=\textwidth]{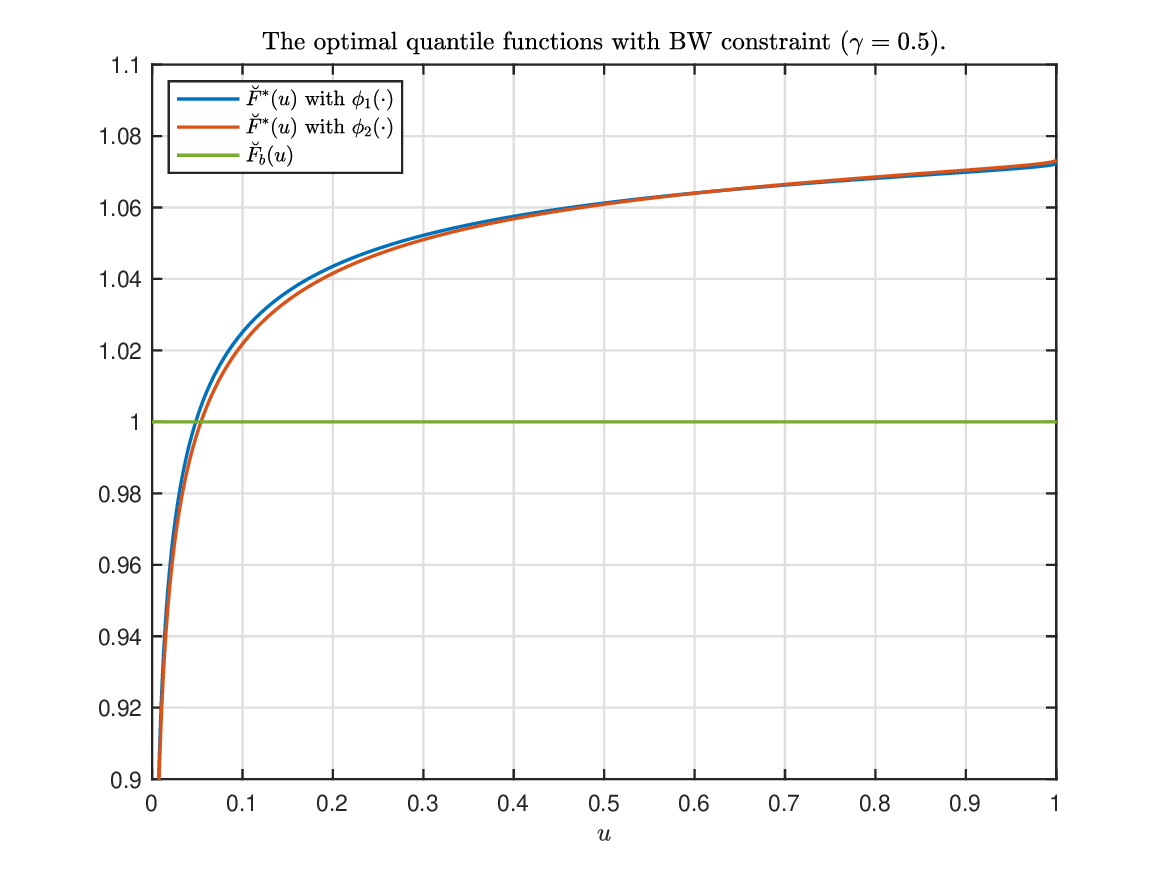}
    \end{subfigure}
    \hfill
    \begin{subfigure}[b]{0.49\textwidth}
        \centering
        \includegraphics[width=\textwidth]{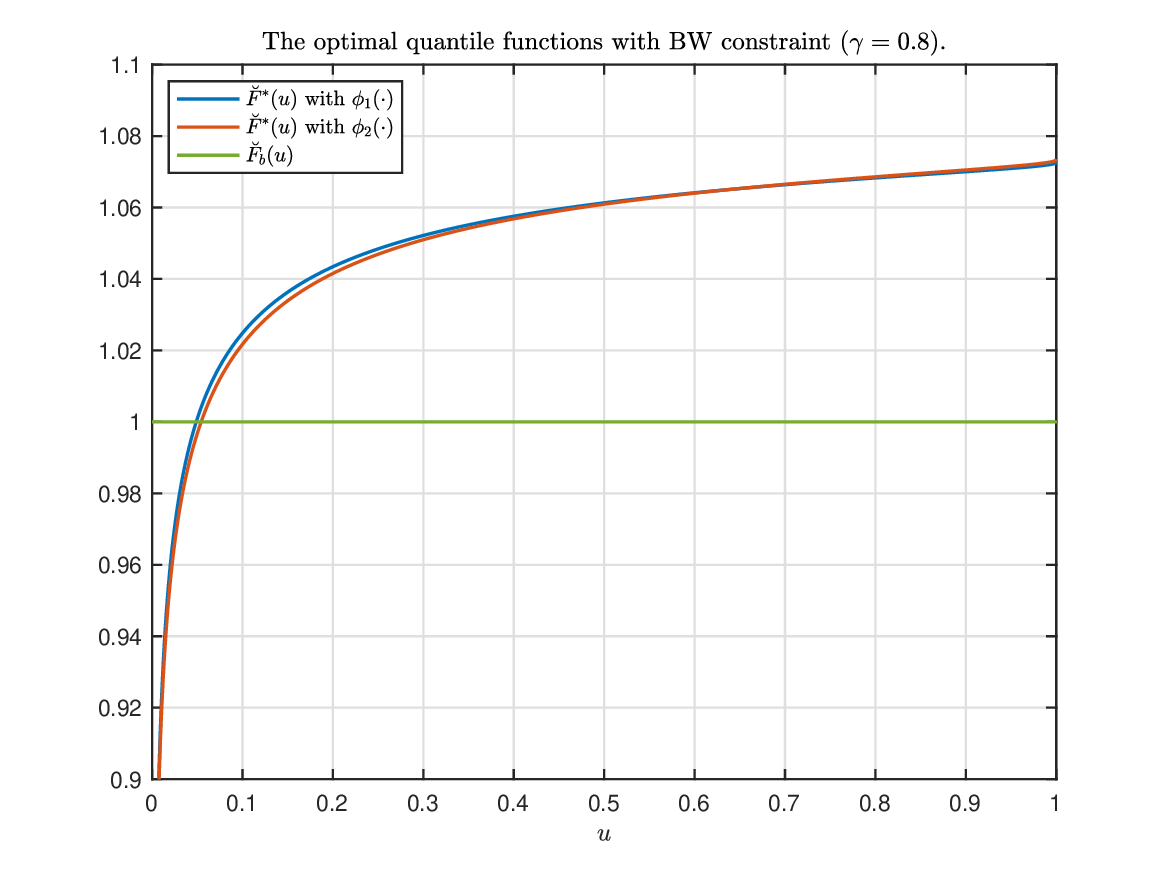}
    \end{subfigure}
	\caption{The optimal quantile functions $\breve{F}^*(u)$ for the generators $\phi_1(x)=x^2$ (blue lines), $\phi_2(x)=x \ln x$ (red lines), and the constant benchmark $\breve{F}_{b}(u)$ (green line) for $\gamma=0.5$ and $\gamma=0.8$, respectively.  }
	\label{Test_gamma_a}
\end{figure}

To further explore the behaviour of the optimal payoffs, \Cref{Test_gamma_c} displays how they depend on the stock  $S_T$. In this regard, recall that $\hat X_{T}={ {\breve{F}^*}}(1-F_{\varphi_{T}}(\varphi_{T}))$. Since $\mu_s\ > r$, it holds that $\varphi_{T}$ is monotonically decreasing in $S_T$, which results that the optimal payoffs are non-decreasing in $S_{T}$. For $\gamma=0.5$, we observe that when the stock price declines by more than 60 percent, i.e., when $S_T < 0.4$, then all optimal payoffs are close to zero. For $0.4 \leq S_T < 1$, the optimal payoffs under BW constraints increase very rapidly with $S_T$ and clearly surpass the optimal payoff without BW constraints.  Moreover, we find that under BW constraints, optimal payoffs will outperform the benchmark as soon as $S_T > 0.87$, also reflecting a capacity for risk mitigation. However, their outperformance will always remain fairly moderate. In contrast, the optimal payoff without the BW constraint only begins to outperform the benchmark when $S_T >1.25$, at which point the outperformance increases rapidly. Similar findings hold for the case in which $\gamma=0.8$ (right panel of \Cref{Test_gamma_c}).
\begin{figure}[htbp]
    \centering
    \begin{subfigure}[b]{0.49\textwidth}
        \centering
        \includegraphics[width=\textwidth]{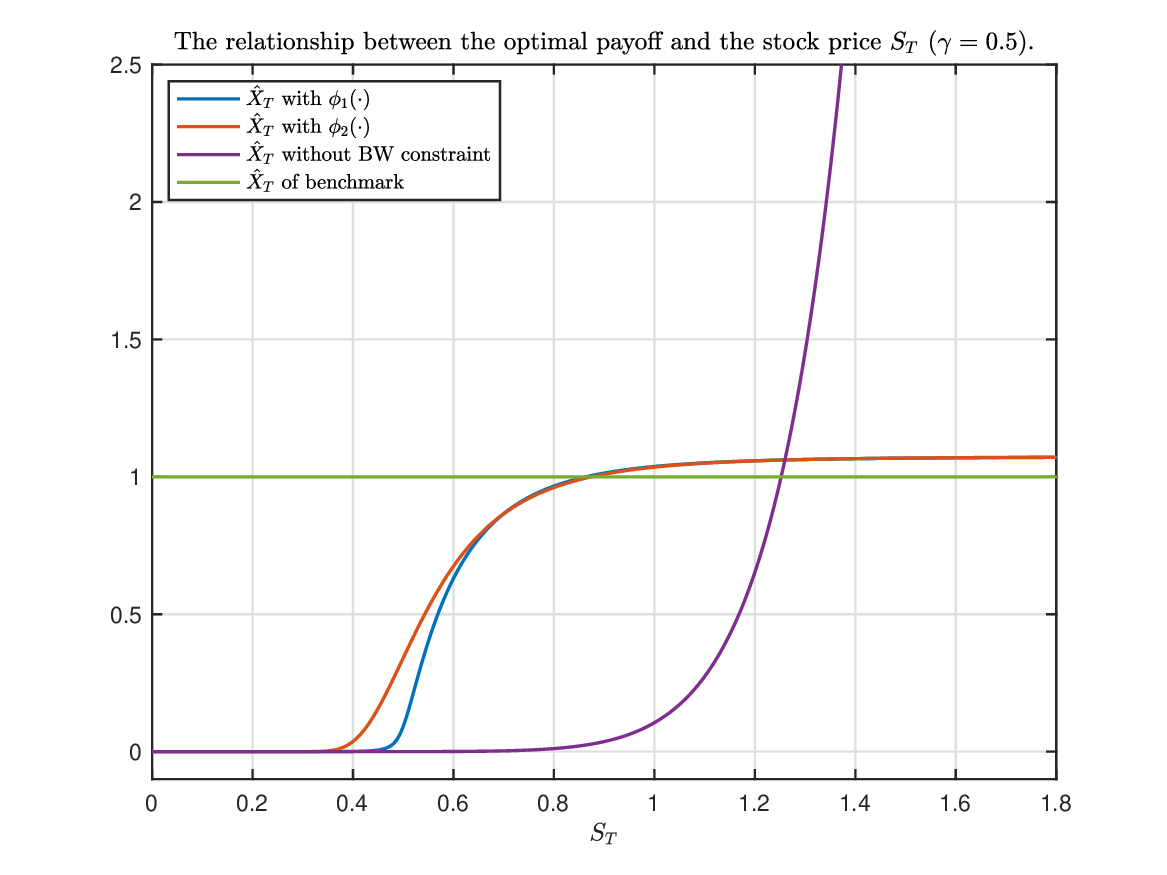}
    \end{subfigure}
    \hfill
    \begin{subfigure}[b]{0.49\textwidth}
        \centering
        \includegraphics[width=\textwidth]{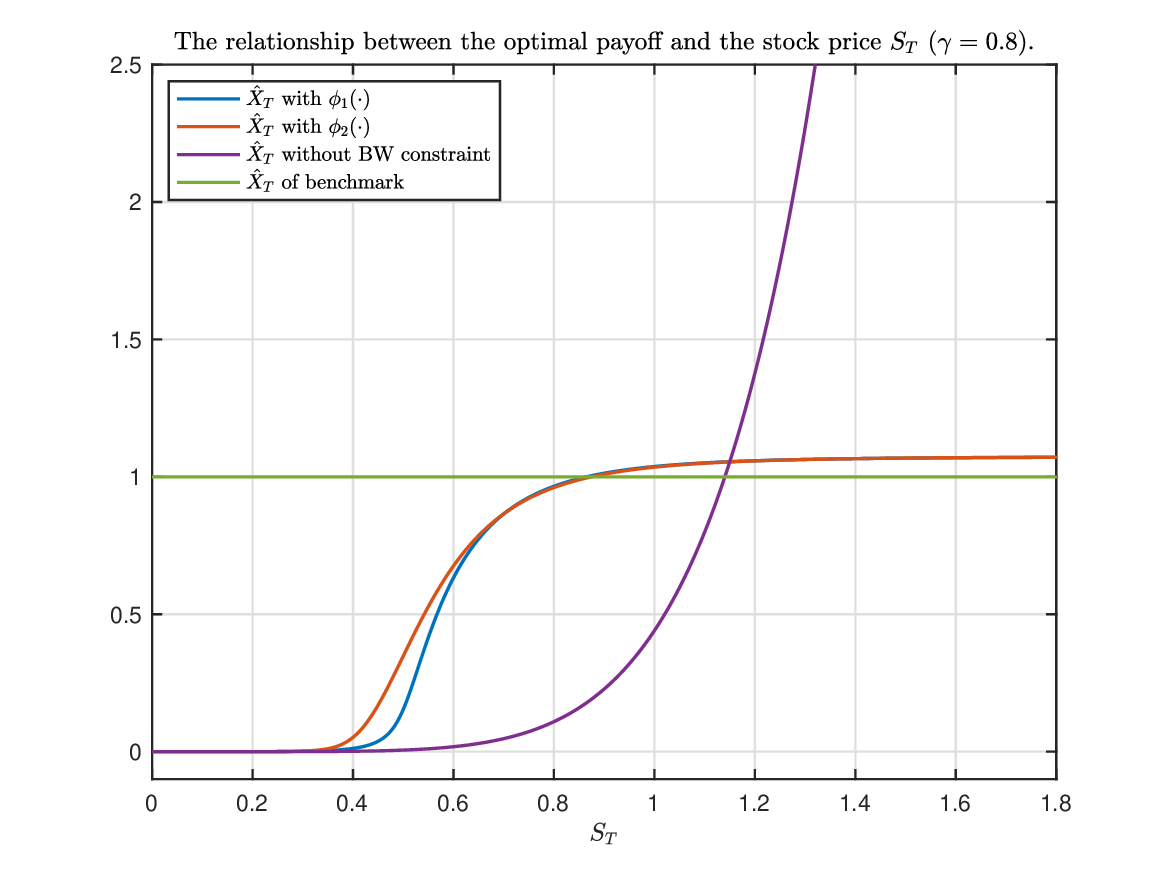}
    \end{subfigure}
	\caption{The optimal payoff $\hat X_{T}$ as a function of $S_T$  under the BW constraint induced by the generators $\phi_1(x)=x^2$ (blue lines), $\phi_2(x)=x \ln x$ (red lines) and when no such constraint is considered (purple lines). The left panel depicts the optimal payoffs  for $\gamma=0.5$ and the right panel for $\gamma=0.8$.}
	\label{Test_gamma_c}
\end{figure}

Next, we study the optimal payoff under the BW divergence that are driven by generators $\tilde{\phi}_1(\cdot;\alpha)$ and $\tilde{\phi}_2(\cdot;\alpha)$ with threshold values $\alpha=0.95$ and $\alpha=1$. \Cref{tab:2} provides, similarly to \Cref{tab:1}, the BW divergence values for the three \textit{acceptable} strategies and the different BW divergences. Again, the last column depicts the tolerance levels $\varepsilon$ that the investor chooses.
\begin{table}[htbp]
  \centering
  \caption{BW divergence between each \textit{acceptable} strategy's terminal wealth and that of the benchmark's. The different rows corresponds to different BW divergences with generators indicated in the first column. The last column provides the chosen tolerance level $\varepsilon$, that is the maximal BW divergences of each row. 
  }
\begin{tabular}{c@{\hskip 0.25in}cccc}
     BW generator    & \multicolumn{3}{c}{BW divergence} &\hspace{0.25em} chosen $\varepsilon$\\
     \cmidrule{2-4}
     & Strategy 1    & Strategy 2    & Strategy 3    &   \\
    \midrule
    $\tilde{\phi}_1(\cdot;1)$ & $0.000088 $& $0.000065 $& $ 0.000500 $& $0.000500 $\\[0.5em]
    $\tilde{\phi}_2(\cdot;1)$  & $0.000045 $ & $0.000033 $& $0.000259 $& $0.000259 $\\[0.5em]
    $\tilde{\phi}_1(\cdot;0.95)$  & $0.000002$ & 0.0 & $0.000125 $& $0.000125$ \\[0.5em]
    $\tilde{\phi}_2(\cdot;0.95)$  & $0.000001$ & 0.0 & $0.000067$& $0.000067$ \\
    \bottomrule
    \end{tabular}
  \label{tab:2}%
\end{table}%
\begin{figure}[htbp]
    \centering
    \begin{subfigure}[b]{0.49\textwidth}
        \centering
        \includegraphics[width=\textwidth]{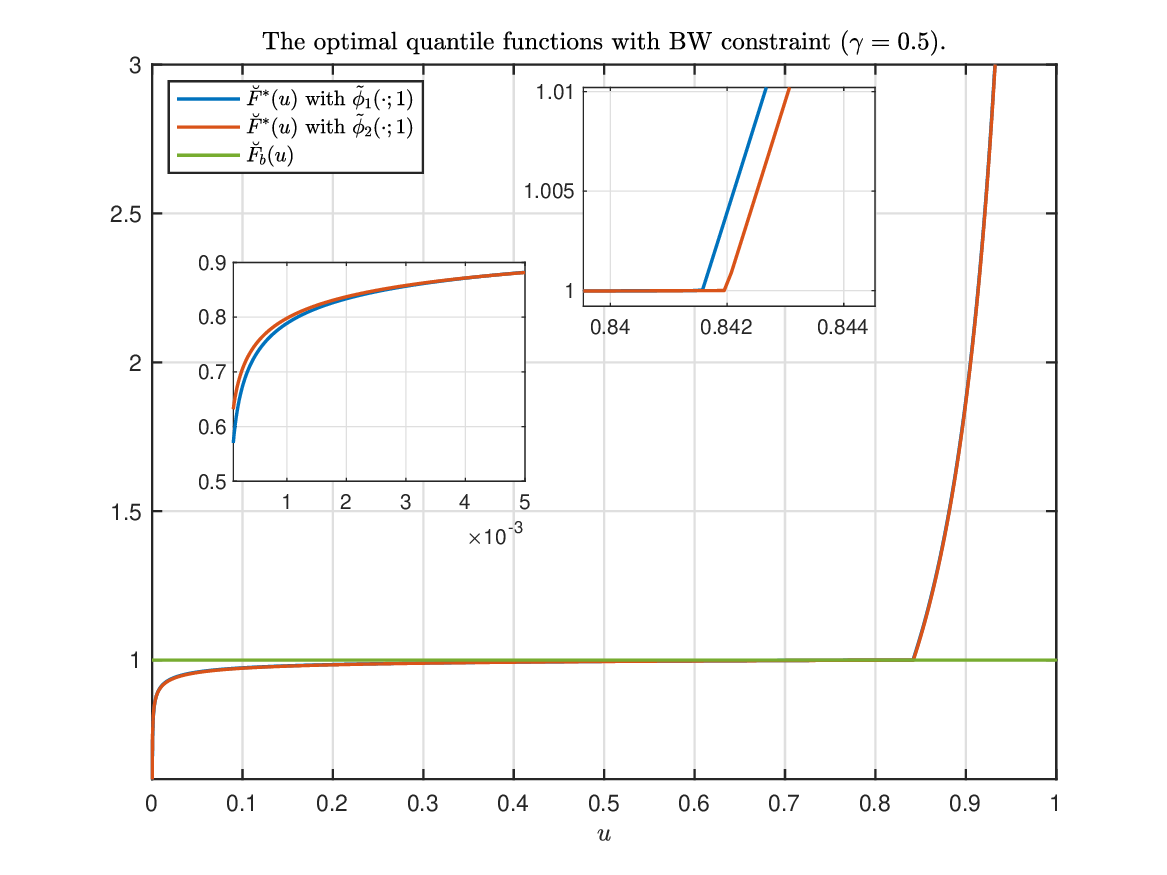}
    \end{subfigure}
    \hfill
    \begin{subfigure}[b]{0.49\textwidth}
        \centering
        \includegraphics[width=\textwidth]{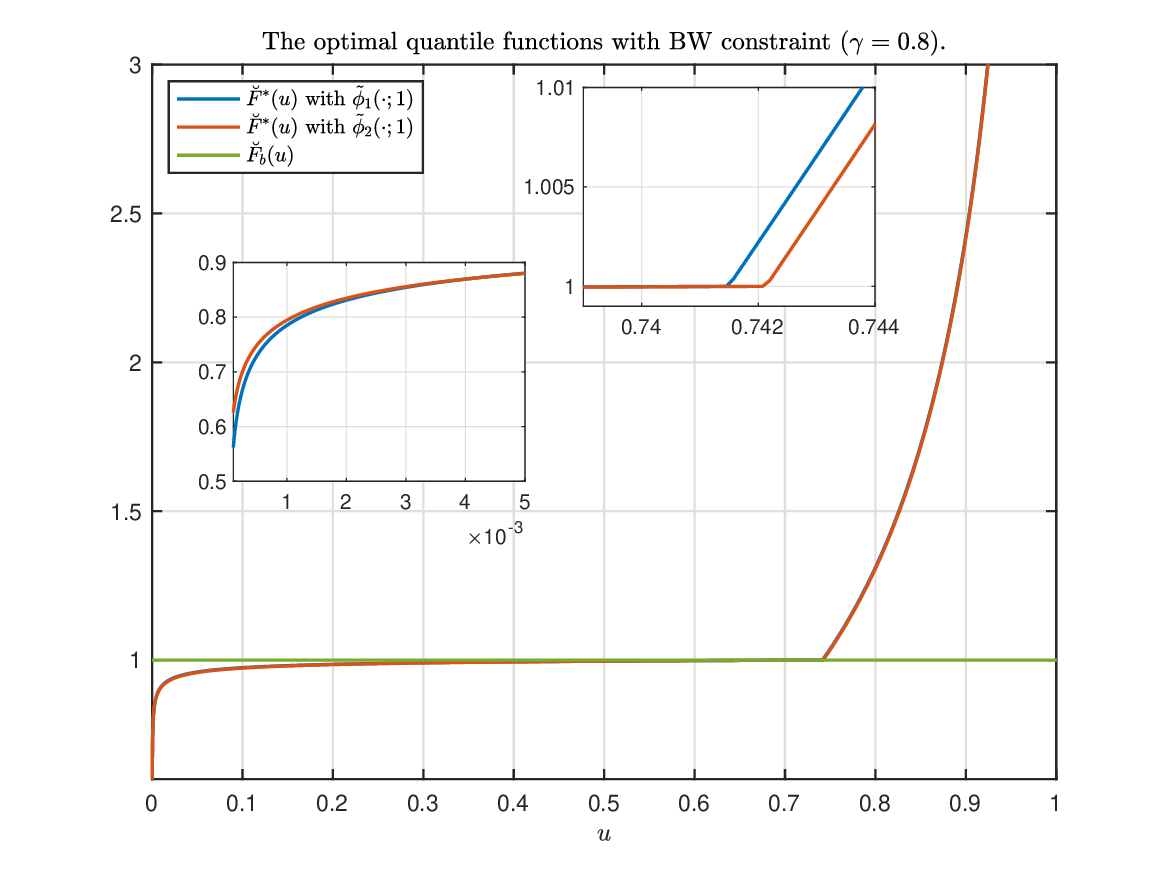}
    \end{subfigure}
    \begin{subfigure}[b]{0.49\textwidth}
        \centering
        \includegraphics[width=\textwidth]{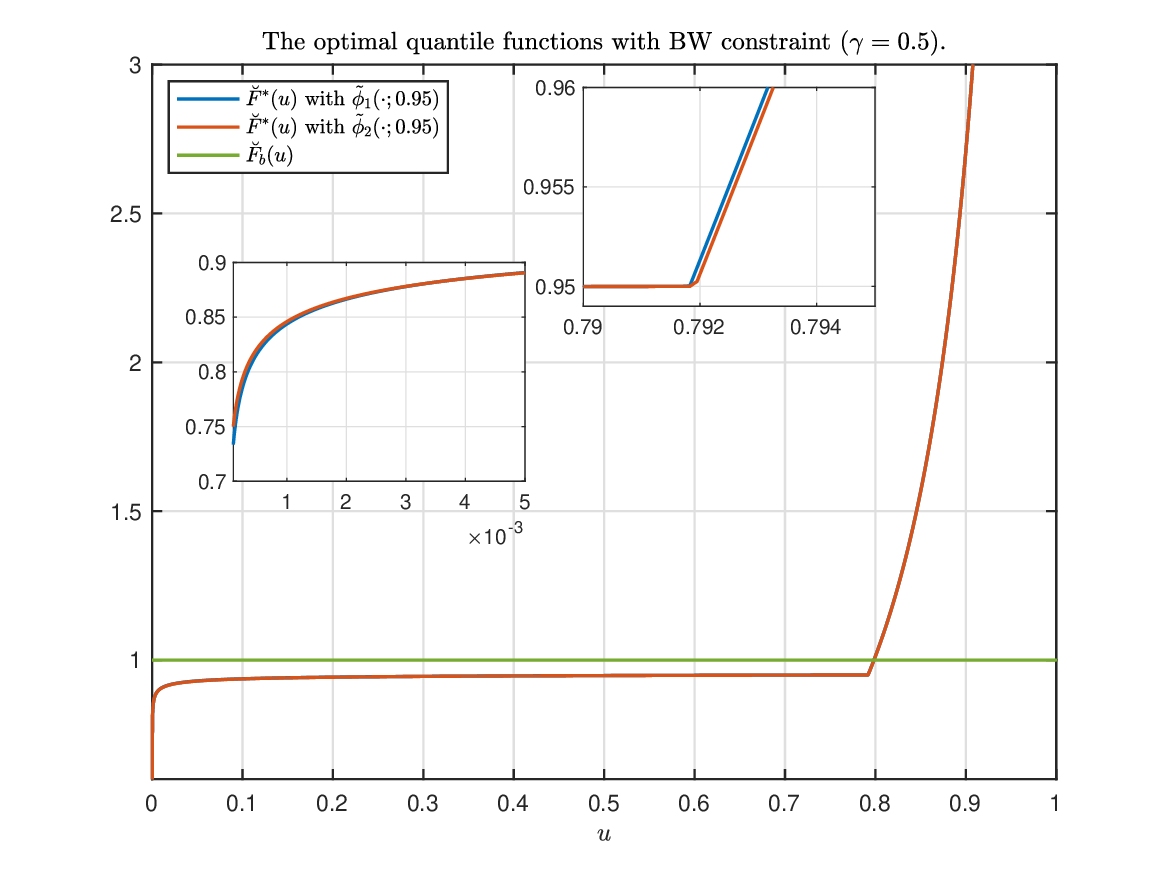}
    \end{subfigure}
    \hfill
    \begin{subfigure}[b]{0.49\textwidth}
        \centering
        \includegraphics[width=\textwidth]{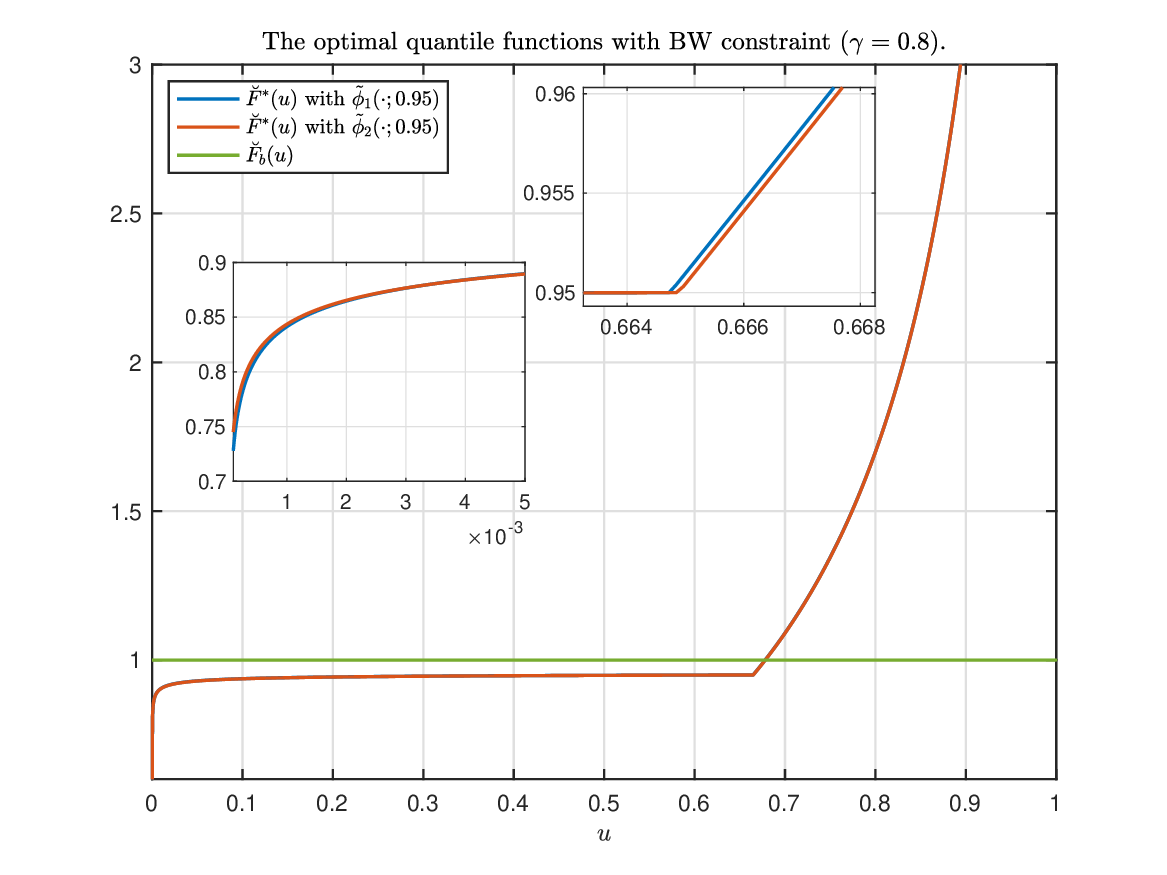}
    \end{subfigure}
	\caption{Optimal quantile functions $\breve{F}^*(u)$ for the generators $\tilde{\phi}_1(\cdot;\alpha)$ (blue lines), $\tilde{\phi}_2(\cdot;\alpha)$ (red lines) and the constant benchmark $\breve{F}_{b}(u)$ (green line). The BW generator threshold is $\alpha=0.95$  (bottom panels) and $\alpha=1$ (top panels) and the investor's risk aversion is $\gamma=0.5$ (left panels) and $\gamma=0.8$ (right panels), respectively. }
	\label{Test_improve_gamma_a}
\end{figure}

\begin{figure}[htbp]
    \centering
    \begin{subfigure}[b]{0.49\textwidth}
        \centering
        \includegraphics[width=\textwidth]{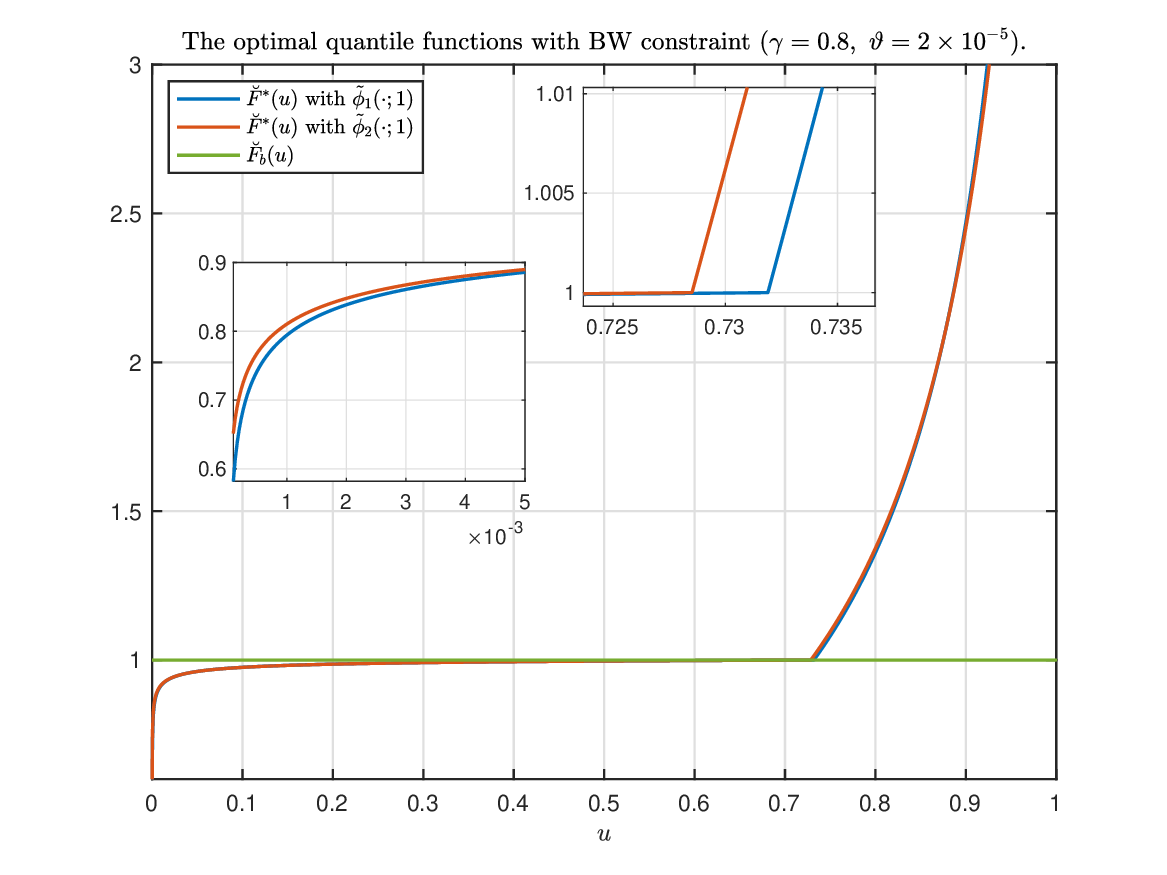}
    \end{subfigure}
    \hfill
    \begin{subfigure}[b]{0.49\textwidth}
        \centering
        \includegraphics[width=\textwidth]{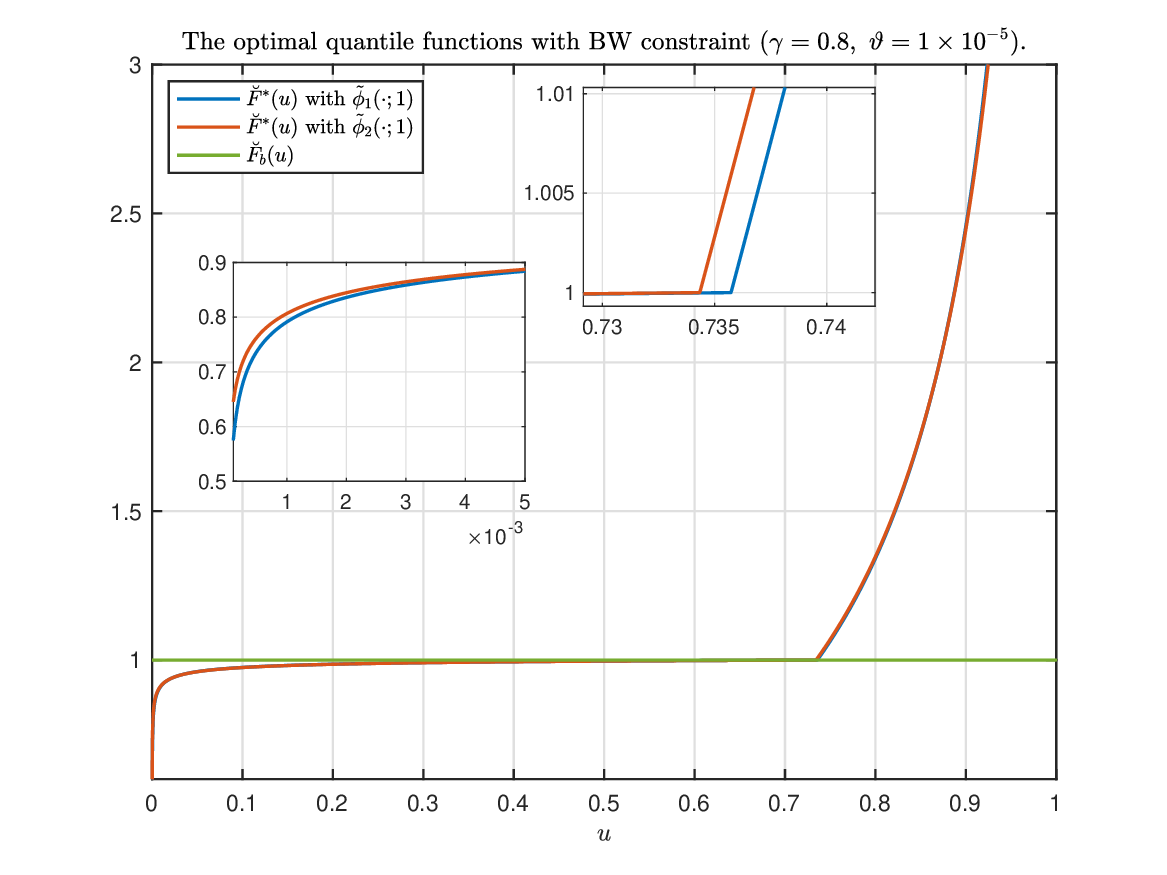}
    \end{subfigure}
    \begin{subfigure}[b]{0.49\textwidth}
        \centering
        \includegraphics[width=\textwidth]{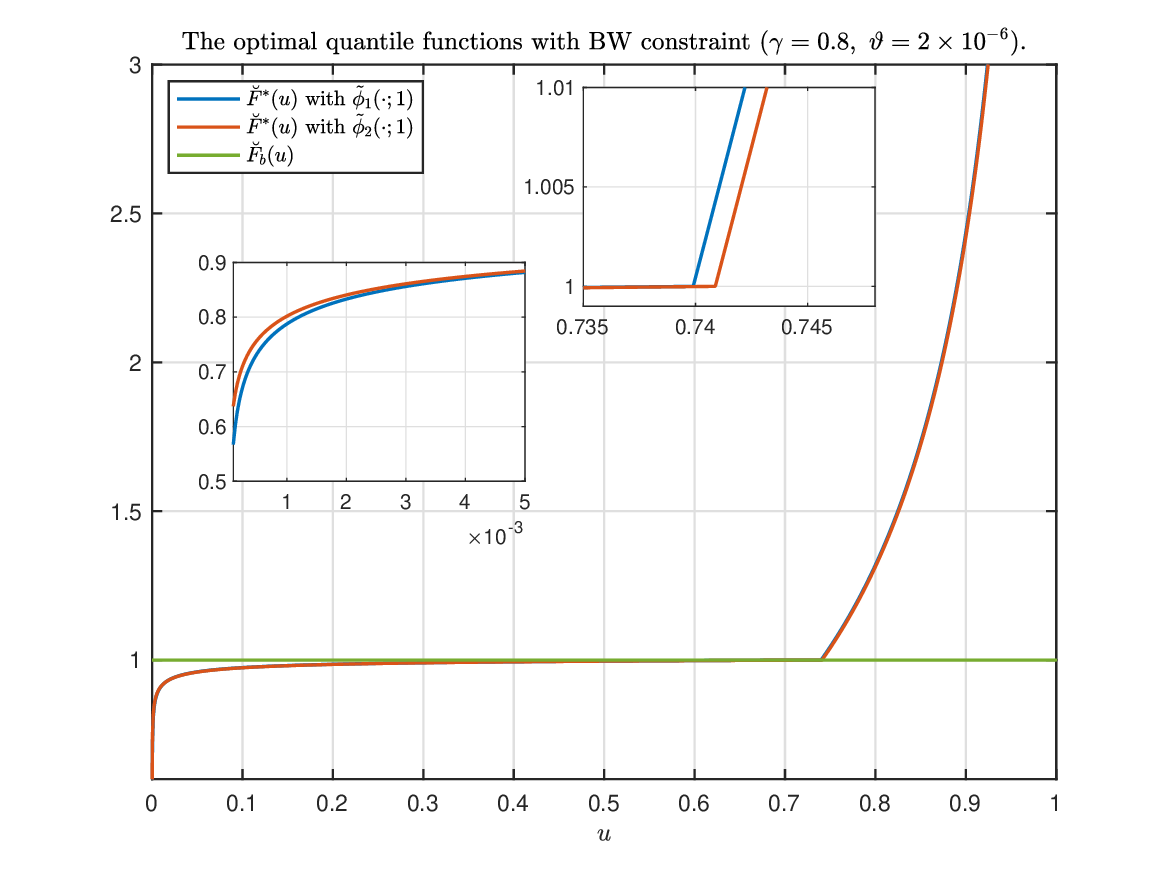}
    \end{subfigure}
    \hfill
    \begin{subfigure}[b]{0.49\textwidth}
        \centering
        \includegraphics[width=\textwidth]{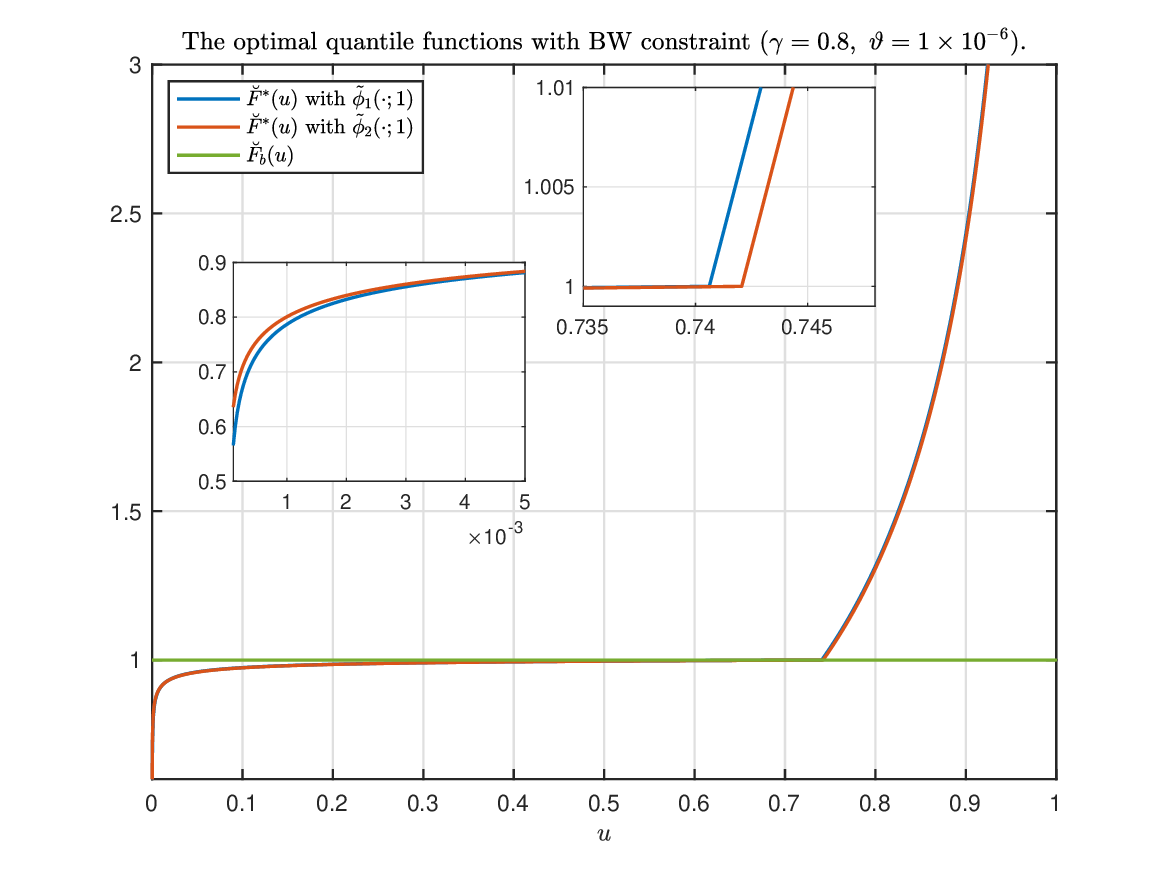}
    \end{subfigure}
	\caption{Optimal quantile functions $\breve{F}^*(u)$ for the generators $\tilde{\phi}_1(\cdot;1)$ (blue lines), $\tilde{\phi}_2(\cdot;1)$ (red lines) and the constant benchmark $\Fbenchinv(u)$ (green line) with different choice of $\vartheta$. The investor's risk aversion is $\gamma=0.8$.}
	\label{different_vartheta}
\end{figure}

  \begin{figure}[htbp]
    \centering
    \begin{subfigure}[b]{0.49\textwidth}
        \centering
        \includegraphics[width=\textwidth]{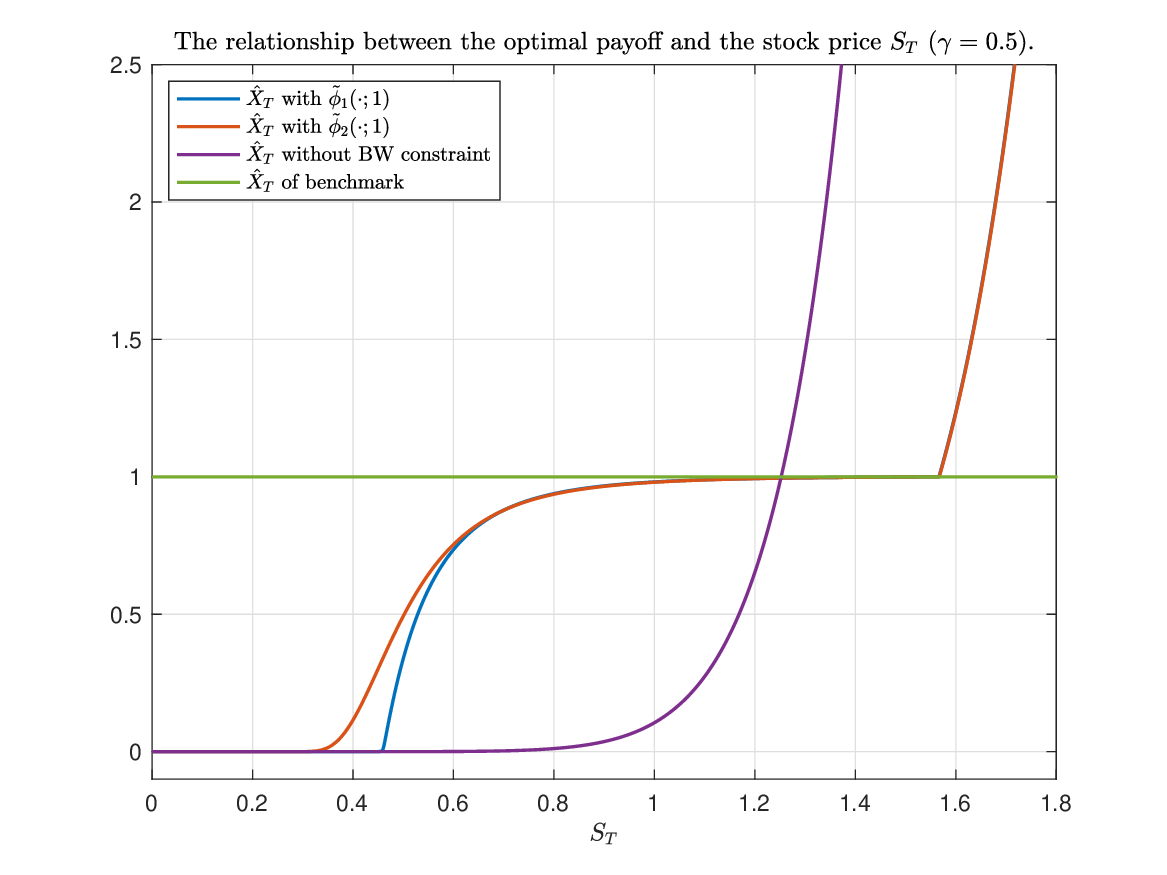}
    \end{subfigure}
    \hfill
    \begin{subfigure}[b]{0.49\textwidth}
        \centering
        \includegraphics[width=\textwidth]{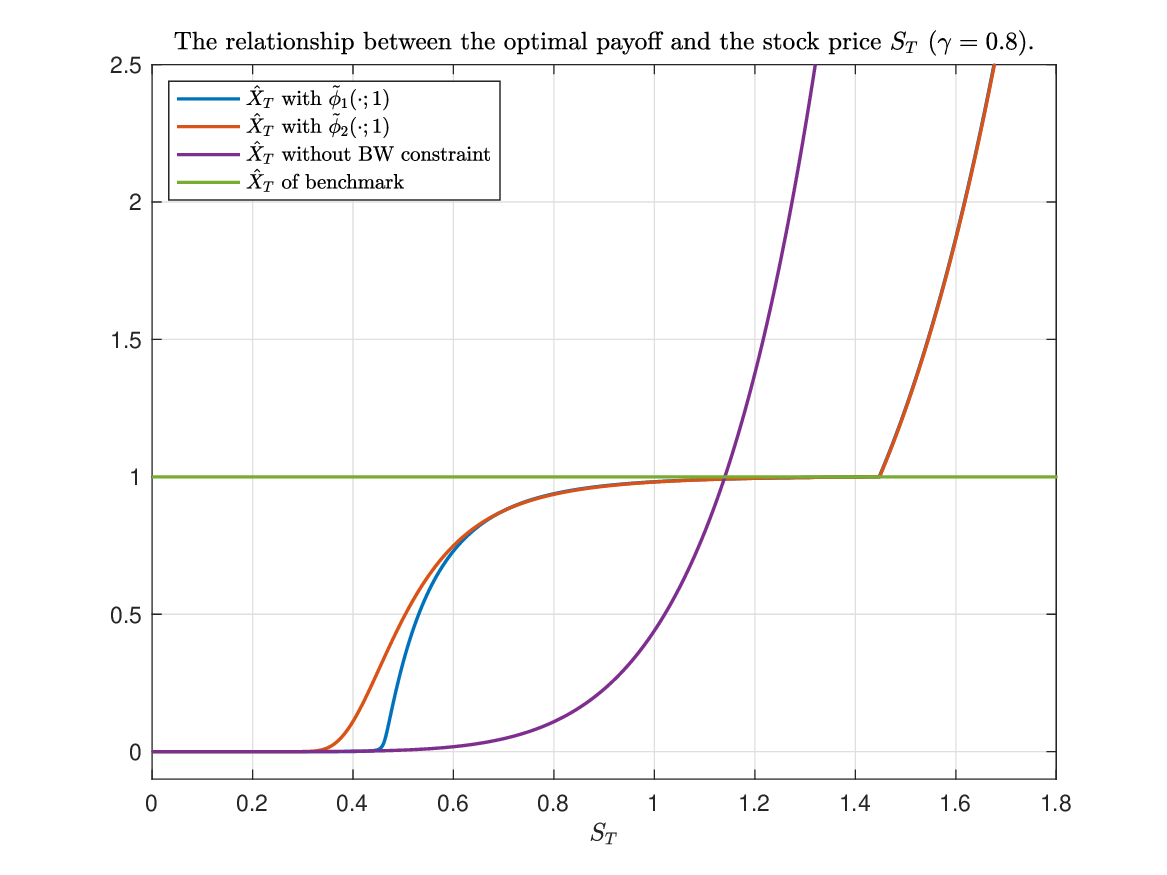}
    \end{subfigure}
    \begin{subfigure}[b]{0.49\textwidth}
        \centering
        \includegraphics[width=\textwidth]{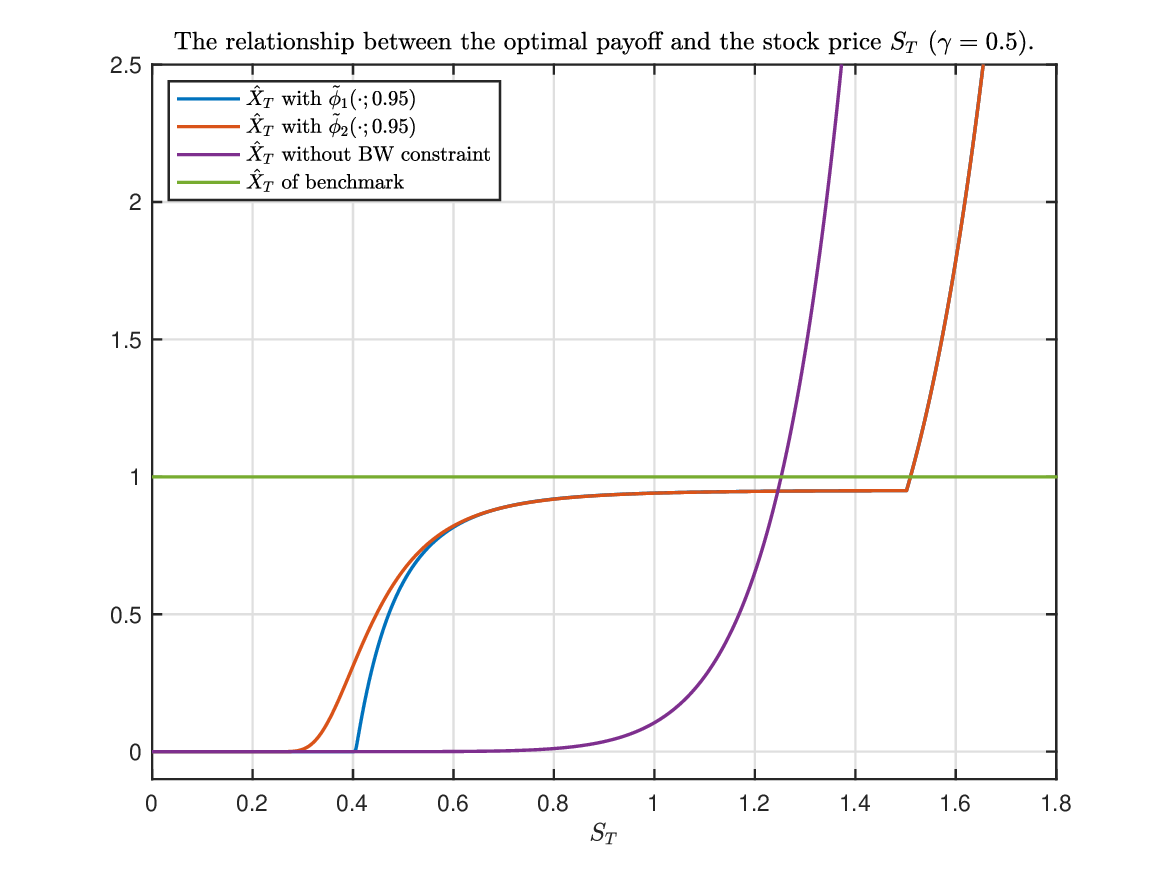}
    \end{subfigure}
    \hfill
    \begin{subfigure}[b]{0.49\textwidth}
        \centering
        \includegraphics[width=\textwidth]{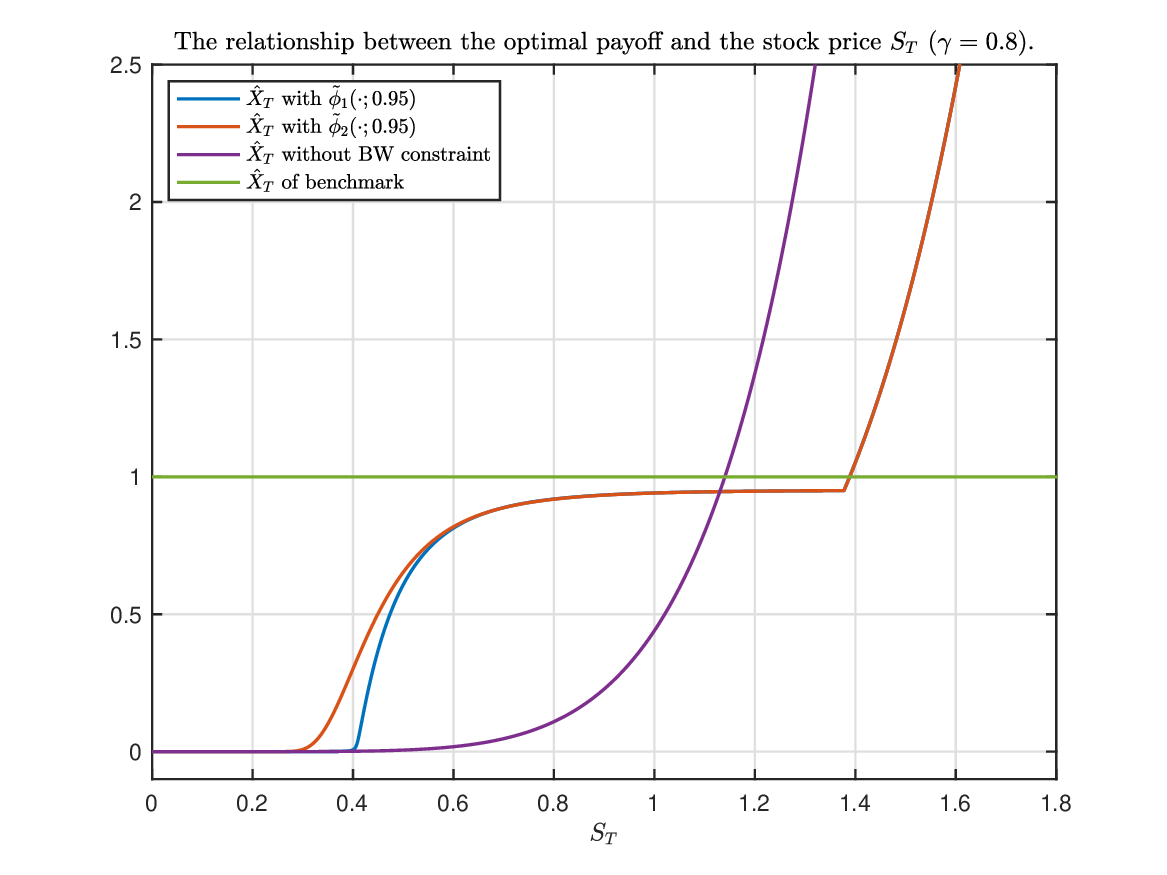}
    \end{subfigure}
	\caption{Optimal payoffs $\hat X_{T}$ as a function of $S_T$ for the generators $\tilde{\phi}_1(\cdot;\alpha)$ (blue lines) and $\tilde{\phi}_2(\cdot;\alpha)$ (red lines).  The purple line shows the payoff when the BW divergence constraint is absent. The BW generator threshold is $\alpha=0.95$  (bottom panels) and $\alpha=1$ (top panels) and the investor's risk aversion is $\gamma=0.5$ (left panels) and $\gamma=0.8$ (right panels), respectively.}
	\label{Test_improve_gamma_c}
\end{figure}

\Cref{Test_improve_gamma_a} presents the optimal quantile functions for the BW divergence with threshold $\alpha$. We observe that when $\breve{F}^*(u)$ is below the threshold value $\alpha = 1$ in the top panels ($\alpha = 0.95$ in the bottom panels), the optimal quantile function and the benchmark quantile function are close. This is in contrast to when $\breve{F}^*(u)$ exceeds $\alpha$, in which case the optimal quantile function increases rapidly. This indicates that the divergence constraint primarily targets the downside risk and penalizes outcomes that fall below the threshold $\alpha$. The magnified insets in each subplot illustrate the details of the quantile function $\breve{F}^*(u)$ within the interval $[0, 0.005]$ (bottom left) and when it crosses the threshold $\alpha$ (top right). Notably, as $\gamma$ increases (from left to right panels), $\breve{F}^*(u)$ reaches $\alpha$ earlier, but the subsequent rate of increase becomes slower. This indicates that risk averse investors (larger $\gamma$) are more inclined to secure the return level $\alpha$ rather than to pursue higher returns. In summary, since the budget and the $\varepsilon$ are fixed, if investors wish to reduce downside risk, they must correspondingly forgo high returns during favourable market conditions.

{\Cref{different_vartheta} showcases the effect of different regularization parameters $\vartheta$ on the numerical computation of the optimal quantile functions. It can be observed that the overall shapes of the optimal quantile functions remain similar across different values of $\vartheta$. As $\vartheta$ decreases, the influence of the regularization term becomes weaker. In fact, when $\vartheta = 10^{-6}$, the resulting optimal quantile functions are already very close to those shown in the upper-left panel of \Cref{Test_improve_gamma_a} ($\vartheta = 10^{-8}$). In practical applications, a smaller value of $\vartheta$ can be chosen to achieve higher numerical accuracy.}

\Cref{Test_improve_gamma_c} illustrates the relationship between the various payoffs   $\hat{X}_T$ and the stock price $S_T$. As compared to \Cref{Test_gamma_c}, we observe that the payoffs are no longer subject to an upper bound, allowing investors to achieve higher returns when the stock price is elevated. Furthermore, the optimal payoffs present strong risk resilience during periods of declining stock prices. The different figures clearly show that using the BW constraint with a wealth threshold makes it possible to effectively balance downside risk (by staying close to the constant benchmark) against the potential to take advantage of more favourable market returns.

\subsection{Example 2: Optimal payoff under a non-constant benchmark}
In this section, we consider the case of an investor with a more aggressive benchmark strategy. Specifically, the benchmark strategy has a constant exposure of $80\%$ in the stock and $20\%$ in the risk-free asset, that is maintained during the investment horizon. The quantile function of terminal wealth the benchmark strategy generates is 
\begin{align*}
    \breve{F}_{b}(u) = e^{  (r+(\mu_s-r)0.8-\frac{1}{2}0.8^2\,\sigma_s^2)T +  0.8\,\sigma_s \,\sqrt{T} \,\breve{\Phi}(u) }, \ u \in (0,1). 
\end{align*}
To determine the tolerance level $\varepsilon$, the investor pursues similarly to \Cref{sec:ex-const-bench} by considering the BW divergence to \textit{acceptable} strategies. Since the benchmark strategy is more aggressive to the constant benchmark in \Cref{sec:ex-const-bench}, we modify the three \textit{acceptable} strategies to reflect the investor's choice of benchmark.  

\textbf{Strategy 1:} The constant mix strategy with 75\% invested in the stock and the remaining 25\% in the risk-free asset.
 
\textbf{Strategy 2:} The buy-and-hold strategy with 85\% invested in the stock and the remaining 15\% in the risk-free asset.

\textbf{Strategy 3:} A digital payoff $Y_T$ on the stock $S_T$ given by  
\begin{align*}
	Y_T = 0.8 \,\mathbb{I}_{ \{ S_T\leq c\} } +  \tfrac{1}{0.9}(e^{rT}-0.08) \,\mathbb{I}_{ \{ S_T > c\} }.   
	\end{align*}
where $c >0$ is such that $\mathbb{Q}(S_T<c)=0.1$.

As before, we consider Bregman generators $\phi_i(\cdot)$, $\tilde{\phi}_i(\cdot;1)$ and $\tilde{\phi}_i(\cdot;0.95)$ for $i=1,2$.  \Cref{tab:3} presents the BW divergences between each \textit{acceptable} strategy's terminal wealth and that of the benchmark as well as the investor's chosen tolerance level $\varepsilon$ for each of the Bregman generators.

\begin{table}[htbp]
  \centering
  \caption{Choice of the tolerance levels $\varepsilon$. The first three columns display the divergences for the three strategies for all different generators that we consider. The last column provides the resulting tolerance level $\varepsilon$.}
    \begin{tabular}{c@{\hskip 0.25in}cccc}
     BW generator    & \multicolumn{3}{c}{BW divergence} & \hspace{0.25em}chosen $\varepsilon$\\
     \cmidrule{2-4}
     & Strategy 1    & Strategy 2    & Strategy 3    &   \\
    \midrule
    $\phi_1(\cdot)$  & 0.000506 & 0.001179 & 0.086821 & 0.086821 \\[0.5em]
    $\phi_2(\cdot)$  & 0.001785 & 0.000367 & 0.032795 & 0.032795 \\[0.5em]
    $\tilde{\phi}_1(\cdot;1)$  & 0.000007 & 0.000009 & 0.001108 & 0.001108 \\[0.5em]
    $\tilde{\phi}_2(\cdot;1)$  & 0.000004 & 0.000005 & 0.000630 & 0.000630 \\[0.5em]
    $\tilde{\phi}_1(\cdot;0.95)$  & 0.000007 & 0.000007 & 0.001023 & 0.001023 \\[0.5em]
    $\tilde{\phi}_2(\cdot;0.95)$  & 0.000004 & 0.000004 & 0.000586 & 0.000586 \\
    \bottomrule
    \end{tabular}% 
  \label{tab:3}%
\end{table}%

\begin{figure}[ht]
    \centering
    \begin{subfigure}[b]{0.49\textwidth}
        \centering        \includegraphics[width=\textwidth]{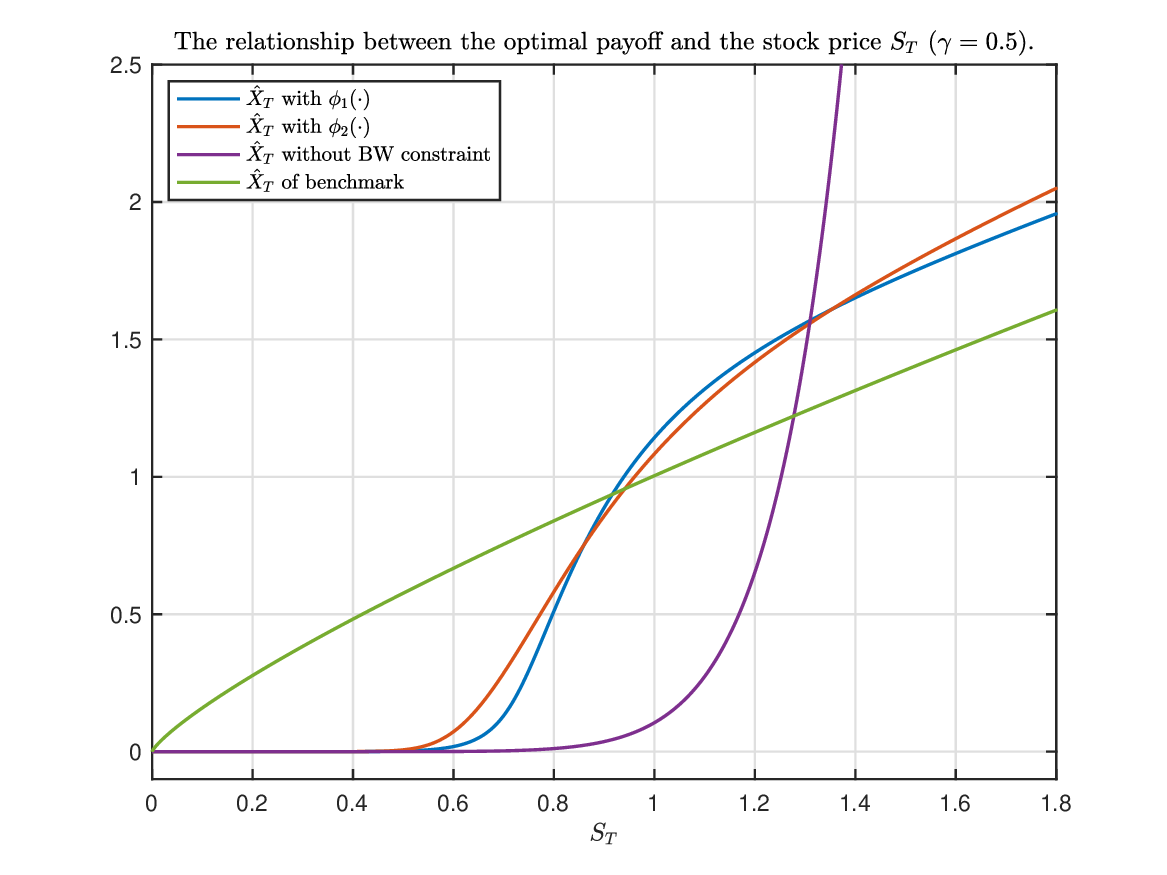}
    \end{subfigure}
    \hfill
    \begin{subfigure}[b]{0.49\textwidth}
        \centering
        \includegraphics[width=\textwidth]{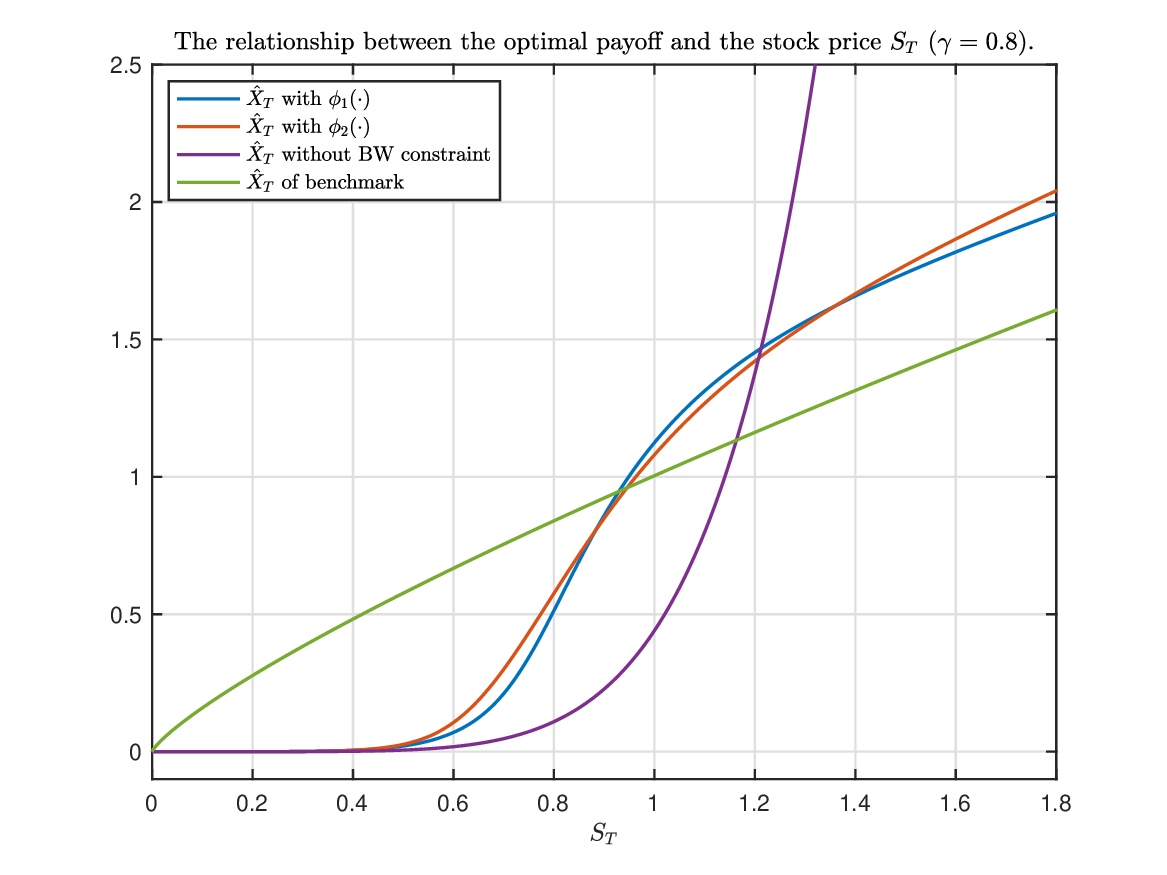}
    \end{subfigure}
	\caption{The optimal payoff as a function of $S_T$ when there are no BW constraints (purple curve) and when there is a BW constraint induced by the generator $\phi_1(x)=x^2$ (blue curve) or by the generator $\phi_2(x)=x \ln x$ (red curve). The non-constant benchmark is depicted by the green curve. The cases  $\gamma=0.5$ (left panel) and $\gamma=0.8$ (right panel) are studied.}
	\label{Test_ngamma_c}
\end{figure}

  \begin{figure}[htbp]
    \centering
    \begin{subfigure}[b]{0.49\textwidth}
        \centering
        \includegraphics[width=\textwidth]{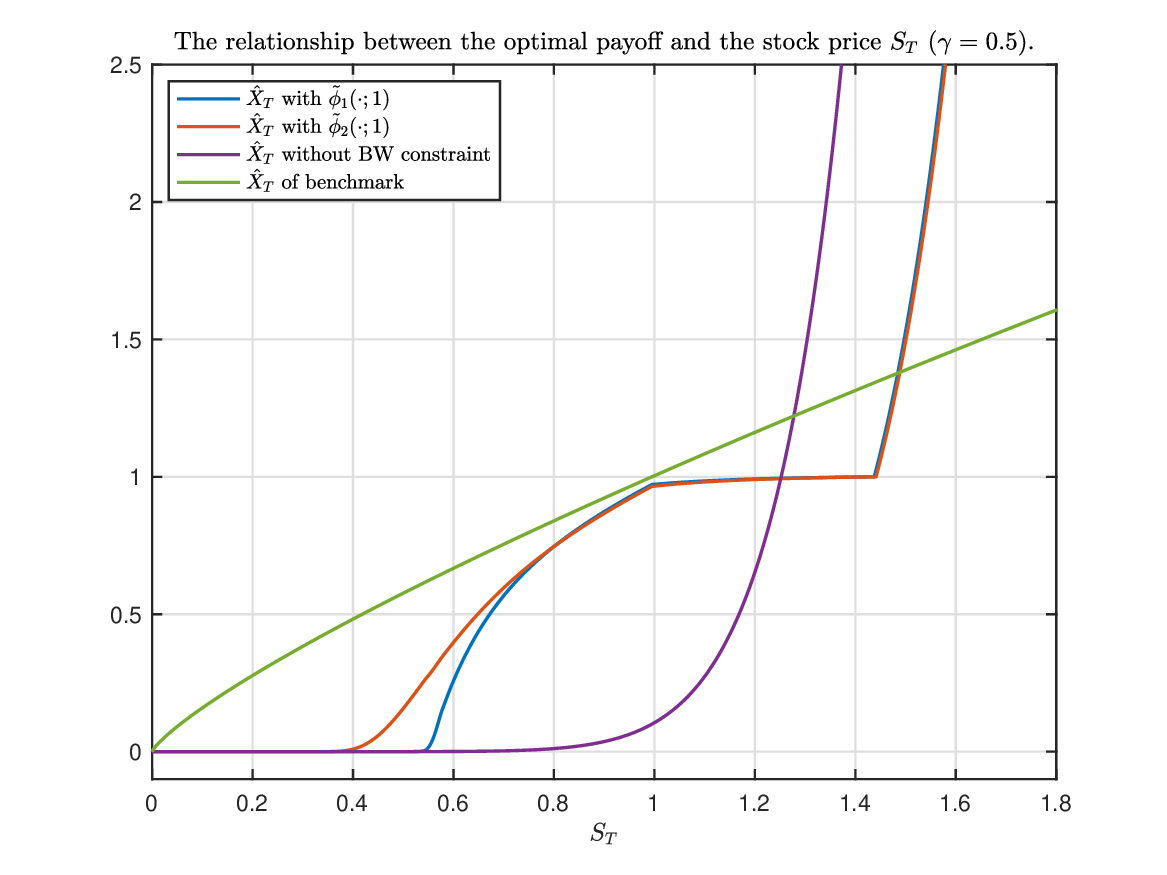}
    \end{subfigure}
    \hfill
    \begin{subfigure}[b]{0.49\textwidth}
        \centering
        \includegraphics[width=\textwidth]{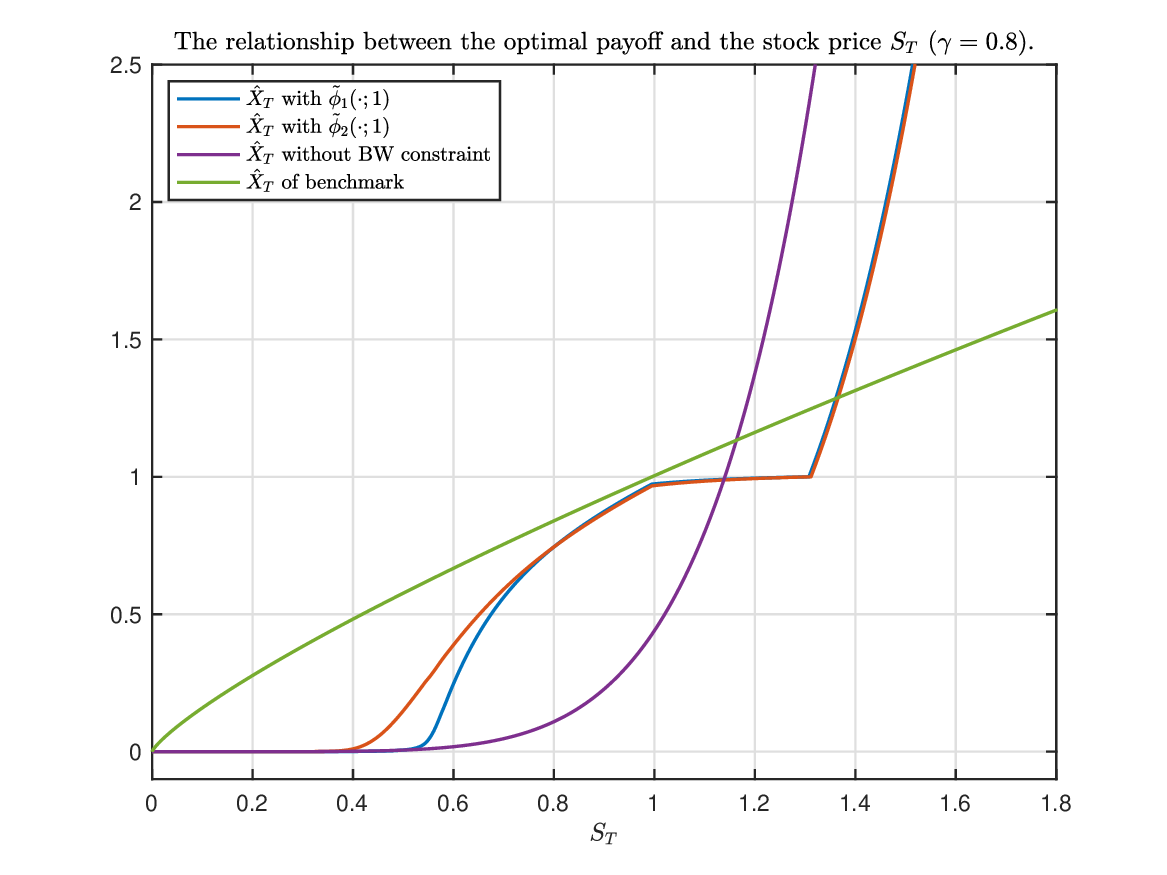}
    \end{subfigure}
    \begin{subfigure}[b]{0.49\textwidth}
        \centering
        \includegraphics[width=\textwidth]{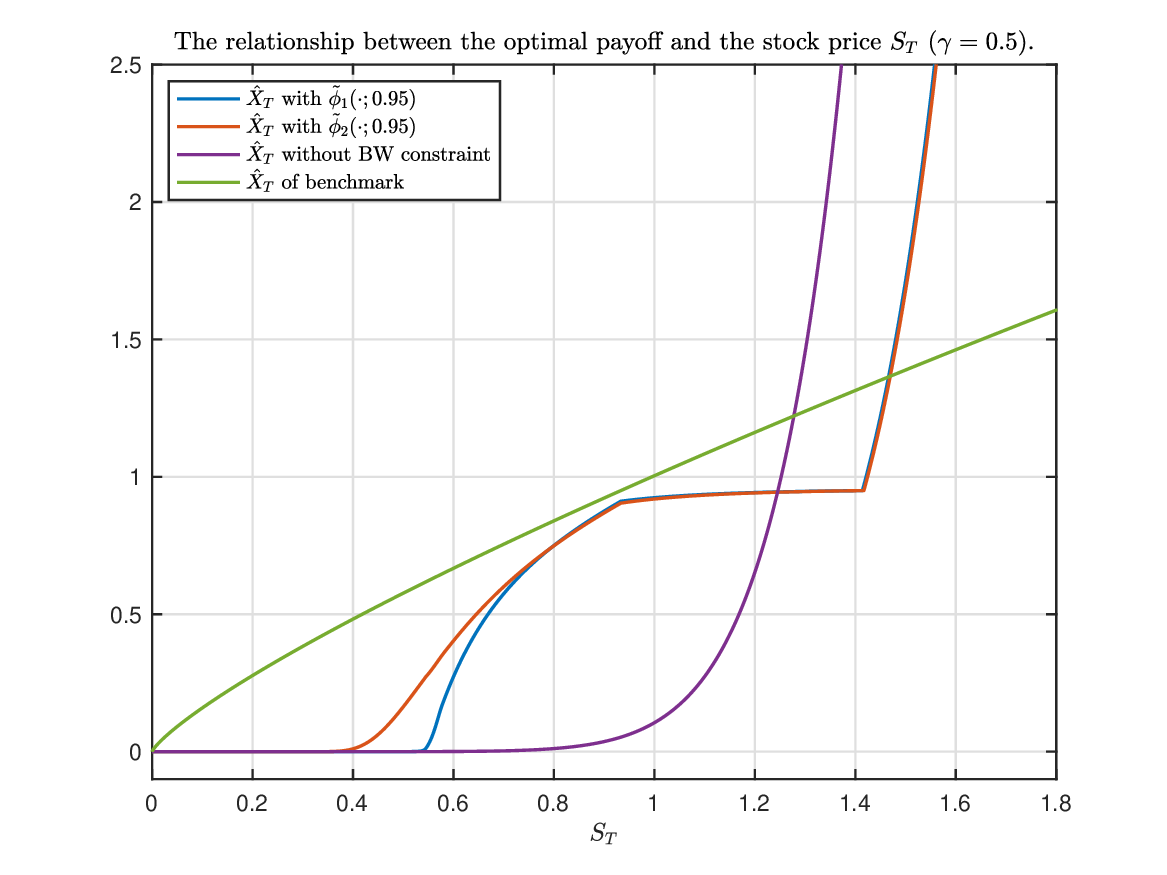}
    \end{subfigure}
    \hfill
    \begin{subfigure}[b]{0.49\textwidth}
        \centering        \includegraphics[width=\textwidth]{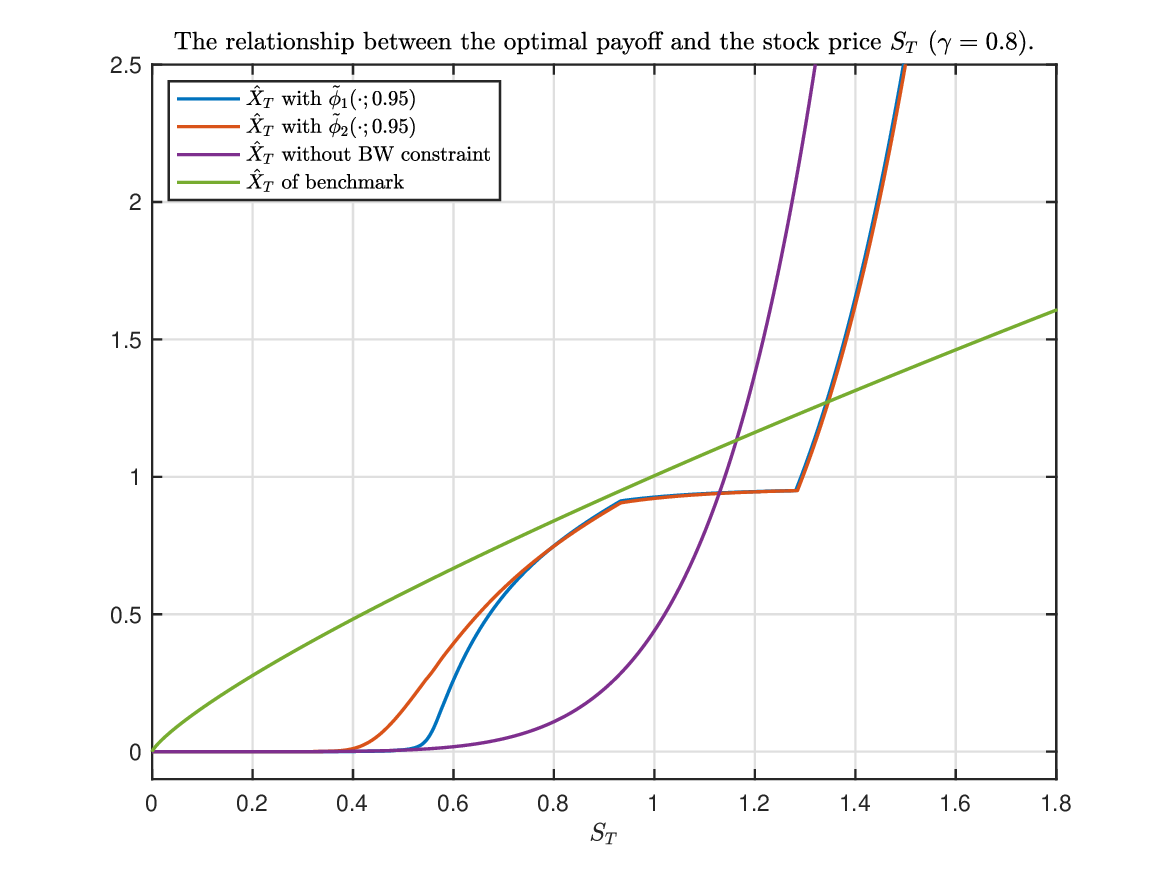}
    \end{subfigure}
	\caption{The optimal payoff as a function of $S_T$, when there are no BW constraints (purple curve) and when there is a BW constraint induced by the generator $\tilde{\phi}_1(x;\alpha)$ (blue curve) or by the generator $\tilde{\phi}_2(x;\alpha)$ (red curve). The non constant benchmark is depicted by the green curve. The cases  $\gamma=0.5$ (left panel) and $\gamma=0.8$ (right panel) are studied.}
	\label{Test_nimprove_gamma_c}
\end{figure}

\Cref{Test_ngamma_c} and \Cref{Test_nimprove_gamma_c} illustrate the relationship between the optimal payoffs $\hat{X}_T$ and the stock price $S_T$. We observe that, overall, the payoff of the benchmark increases in a linear fashion with $S_T$. This is due to the large allocation of stock within the benchmark. Furthermore, when the BW constraint is absent, $\varepsilon = \infty$, (purple curve) the optimal payoff is very different from the Benchmark. Adding the BW constraints brings, as expected, the optimal payoffs (blue and red curves) closer to the benchmark. When comparing the impact of the BW generators $\phi_i$, $i = 1,2$, (\Cref{Test_ngamma_c}) with the BW generators with threshold $\tilde{\phi_i}$, $i = 1,2$, (\Cref{Test_nimprove_gamma_c}), in the latter case the optimal payoffs stay closer to the benchmark when the underlying stock is declining and significantly outperform the benchmark when the stock is appreciating. 

In summary, the BW constraints allow the investor to create payoffs that yield large expected utility while staying close to the benchmark and in particular to obtain payoffs that show resilience when markets are declining.

\section{Concluding remarks}\label{sec:5}

In this paper, we deal with optimal payoff choice for an expected utility maximizer who aims to stay close to a benchmark, but considers asymmetry in that positive deviations are measured differently than negative ones. Doing so has significant impact on the payoff choice. By acknowledging and accounting for the asymmetric preferences of investors, financial professionals can design products and strategies that better align with their clients' preferences and risk tolerances.

\section*{Acknowledgement}
SP acknowledges support from the Natural Sciences and Engineering Research Council of Canada (grants DGECR-2020-00333 and RGPIN-2020-04289). SV is grateful to FWO for financial support (grant numbers FWO SBO S006721N and FWO WOG W001021N).
JY acknowledges the support from the National Natural Science Foundation of China (12371474).

\appendix
\section{Convex analysis in general vector spaces} \label{convex opti}
{ In this appendix, we collect auxiliary results on convex optimisation problems.

\begin{definition}\label{app_def}
    Let $X$ be a real linear (vector) space, and let $f:X \to \bar{\mathbb{R}}$.  The domain of $f$ is defined by
    \begin{align*}
     \operatorname{dom} f := \{ x \in X \ | \ f(x) < \infty \}.
    \end{align*} 
    The function $f$ is called proper if
    \begin{align*}
    \operatorname{dom} f\neq\emptyset \quad\text{and}\quad
    f(x)>-\infty,\quad \forall x\in X. 
    \end{align*}
    Moreover, $f$ is convex if  $\forall x,y \in X,  \forall \lambda \in [0,1]$
    \begin{align*}
    f(\lambda x + (1-\lambda)y) \leq \lambda f(x) + (1-\lambda) f(y).
    \end{align*}
\end{definition}

\begin{theorem}\label{main_optimal_th}
    Let $X$ be a separated locally convex space, and let $\Lambda(X)$ denote the class of proper convex functions on $X$. Let $f,g_{1},\ldots,g_{n} \in \Lambda(X)$ and consider the problem
    \begin{align} \label{constrained problem}
    \min \{ f(x) \ | \ g_i(x) \leq 0, \ i = 1,\ldots, n \}.  
    \end{align}
    Suppose that the Slater condition holds, namely, 
    \begin{align*}
     \exists x_0 \in \operatorname{dom} f, \ g_{i}(x_0) < 0, \ \forall i=1,\ldots,n.  
    \end{align*}
    Then, let $\bar{x} \in \operatorname{dom} f$; $\bar{x}$ is a solution of problem \eqref{constrained problem} if and only if $g_{i}(\bar{x}) \leq 0$ for every $i =1, \ldots,n$ and there exist $\bar{\lambda}_1,\ldots,\bar{\lambda}_n \in \mathbb{R}_{+}$ such that $\bar{\lambda}_{i}g_{i}(\bar{x})=0$ for $i=1,\ldots,n$ and 
    \begin{align*}
     0 \in \partial (f+\bar{\lambda}_1 g_1 + \ldots + \bar{\lambda}_n g_n )(\bar{x}).
    \end{align*}
    where $\partial f(x) $ denotes the Fenchel subdifferential of $f$ at $x$.
\end{theorem}
\begin{proof}
    See Theorem 2.9.3 in \cite{zualinescu2002convex}.
\end{proof}

\begin{theorem}\label{app_partial}
    If $f \in \Lambda(X)$, then $x \in \operatorname{dom} f$ is a minimum point for $f$ if and only if $0 \in \partial f(x)$.
\end{theorem}
\begin{proof}
    See Theorem 2.5.7 in \cite{zualinescu2002convex}.
\end{proof}}

\section{Proof of Proposition \ref{prop_MBW}}\label{app:proofs}
This appendix is devoted to additional proofs not provided in the main body of the paper.
\begin{proof}
We first prove that $\alpha \mapsto B_{\tilde{\phi}} \big(z_1, z_2; \alpha\big)$ is non-decreasing for fixed non-negative $z_1$ and $z_2$. First, let $\alpha_1 \leq z_1 \leq \alpha_2 \leq z_2 \leq \alpha_3$, then we prove that 
\begin{align*}
    B_{\tilde{\phi}} \big(z_1, z_2; \alpha_1 \big) \leq  B_{\tilde{\phi}} \big(z_1, z_2; \alpha_2 \big) \leq B_{\tilde{\phi}} \big(z_1, z_2; \alpha_3 \big). 
\end{align*}
From \eqref{MBG}, we have 
\begin{align*}
    B_{\tilde{\phi}} \big(z_1, z_2; \alpha_1 \big) =& \ 0, \\
    B_{\tilde{\phi}} \big(z_1, z_2; \alpha_2 \big) =& \ \phi(z_1) - \phi(\alpha_2) - \phi'(\alpha_2) (z_1 - \alpha_2), \\
    B_{\tilde{\phi}} \big(z_1, z_2; \alpha_3 \big) =& \ \phi(z_1) - \phi(z_2) - \phi'(z_2)(z_1-z_2) = B_{\phi} \big(z_1, z_2 \big).
\end{align*}
Note that only $B_{\tilde{\phi}} \big(z_1, z_2; \alpha_2 \big)$ depends on the threshold $\alpha$. Let $\alpha_{2,1},\alpha_{2,2} \in [z_1,z_2]$, with $\alpha_{2,1} \leq \alpha_{2,2}$. Then using in the second equation the Mean Value Theorem with $\xi \in [\alpha_{2,1}, \alpha_{2,2}]$, we obtain
\begin{align*}
    &B_{\tilde{\phi}} \big(z_1, z_2; \alpha_{2,2} \big)-B_{\tilde{\phi}} \big(z_1, z_2; \alpha_{2,1} \big) \\
    = & - \phi(\alpha_{2,2}) + \phi(\alpha_{2,1}) - \phi'(\alpha_{2,2}) (z_1 - \alpha_{2,2}) + \phi'(\alpha_{2,1}) (z_1 - \alpha_{2,1}) \\
    = & - \phi'(\xi) (\alpha_{2,2} - \alpha_{2,1}) - \left( \phi'(\alpha_{2,2}) - \phi'(\alpha_{2,1})   \right)z_1 +  \phi'(\alpha_{2,2}) \alpha_{2,2} -  \phi'(\alpha_{2,1}) \alpha_{2,1} \\
    = & \left( \phi'(\alpha_{2,2}) -  \phi'(\xi) \right)\alpha_{2,2} - \left( \phi'(\alpha_{2,2}) - \phi'(\alpha_{2,1})   \right)z_1 + \left( \phi'(\xi) - \phi'(\alpha_{2,1})  \right) \alpha_{2,1}\,. 
\end{align*}
Since $\phi'(\cdot)$ is non-decreasing, we have $\phi'(\alpha_{2,2}) -  \phi'(\xi) \geq \phi'(\alpha_{2,2}) - \phi'(\alpha_{2,1})$. Combing with $\alpha_{2,2} \geq z_1$ and that $\alpha_{2,1} \ge 0$, we get  
\begin{align}\label{pB}
    B_{\tilde{\phi}} \big(z_1, z_2; \alpha_{2,2} \big)-B_{\tilde{\phi}} \big(z_1, z_2; \alpha_{2,1} \big) \geq  \left( \phi'(\xi) - \phi'(\alpha_{2,1})  \right) \alpha_{2,1} \geq 0,
\end{align}
which implies that $\alpha_2 \mapsto B_{\tilde{\phi}} \big(z_1, z_2; \alpha_2 \big)$ is non-decreasing on $[z_1, z_2]$. Thus, we obtain 
\begin{align}\label{neq}
     B_{\tilde{\phi}} \big(z_1, z_2; \alpha_1 \big) = B_{\tilde{\phi}} \big(z_1, z_2; z_1 \big) \leq   B_{\tilde{\phi}} \big(z_1, z_2; \alpha_2 \big) \leq B_{\tilde{\phi}} \big(z_1, z_2; z_2 \big) = B_{\tilde{\phi}} \big(z_1, z_2; \alpha_3 \big).
\end{align}
Similarly, when $z_1 > z_2$, we can prove that the inequalities \eqref{pB} and \eqref{neq} hold, i.e., $\alpha \mapsto B_{\tilde{\phi}} \big(z_1, z_2; \alpha\big)$ is non-decreasing. Thus, we conclude that $\alpha \mapsto \DB_{\tilde{\phi}} (\Finv_1, \Finv_2; \alpha)$ is non-decreasing. Finally, taking the limit in $\alpha$, we obtain \eqref{lima}.
\end{proof}

% \section*{Acknowledgments}

% \bibliographystyle{siamplain}
\bibliographystyle{apalike}
\bibliography{references}

\end{document}